\documentclass[%
reprint,
superscriptaddress,
nofootinbib,
amsmath,amssymb,
aps,
]{revtex4-2}

\usepackage{graphicx}%
\usepackage{dcolumn}%
\usepackage{bm}%

\usepackage{hyperref}
\hypersetup{
    colorlinks = true,
    linkcolor = blue
    }

\usepackage{braket}
\usepackage{bbm}
\usepackage{comment}
\usepackage{xcolor}
\usepackage{mathtools}
\usepackage{physics}
\usepackage{float}
\usepackage{qcircuit}
\usepackage{subcaption}
\usepackage[colorinlistoftodos, textsize=small, color=orange!50]{todonotes}
\usepackage[margin=0.75in]{geometry}
\usepackage[normalem]{ulem}
\usepackage{booktabs}
\usetikzlibrary{calc}

\usepackage[ruled]{algorithm2e}
\SetKwFor{UnaryFor}{multiplexing over register}{do}{end}
\SetKw{MiddleFor}{, for}

\DeclarePairedDelimiterX{\inp}[2]{\langle}{\rangle}{#1, #2}

\newcommand{\innerproductcomma}[2]{\left\langle #1, #2 \right\rangle}

\usepackage{amsthm}
\usepackage{thmtools}
\usepackage{thm-restate}
\usepackage{apptools} %

\usepackage[capitalise]{cleveref}

\crefname{algocf}{alg.}{algs.}
\Crefname{algocf}{Algorithm}{Algorithms}

\newtheorem{theorem}{Theorem} 

\newenvironment{theorem*}
 {\expandafter\def\expandafter\thetheorem\expandafter{\thetheorem*}\theorem}
 {\endtheorem}
\newtheorem{lemma}[theorem]{Lemma}
\newtheorem{proposition}[theorem]{Proposition}

\newtheorem{corollary}[theorem]{Corollary}

\usepackage[english]{babel}

\makeatletter
\def\bbl@set@language#1{%
  \edef\languagename{%
    \ifnum\escapechar=\expandafter`\string#1\@empty
    \else\string#1\@empty\fi}%
  \@ifundefined{babel@language@alias@\languagename}{}{%
    \edef\languagename{\@nameuse{babel@language@alias@\languagename}}%
  }%
  \select@language{\languagename}%
  \expandafter\ifx\csname date\languagename\endcsname\relax\else
    \if@filesw
      \protected@write\@auxout{}{\string\select@language{\languagename}}%
      \bbl@for\bbl@tempa\BabelContentsFiles{%
        \addtocontents{\bbl@tempa}{\xstring\select@language{\languagename}}}%
      \bbl@usehooks{write}{}%
    \fi
  \fi}
\newcommand{\DeclareLanguageAlias}[2]{%
  \global\@namedef{babel@language@alias@#1}{#2}%
}
\makeatother

\DeclareLanguageAlias{en}{english}
\usepackage{dsfont}

\usepackage[
]{geometry}
\newcommand{\bill}[1]{}
\newcommand{\torin}[1]{}
\newcommand{\oskar}[1]{}
\newcommand{\TODO}[1]{}
\newcommand{\LATER}[1]{}

\begin{document}

\preprint{APS/123-QED}

\title{Efficient state preparation for the quantum simulation of molecules in first quantization}%

\author{William J. Huggins}
\email{whuggins@google.com}
\thanks{These two authors have contributed equally to this work.}
\affiliation{Google Quantum AI, San Francisco, CA 90291, United States}

\author{Oskar Leimkuhler}
\email{ol22@berkeley.edu}
\thanks{These two authors have contributed equally to this work.}
\affiliation{Department of Chemistry, University of California, Berkeley, CA 94720}
\affiliation{Berkeley Quantum Information and Computation Center,
University of California, Berkeley, CA, 94720, United States}
\affiliation{Challenge Institute for Quantum Computation, University of California, Berkeley, CA 94720}

\author{Torin F. Stetina}
\affiliation{Simons Institute for the Theory of Computing, Berkeley, CA, 94704, United States}
\affiliation{Berkeley Quantum Information and Computation Center,
University of California, Berkeley, CA, 94720, United States}

\author{K. Birgitta Whaley}
\affiliation{Department of Chemistry, University of California, Berkeley, CA 94720}
\affiliation{Berkeley Quantum Information and Computation Center,
University of California, Berkeley, CA, 94720, United States}
\affiliation{Challenge Institute for Quantum Computation, University of California, Berkeley, CA 94720}

\date{\today}%
\begin{abstract}
The quantum simulation of real molecules and materials is one of the most highly anticipated applications of quantum computing.
Algorithms for simulating electronic structure using a first-quantized plane wave representation are especially promising due to their asymptotic efficiency.
However, previous proposals for preparing initial states for these simulation algorithms scale poorly with the size of the basis set.
We address this shortcoming by showing how to efficiently map states defined in a Gaussian type orbital basis to a plane wave basis with a scaling that is logarithmic in the number of plane waves.
Our key technical result is a proof that molecular orbitals constructed from Gaussian type basis functions can be compactly represented in a plane wave basis using matrix product states.
While we expect that other approaches could achieve the same logarithmic scaling with respect to basis set size, our proposed state preparation technique is also highly efficient in practice.
For example, in a series of numerical experiments on small molecules, we find that our approach allows us to prepare an approximation to the Hartree-Fock state using orders of magnitude fewer non-Clifford gates than a naive approach.
By resolving the issue of state preparation, our work allows for the first quantum simulation of molecular systems whose end-to-end complexity is truly sublinear in the basis set size.

\end{abstract}
\maketitle

\section{Introduction}
The electronic structure problem in quantum chemistry is considered one of the most promising applications of fault tolerant quantum computers.
Quantum computing promises to enable the direct simulation of electronic dynamics from first principles, a feat believed to be beyond the capabilities of classical algorithms due to the combinatorial explosion of the state space.
The quantum algorithms with the best known asymptotic scaling with respect to basis set size rely on simulating the electronic structure Hamiltonian using highly structured basis sets in a first-quantized representation~\cite{Babbush2019-xw,Su2021-uj,Babbush2023-ud}.
It has even been argued that dynamical simulations using these methods can outperform classical simulations based on mean-field theory~\cite{Babbush2023-ud}.

The cost of initial state preparation is frequently neglected in the analysis of quantum algorithms for quantum chemistry.
However, for the first-quantized algorithms with the best asymptotic scaling, prior work has focused on using approaches to initial state preparation that scale asymptotically worse (with the basis set size) than any other component of the algorithm~\cite{Su2021-uj}.
This is true even when preparing a relatively simple initial state, such as a mean-field approximation of the ground state.
For simulations of molecular systems, the most efficient first quantized algorithms introduce unphysical periodic boundary conditions whose effects can be suppressed by using a large enough computational unit cell.
It is especially dangerous to neglect the cost of state preparation when increasing the size of the computational unit cell (while keeping the molecular system and the kinetic energy cutoff fixed).
In this situation, every aspect of these first-quantized simulation algorithms actually scales logarithmically with the basis set size except for the initial state preparation, which scales linearly.

To improve the efficiency of state preparation for first-quantized plane wave simulations, one possible solution is to adapt the methods originally proposed by Zalka in Ref.~\citenum{Zalka1998-zt} and later explored in the context of real-space simulations of many-body fermionic systems by Ward \textit{et al.} in Ref.~\citenum{Ward2009-ln}.
These approaches, which are closely related to the well-known Grover-Rudolph approach for preparing quantum states that encode probability distributions~\cite{Grover2002-st}, rely on being able to construct quantum circuits that evaluate definite integrals of the desired wavefunction.
By querying these circuits in superposition, one can efficiently perform a series of rotations that prepare the target state qubit by qubit.
One could imagine adapting these approaches to prepare a variety of initial states in a first-quantized plane wave setting, leveraging the fact that classical calculations in a Gaussian basis set yield descriptions of single-particle states that are linear combinations of efficiently integrable functions.
However, in practice, coherently performing the required integrals could be prohibitively expensive.

In this work, we present an alternative approach that makes use of tensor train (TT) or matrix product state (MPS) factorizations to compactly represent single-particle states originating from calculations in a Gaussian basis set.
This allows for the preparation of mean-field states as well as more sophisticated correlated initial states defined in a traditional Gaussian orbital basis set.
For approximating an \(\eta\) particle Slater determinant calculated in a Gaussian orbital basis with \(N_g\) primitive Gaussian basis functions and projected onto a basis of \(N\) plane waves, our approach has a (non-Clifford) gate complexity that scales as \(\mathcal{O}\left(\eta^2 N_g^{3/2} \log(N)\right)\).
By contrast, the linear scaling approach of Ref.~\citenum{Babbush2023-ud} has a cost that scales as \(\mathcal{O}(\eta N)\).
Furthermore, the efficiency of these low-rank representations allows us to achieve the desired asymptotic scaling with relatively low prefactors, reducing the cost of state preparation by orders of magnitude when compared with prior work.
In fact, our numerics suggest that the factor of \(N_g^{3/2}\) in our analytical bounds may be overly pessimistic, and that the cost may be nearly independent of \(N_g\) in practice.
Like other recent work that makes use of these factorizations for state preparation~\cite{Garcia-Ripoll2019-jk, Holmes2020-yx, Lubasch2020-hw, Plekhanov2022-md, Melnikov2023-ny, Iaconis2024-ny, Gonzalez-Conde2023-gw, Cervero-Martin2023-gu}, we benefit from the growing body of literature that studies these same tools in the context of classical computing~\cite{Grasedyck2010-pn,Oseledets2010-hb,Oseledets2011-xm,Oseledets2013-vy,Grasedyck2013-ib,Stoudenmire2016-np}.
In particular, we leverage the results of Ref.~\citenum{Grasedyck2010-pn}, which prove that vectors arising from evaluating low-degree polynomials on a uniform grid have compact representations as tensor trains.

We organize the main body of our manuscript as follows.
We begin by reviewing relevant background material and prior work in \Cref{sec:preliminaries}.
We present our analytical results in \Cref{sec:state_prep}, culminating in \Cref{sec:overall_cost}, where we discuss the overall asymptotic scaling of our approach.
In \Cref{sec:MO_MPS_numerics}, we provide numerical evidence that our technique for state preparation can be orders of magnitude less costly than a naive approach, even for small molecular systems.
We conclude with a brief discussion in \Cref{sec:discussion}.

\section{Preliminaries}
\label{sec:preliminaries}

In this section, we briefly review a few key concepts.
We first discuss the two single-particle basis sets most relevant to our work in \Cref{sec:basis_sets}, and then present some background material on first-quantized quantum simulation in \Cref{sec:first_quantized_preliminaries}.
In \Cref{sec:MPS_and_TT}, we review the definitions of tensor train (TT) and matrix product state (MPS) factorizations along with some of their key properties.
We clarify some notation used throughout the manuscript in \Cref{sec:notation}.

\subsection{Single-particle basis sets}
\label{sec:basis_sets}

Electronic structure wavefunctions in \textit{ab initio} quantum chemistry are typically represented in a finite basis set of single-particle orbitals.
Various considerations, such as computational efficiency, representational power, and symmetry motivate the choice of single-particle basis set.
This work uses two different families of single-particle basis functions: atom-centered Gaussian type orbitals (GTOs) and plane waves.
While many quantum chemical calculations are performed in one of these two basis sets, there is a rich body of research exploring alternative approaches that we do not attempt to review here~\cite{Hehre1969-uz, Gilbert1974-fh, Heinemann1988-ap, White1989-qx, Frediani2015-ea, White2019-zu, McClean2020-sg}.

Gaussian type orbitals consist of linear combinations of primitive basis functions translated so that they are centered on atomic nuclei.
Expressed in Cartesian coordinates, these primitive basis functions take the form
\begin{equation}
  g(x, y, z) = x^{l} y^{m} z^{n} e^{-\gamma \left( x^2 + y^2 + z^2 \right)},
\end{equation}
where \(\gamma\) is the Gaussian broadening parameters, and \(l\), \(m\), and \(n\) are the angular momentum quantum numbers~\cite{Helgaker2014-dh}.
Gaussian type orbitals are typically numerically optimized to approximate solutions to the Schrödinger equation for hydrogenic atoms.
As a result, they provide a physically motivated basis set for calculations of ground states and chemically relevant excited states of molecular systems under the Born-Oppenheimer approximation of clamped nuclei.
Integrals involving Gaussian type orbitals have an analytical form, leading to relatively efficient implementations compared to explicit Slater functions.
Furthermore, relatively few Gaussian type orbitals are usually sufficient to obtain good quantitative accuracy for low-energy states~\cite{Hill2013-un}.
This makes them the basis set of choice for many classical methods~\cite{Helgaker2014-dh}.

While Gaussian basis sets form a reasonable and frequently-used discretization for the electronic structure problem, they are not without their drawbacks.
At long ranges, the electron density of any wavefunction expressed in a Gaussian basis set will necessarily exhibit a Gaussian decay rather than the expected exponential decay~\cite{Kanungo2019-ag}.
Additionally, Gaussian basis sets are incapable of correctly representing electron-nucleus cusps~\cite{Regier1985-as, Kanungo2019-ag}, which sometimes motivates the use of the less computationally efficient Slater-type orbitals~\cite{Regier1985-as}.
Beyond these concerns, the compactness of Gaussian basis sets also comes at a cost: primitive Gaussian basis functions are not orthogonal in general.

In many contexts, it is simpler to work with an orthonormal basis than to work directly with Gaussian type orbitals.
One standard approach is to calculate a set of ``molecular orbitals'' (MOs) using the Hartee-Fock method.
This approach self-consistently determines a single-particle basis where the mean field approximation to the Hamiltonian is diagonal.
The resulting single-particle basis functions, the molecular orbitals, are orthonormal by construction.
The Hartree-Fock wavefunction is the $N$-orbital Slater determinant where the \(\eta\) molecular orbitals with the lowest eigenvalues are occupied by electrons.

In this work, we also make heavy use of single-particle basis sets composed of plane waves.
Plane waves are eigenfunctions of the momentum operator, and we can express them in Cartesian coordinates as 
\begin{equation}
  \varphi_{\bm{k}}\left(\bm{r}\right) = \sqrt{\frac{1}{L^3}} e^{i \, \bm{k} \cdot \bm{r}},
\end{equation}
where \(\bm{r}\) denotes the coordinates in real space and \(\bm{k}\) denotes the momentum.
We specifically consider plane wave basis sets defined on a cubic reciprocal lattice.
This means that the momenta satisfy \(\bm{k} \in \mathbb{K}^3\), with \(\mathbb{K}\) defined by
\begin{equation}
  \mathbb{K} = \left\{ \frac{2 \pi p}{L} : p \in \mathbb{Z}\right\}.
  \label{eq:K_def}
\end{equation}
These plane waves have the property that they are invariant under translations by a distance \(L\) in the \(x, y,\) or \(z\) directions, and they are especially useful for treating periodic systems that are defined by tiling the cube \(\left[-L/2, L/2\right]^{3}\).
We refer to this cube as the ``computational unit cell.''

To restrict ourselves to a finite basis, we take a cubic grid of \(N\) plane waves centered at the origin in momentum space.
More precisely, we consider only the plane waves with momenta \(\bm{k} \in \mathbb{K}_{cut}^3\), with
\begin{equation}
  \mathbb{K}_{cut} = \left\{ \frac{2 \pi p}{L} : p \in \left[-\frac{N^{1/3} - 1}{2}, \frac{N^{1/3} - 1}{2} \right] \subset \mathbb{Z} \right\}.
\end{equation}
In this work, we will usually find it more convenient to specify \(\mathbb{K}_{cut}\) in terms of a momentum cutoff \(K \in \mathbb{R}_{> 0}\),
\begin{equation}
  \mathbb{K}_{cut} = \left\{ k \in \mathbb{K} : \abs{k} \leq K\right\}.
  \label{eq:K_cut_def}
\end{equation}
Note that fixing an \(L\) and a \(K\) implicitly determines the number of single particle basis functions (\(N\)), with
\begin{equation}
  N = \left|\mathbb{K}_{cut}\right|^3 = \left(2 \left\lfloor \frac{KL}{2 \pi} \right\rfloor + 1\right)^3.
  \label{eq:N_implicit_from_KL}
\end{equation}
Qualitatively, we can interpret the momentum cutoff as defining the spatial resolution of the computational unit cell in real space.

Plane wave basis sets offer distinct advantages, but they are rarely used when treating molecular systems with classical computational approaches.
One key benefit is the fact that plane waves are inherently orthonormal. 
They also exhibit minimal bias in both real and momentum space, unlike atom-centered Gaussian basis sets, which are designed to represent specific types of states. 
This makes plane wave basis sets particularly useful for dynamical simulations, where it might be inappropriate to assume a priori that the electron density is localized around atomic sites.
Furthermore, plane wave basis sets are periodic, which makes them well-suited for simulating periodic systems such as crystalline materials. 

This periodicity complicates the treatment of isolated molecular systems, which either require the use of a large computational unit cell (to suppress the interaction of the molecular system with its own periodic images) or some other modification to the simulation. 
Furthermore, calculations in a plane wave basis set typically require significantly larger basis set sizes than analogous calculations using a Gaussian basis set~\cite{Andrews1996-ze, Fusti-Molnar2002-ch, Booth2016-rs, Babbush2018-zo}.
This is partially due to the very high momentum cutoff required to approximate sharp features in the wavefunction.
In many plane wave calculations, the sharp features due to electron-nucleus cusps are treated using pseudopotentials in order to reduce the required momentum cutoff~\cite{Hamann1979-gq, Bachelet1982-xg, Schwerdtfeger2011-jr}.

\subsection{First quantized quantum simulations}
\label{sec:first_quantized_preliminaries}

First quantized approaches to quantum mechanics long predate the development of quantum computing, and some of the earliest proposed quantum algorithms for quantum simulation make use of this formalism~\cite{Wiesner1996-mc,Abrams1997-nc, Zalka1998-zt,Boghosian1998-ow,Lidar1999-er,Kassal2008-jb}.
Recent works have revived interest in this approach by demonstrating that modern algorithmic techniques enable extremely efficient first quantized simulations of the electronic structure Hamiltonian~\cite{Babbush2019-xw, Su2021-uj, Babbush2023-ud}.
Building on the interaction picture techniques of Ref.~\citenum{Low2018-zx}, Ref.~\citenum{Babbush2019-xw} showed that it is possible to implement time evolution by the electronic structure Hamiltonian (in a plane wave basis set) with a gate complexity that scales as \(\widetilde{\mathcal{O}}(\eta^{8/3} N^{1/3} L^{-1} t)\), where \(\eta\) represents the number of electrons, \(N\) the number of plane waves, \(L\) the linear size of the cubic computational unit cell, and \(t\) the total time.
Ref.~\citenum{Su2021-uj} also explained how a similar scaling can be achieved using a real-space grid basis, although such basis sets have some technical shortcomings, so we do not consider them in this work.\footnote{Using a real-space grid basis necessitates making an approximation to the kinetic energy operator that makes the simulation non-variational. Specifically, it is possible for the lowest eigenvalue in a grid basis to be below the lowest eigenvalue of the actual electronic structure Hamiltonian operator.
}

In first-quantized simulation approaches, each of the \(\eta\) particles in the simulation is assigned a dedicated register of qubits that encodes that particle's state in a finite orthonormal basis of size \(N\).
Unlike in second quantization, first-quantized simulations do not enforce the correct particle statistics by construction.
Instead, fermionic (or bosonic) statistics are accounted for by demanding that the initial wavefunction is antisymmetric (or symmetric) with respect to the exchange of two registers representing identical particles~\cite{Abrams1997-nc, Berry2018-ey}.
This representation is especially convenient when \(N \gg \eta\), as is the case with a plane wave basis, since the overall space complexity is \(\mathcal{O}(\eta \log N)\).

While any choice of single particle basis is compatible with first quantization in principle, the efficient \(N^{1/3}\) scaling obtained by Ref.~\citenum{Babbush2019-xw} and subsequent works is due to the highly structured nature of a plane wave basis set.
In fact, this scaling can be misleading, since there are multiple distinct ways of scaling the basis set size that should be discussed separately.
In one limit, we keep \(L\) (the linear size of the computational unit cell) fixed as we increase \(N\).
This corresponds to an increase in the kinetic energy cutoff.
In other words, we increase the maximum frequency of the plane waves that we include in the simulation, effectively increasing the resolution of our simulation in real space by allowing us to represent features on a smaller length scale.
This kind of scaling is necessary, for example, to represent electrons in highly-localized core orbitals.

Alternatively, we can increase \(N\) and \(L^3\) proportionately while fixing the maximum frequency.
When we scale in this fashion, the size of the computational unit cell in real space increases while the kinetic energy cutoff (and the minimum length scale we can accurately represent) stays the same.
This kind of scaling is necessary when we wish to suppress the interaction of a molecular system with its artificial periodic images by increasing the size of the computational unit cell.
In this limit, the complexity of performing real-time evolution, phase estimation, or similar simulation tasks with the approach of Ref.~\citenum{Babbush2019-xw} and Ref.~\citenum{Su2021-uj} actually scales logarithmically in \(N\) (neglecting the cost of initial state preparation).

The best known algorithms for preparing even a single arbitrary Slater determinant (mean-field wavefunction) scale as \(\tilde{\mathcal{O}}(\eta N)\)~\cite{Babbush2023-ud}.
Contrast this with the cost of, e.g., time evolution, which scales with the basis set and computational unit cell size as \(N^{1/3} L^{-1}\) (up to logarithmic factors).
Prior approaches to initial state preparation have a cost that scales cubically worse than the other components of the simulation when we consider increasing \(N\) while keeping \(L\) fixed and exponentially worse when we increase \(N\) and \(L^3\) in proportion.
This paper addresses the scaling of state preparation in both limits, ultimately finding that we can achieve a cost that is polylogarithmic with the basis set size in both cases.

In order to achieve a scaling that is better than the number of free parameters required to describe an arbitrary Slater determinant (\(\mathcal{O}(\eta N)\)), we restrict ourselves to preparing Slater determinants that are obtained from calculations in a smaller Gaussian-type orbital basis.
Under a similar restriction, it should be possible to achieve the same logarithmic scaling that our approach does using an extension of the state preparation approach proposed by Zalka in Ref.~\citenum{Zalka1998-zt}.
However, as we alluded to in the introduction of this manuscript, such an approach would likely require the practically infeasible coherent calculation of a large number of definite integrals. 

\subsection{Matrix product states and tensor trains}
\label{sec:MPS_and_TT}

The main results of this paper make use of matrix product states (MPS) or tensor-train (TT) factorizations~\cite{Ostlund1995-gk, Vidal2003-dq, Grasedyck2010-pn, Oseledets2011-xm,Oseledets2013-vy}, which we briefly review here.
Consider a rank-\(n\) tensor \(T^{s_1 s_2 \cdots s_n}\).
We call a collection of \(n\) rank-\(2\) and rank-\(3\) tensors \(\left\{ A_{\alpha_1}^{s_1} A_{\alpha_1 \alpha_2}^{s_2} \cdots A_{\alpha_{n-1}}^{s_n} \right\}\), such that
\begin{equation}
  T^{s_1 s_2 \cdots s_n} = \sum_{\left\{ \alpha \right\}} A_{\alpha_1}^{s_1} A_{\alpha_1 \alpha_2}^{s_2} \cdots A_{\alpha_{n-1}}^{s_n}
\end{equation}
a tensor train factorization of \(T\).
Consider an \(n\)-qubit quantum state \(\ket{\psi}\) whose coefficients in the computational basis are given by \(T\),
\begin{align}
  \ket{\psi} &=
  \sum_{\bm{s} \in \left\{ 0, 1 \right\}^n} T^{s_1 s_2 \cdots s_n} \ket{s_1 s_2 \cdots s_n}
  \nonumber \\
  &=
  \sum_{\bm{s} \in \left\{ 0, 1 \right\}^n,\;  \left\{ \alpha \right\}} A_{\alpha_1}^{s_1} A_{\alpha_1 \alpha_2}^{s_2} \cdots A_{\alpha_{n-1}}^{s_n} \ket{s_1 s_2 \cdots s_n}.
\end{align}
In this paper, we use the term ``matrix product state'' to refer to quantum states whose coefficients are given by a tensor train factorization.
For both tensor trains and matrix product states, we refer to the indices \(\alpha_1, \alpha_2, \cdots \alpha_{n-1}\) as ``virtual indices'' and we refer to the dimensions of these indices, \(m_1, m_2, \cdots m_{n-1}\), as their ``bond dimensions.''
Sometimes we refer to the bond dimension of a particular MPS or TT with a single number, \(m = \max(m_1, m_2, \cdots m_{n-1})\).

Matrix product states have the convenient feature that we can easily convert them into efficient circuits for state preparation~\cite{Perez-Garcia2006-oe, Fomichev2023-vs}.
Without loss of generality, we can efficiently rewrite a matrix product state into a left-canonical form using a series of singular value decompositions and tensor contractions.
In left-canonical form, the tensors satisfy the following relations,
\begin{align}
  \sum_{s_1, \alpha_1} A^{s_1}_{\alpha_1} A^{s_1 *}_{\alpha_1}
   & =
  1,
  \nonumber
  \\
  \sum_{s_j, \alpha_{j}} A^{s_j}_{\alpha_{j-1}, \alpha_j} A^{s_j *}_{\alpha'_{j-1}, \alpha_j}
   & =
  \delta_{\alpha_{j-1}, \alpha'_{j-1}} \;\; \text{for} \;\; 1 < j < n
  \nonumber
  \\
  \sum_{s_n} A^{s_n}_{\alpha_{n-1}} A^{s_n *}_{\alpha'_{n-1}}
   & =
  \delta_{\alpha_{n-1}, \alpha'_{n-1}},
\end{align}
where \(\delta\) denotes the Kronecker delta function.
In other words, we can think of the first tensor as a normalized state of dimension \(2 m_1\), the middle tensors as isometries of dimension \(m_{j-1} \times 2 m_j\), and the final tensor as an isometry of dimension \(m_{n-1} \times 2\).
Standard techniques can then be used to implement the state preparation step and the isometries~\cite{Low2018-uu}.
We illustrate the resulting state preparation circuit using a hybrid quantum circuit/tensor network diagram in \Cref{fig:mps_diagram}.

\begin{figure}
  \includegraphics[width=.48\textwidth]{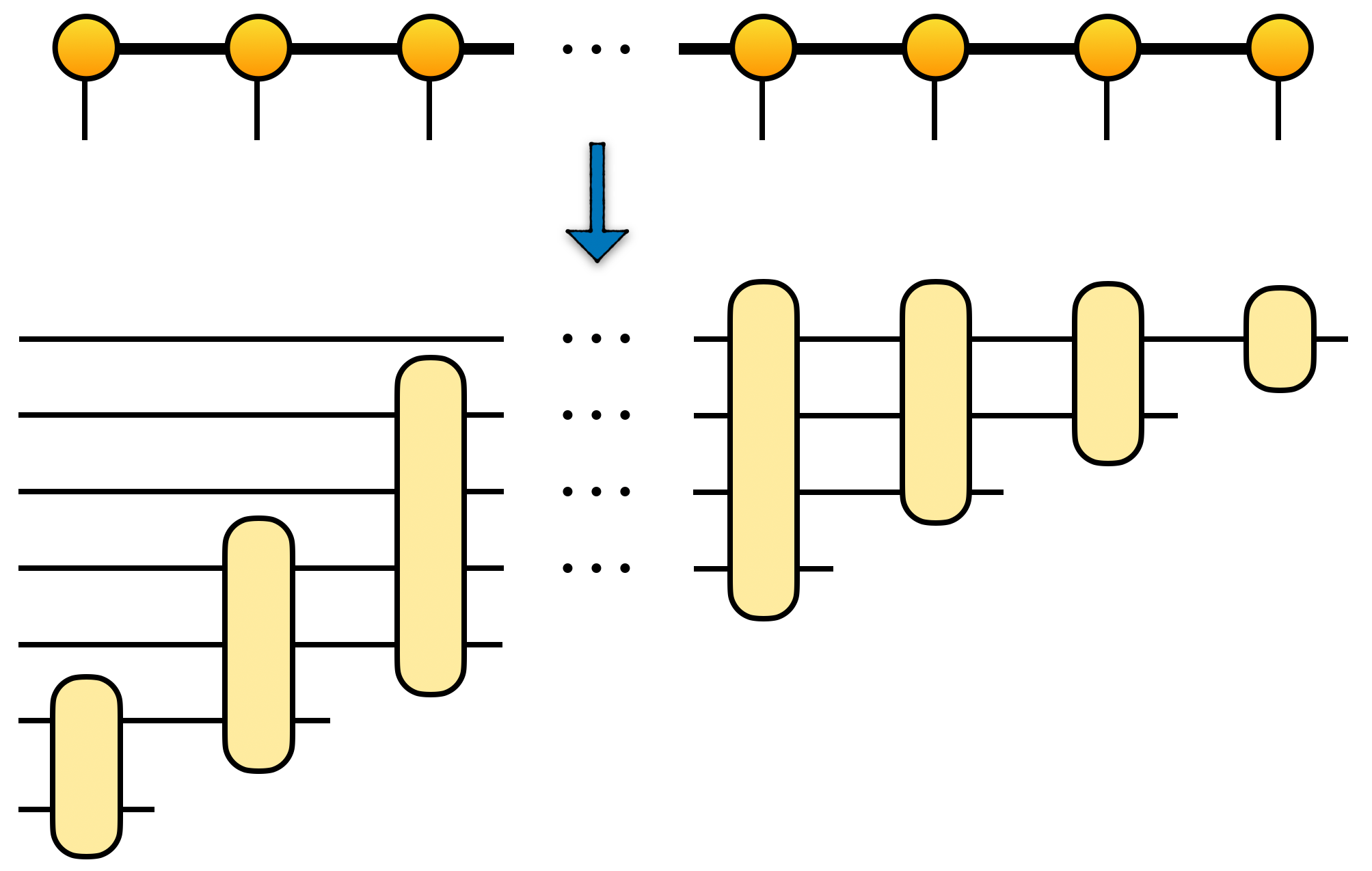}
  \caption{An illustration of the relationship between a matrix product state and the quantum circuit for preparing that state. The diagram above is a tensor network diagram for the matrix product state, where each circle corresponds to a tensor. Each line segment indicates an index of the tensors it touches. The segments connecting two tensors represent indices that are summed over. Each tensor in the matrix product state corresponds to a few-qubit unitary in the quantum circuit diagram below. When acting on the \(\ket{0}\) state, the circuit below prepares the matrix product state.}
  \label{fig:mps_diagram}
\end{figure}

The first-quantized quantum simulation algorithms that motivate this paper are usually considered in the context of a fault-tolerant quantum computer, capable of running large computations without error.
We usually quantify the cost of a quantum circuit in this setting by counting the number of non-Clifford gates, usually T or Toffoli gates, required for its implementation.
Building on Ref.~\citenum{Low2018-uu}, Ref.~\citenum{Fomichev2023-vs} shows that it is possible to prepare a bond dimension \(m\) matrix product state \(\ket{\psi}\) on \(n\) qubits using
\begin{equation}
  \textsc{Toffoli}\left( \ket{\psi} \right) = \widetilde{\mathcal{O}}(n m^{3/2})
\end{equation}
Toffoli gates, provided that one is willing to use \(\widetilde{\mathcal{O}}(m^{1/2})\) extra space to minimize the Toffoli cost.\footnote{Note that this extra space can be provided in the form of ``dirty ancilla'' qubits, i.e., qubits in an unknown state that are temporarily borrowed and then returned to their original state.}
Because Toffoli gates are so much more costly to implement than Clifford gates in the context of quantum error correction, this is a beneficial optimization even though the overall number of gates will still scale as \(n m^2\).
Although the circuit depth scales linearly with \(n\) using this approach, it achieves the best known scaling with respect to the number of non-Clifford gates (see Appendix E of Ref.~\citenum{Fomichev2023-vs} for more details).
Other recent works have considered other kinds of optimizations, such as reducing the circuit depth from linear in the number of qubits to logarithmic~\cite{Malz2024-vr}.
Since \(n\) is relatively small for our application we do not consider these approaches, but it would be interesting to understand if it is possible to achieve a reduction in gate depth while also minimizing the number of non-Clifford gates.

In this work, we make heavy use of tensor train and matrix product state factorizations to represent functions.
Informally, we ``tensorize'' a function \(f\) by first evaluating it on a grid of points over some interval \(\left[ a, b \right]\) to obtain a vector and then reshaping that vector into a tensor.
More precisely, we introduce the following definition.
\begin{restatable}[Tensorization of a function]{definition}{tensorizationdef}
  Let \(f: \mathbb{R} \rightarrow \mathbb{C}\) be a function defined on an interval \(\left[ a, b \right]\).
  Let \(\left\{ x_j \right\}\) be an \(N\)-point, equispaced, discretization of \(\left[ a,b \right]\), i.e., let
  \begin{equation}
    x_j = a + \frac{b - a}{N-1}j, \;\;\; j \in \left[ 0..N-1 \right].
  \end{equation}
  For any \(n \in \mathbb{Z}\) such that \(n \geq \log_2 N\), let
  \(y\) be a \(2^n\)-dimensional vector defined by
  \begin{equation}
    y_j =
    \begin{cases}
      f(x_j) & \text{ if } j < N
      \\
      0      & \text{ otherwise.
      }
    \end{cases}
  \end{equation}

  Let the rank-\(n\) tensor \(T\) be defined b
  \begin{equation}
    T^{s_1 s_2 \cdots s_n} = y_j, \;\;\; j = \sum_{k=1}^{n}2^{n-k} s_k,
  \end{equation}
  where each of the indices have dimension \(2\).
  Then we call \(T\) a rank-\(n\) tensorization of \(f\) over the aforementioned grid.
  \label{def:tensorization}
\end{restatable}

Tensor trains are especially well-suited to representing the tensorizations of certain elementary functions, such as polynomials, exponential functions, and various trigonometric functions~\cite{Grasedyck2010-pn,Oseledets2013-vy}.
In particular, Ref.~\citenum{Grasedyck2010-pn} proved that the tensorization of a degree \(d\) polynomial (in one variable) can be represented as a tensor train with bond dimension \(d+2\), regardless of the number of grid points or the size of the interval.
The bond dimensions required to represent a tensor \(T^{s_1 s_2 \cdots s_n}\) as a tensor train depend on the Schmidt rank of \(T\) across various partitionings.
In Ref.~\citenum{Grasedyck2010-pn}, Grasedyck observed that when we take the Schmidt decomposition of the tensorization of a degree \(d\) polynomial, the Schmidt vectors themselves can also be interpreted as tensorizations of degree \(d\) polynomials on some fixed grid.\footnote{This explanation ignores some small details that must be addressed when the grid size is not a power of two, but it is essentially correct.}
Because tensorization is a linear map, the Schmidt vectors are elements of a \(d+1\)-dimensional vector space, and therefore the Schmidt rank is at most \(d+1\) (or \(d+2\) for technical reasons when the grid size is not a power of two).

Tensorizations of functions also inherit other useful properties from tensor trains, such as the ability to efficiently add them together, take their (element-wise) product, or take their tensor product~\cite{Oseledets2011-xm}.
When a low bond dimension tensor train representation of a function exists, it is possible to efficiently construct this representation without explicitly tensorizing the function first~\cite{Oseledets2010-hb}.
As with tensor trains that arise in other contexts, we can efficiently ``round'' the tensor train representing a function, reducing its bond dimension at the expense potentially incurring some additional error.
This rounding can be implemented using a series of singular value decompositions, and the error can be controlled by discarding only the singular vectors whose singular values are below some chosen singular value cutoff~\cite{Oseledets2011-xm}.
These properties, combined with the ability to convert a matrix product state into a quantum circuit for state preparation, suggest a flexible strategy for preparing quantum states that encode a variety of functions.

Matrix product states are often encountered in the context of quantum systems in one spatial dimension, where they are a natural ansatz for approximating the ground state~\cite{Schollwock2011-mp,Orus2014-pj}.
While methods based on matrix product states typically scale exponentially when used to treat the ground state problem in more than one spatial dimension, the same is not true when using TT (or MPS) to encode functions.
For example, any degree \(d\) polynomial in two variables can be written as a linear combination of \(d+1\) polynomials that factorize between the two variables,
\begin{equation}
  f(x, y) = \sum_{j=0}^{d} g_j(x)h_j(y),
\end{equation}
where \(\operatorname{deg}(g_j) = j\), \(\operatorname{deg}(h_j) = d - j\).
For separable functions, the natural generalization of tensorization to multivariate functions is to take tensor products of the tensorizations of the individual factors.
We can define the tensorization of an arbitrary function by linear extension.
Using the notation \(TT_{w}\) to denote the tensor train encoding the tensorization of a function \(w\), we can write
\begin{equation}
  TT_{f} = \sum_{j=0}^d TT_{g_j} \otimes TT_{h_j}.
\end{equation}
Since the bond dimension of the sum is at most the sum of the bond dimensions, the bond dimension of \(TT_{f}\) is at most \(\frac{1}{2} \left( d^2 + 5d + 4 \right)\).
This example suggests that tensor trains and matrix product states retain their efficiency when applied to functions defined in a small number of spatial dimensions.

Our work is not the first to make use of these factorizations for the purposes of state preparation or other related tasks~\cite{Garcia-Ripoll2019-jk, Holmes2020-yx, Melnikov2023-ny, Iaconis2024-ny, Gonzalez-Conde2023-gw, Plekhanov2022-md, Lubasch2020-hw, Mc-Keever2024-jk, Mc-Keever2023-cd, Cervero-Martin2023-gu}.
In particular, Ref.~\citenum{Holmes2020-yx} makes use of the observation that the tensor train representation of a tensorized polynomial has low bond dimension in order to obtain efficient state preparation circuits for piecewise polynomial functions.
In this work, we focus on solving a particular set of state preparation tasks that enable us to efficiently implement our desired change-of-basis operation, but it would be interesting to further develop state preparation techniques based on tensor trains and compare them with the wide variety of other approaches that have been explored in the literature~\cite{Zalka1998-zt,Grover2000-rb,Grover2002-st,Ward2009-ln,Sanders2019-os,Wang2021-ti,Rattew2022-ck,McArdle2022-ux,Moosa2023-qr, Rosenkranz2024-ik}.

There is also a growing body of work that seeks to use these tools to directly solve quantum chemical problems (besides the standard application of the density matrix renormalization group algorithm to find the ground state of medium-sized strongly-correlated systems).
For example, there has been significant progress in constructing efficient Hartree-Fock SCF solvers that use tensor train representations of single-particle wavefunctions in real space to solve the Hartree-Fock equations on large grids~\cite{Khoromskaia2011-fd, Khoromskaia2015-vi, Jolly2023-ft}.
In the future, it may make sense to use the outputs of similar calculations to prepare initial states for quantum algorithms, rather than projecting solutions obtained in a Gaussian-type basis set as we do here.
Even more directly related to this work, Ref.~\citenum{Jolly2023-ft} has actually already numerically demonstrated that various orbitals, including the Gaussian-type orbitals that we consider here, can be efficiently represented on a real space grid using tensor trains.
Although their results differ slightly from ours and do not rigorously analyze the asymptotic scaling, they do provide additional evidence that the tensor train representations we make use of are efficient in a wide variety of situations.

\subsection{Notation}
\label{sec:notation}

Before we present our main results, we briefly introduce some notation we use throughout the manuscript.
We reserve Dirac bra-ket notation for (not necessarily normalized) wavefunctions of one or more qubits.
Except when indicated otherwise, we use the standard inner product notation \(\inp{f}{g}\) to denote the following inner product over square-integrable functions:
\begin{equation}
  \inp{f}{g} = \int_{x = -L/2}^{L/2} f^*(x) g(x) \, dx.
\end{equation}
When \(f\) and \(g\) are multivariate functions, we use the same notation to denote the analogous inner product
\begin{equation}
  \inp{f}{g} = \int_{\left( x, y, z \right) \in \left[ -L/2, L/2 \right]^3} f^*(x, y, z) g(x, y, z) \, dx \, dy \, dz.
\end{equation}
In other words, we perform these integrals over the computational unit cell in one or more spatial dimensions.
Likewise, we use the notation \(\norm{f}\) to indicate the norm induced by these inner products, i.e.,
\begin{equation}
  \norm{f} = \sqrt{\inp{f}{f}}.
\end{equation}

When we quantify the distance between quantum states, we mainly use the trace distance, since it is a general measure of the distinguishability of two quantum states and it can easily be used to bound the error in the expectation value of any observable of interest.
We use the notation \(D(\rho, \sigma)\) to denote the trace distance between two quantum states \(\rho\) and \(\sigma\),
\begin{equation}
  D(\rho, \sigma) = \frac{1}{2} \Tr \left( \sqrt{\left( \rho - \sigma \right)^\dagger \left( \rho - \sigma \right)} \right).
\end{equation}
For two normalized pure states \(\ket{\psi}\) and \(\ket{\phi}\), we use the notation \(D(\ket{\psi}, \ket{\phi})\) to denote the trace distance. It is useful to note that
\begin{equation}
  D(\ket{\psi}, \ket{\phi}) = \sqrt{1 - \abs{\braket{\psi}{\phi}}^2},
  \label{eq:nielsen_chuang_pure_state_inequality}
\end{equation}
even for infinite-dimensional quantum states.\footnote{Even for infinite-dimensional quantum systems, density operators have a well-defined trace and the trace distance between them is also well-defined. See Ref.~\citenum{Nielsen2012-ym} for a proof of \Cref{eq:nielsen_chuang_pure_state_inequality} that does not rely on any assumptions about the dimensionality of \(\ket{\psi}\) or \(\ket{\phi}\).}

Because we can interpret normalized square-integrable functions as wavefunctions of quantum particles, we also sometimes discuss the trace distance between two such functions \(f\) and \(g\), which we define in the natural way as 
\begin{equation}
  D(f, g) = \sqrt{1 - \abs{\inp{f}{g}}^2}.
\end{equation}

\section{State preparation with tensor networks}
\label{sec:state_prep}

We aim to take wavefunctions defined in a standard Gaussian basis set and efficiently prepare states that approximate them in a first-quantized plane wave representation.
Our starting point is a collection of \(N_{mo}\) molecular orbitals (MOs) obtained from classical mean-field calculations in a Gaussian basis set.
Each molecular orbitals is a normalized linear combination of $N_g$ primitive Gaussians basis functions, i.e.,
\begin{equation}
  \chi = \sum_j^{N_g} c_j g_j,
\end{equation}
where each \(g_j\) represents a primitive Gaussian basis function.
We make the assumption that \(L\), the linear size of the computational unit cell, is large enough that all of the molecular orbitals (and all of the primitive Gaussian basis functions) have negligible support outside of the computational unit cell.

Given a molecular orbital \(\chi\), we demonstrate that it can be well-approximated (over the computational unit cell) by a linear combination of plane waves with a finite momentum cutoff.
Furthermore, we show that the qubit wavefunction encoding this approximation in the usual first-quantized representation can itself be efficiently approximated by a matrix product state with low bond dimension.
We organize these results by first providing bounds on the bond dimensions required to approximate primitive Gaussian basis functions as matrix product states in \Cref{sec:primitive_gaussian_basis_functions_as_tensor_trains}.
In \Cref{sec:individual_orbitals_as_MPS}, we present the extension of these bounds to molecular orbitals formed from linear combinations of primitive Gaussian basis functions.

Given a set of orthonormal qubit wavefunctions \(\ket{\chi_1}, \ket{\chi_2}, \cdots, \ket{\chi_{N_{mo}}}\) representing the molecular orbitals \(\chi_1, \chi_2, \cdots, \chi_{N_{mo}}\), we can define the isometry that encodes the change of basis
\begin{equation}
  V = \sum_{k=1}^{N_{mo}} \ketbra{\chi_k}{k}.
\end{equation}
In \Cref{sec:overall_cost}, we explain how to efficiently implement an approximation to \(V\) given the matrix product state representations for the molecular orbitals.
This, in turn, implies that we can efficiently perform a variety of state preparation tasks.
For example, we show that we can prepare an approximation to an \(\eta\)-electron Hartree-Fock state obtained from a classical calculation in a Gaussian basis set with \(N_g\) primitive function using a number of Toffoli gates that scales as \(\eta^2 N_g^{3/2}\) (neglecting logarithmic factors and assuming that the angular momentum quantum numbers of the primitive Gaussian basis functions are \(\mathcal{O}(1)\)).

\subsection{Representing primitive Gaussian basis functions with matrix product states}
\label{sec:primitive_gaussian_basis_functions_as_tensor_trains}

In this section, we address the problem of efficiently representing a single primitive Gaussian basis function as a matrix product state in our first-quantized plane wave representation.
For now, we consider a function centered at the origin, 
\begin{align}
  g(x,y,z) \propto x^{l}y^{m}z^{n}e^{-\gamma(x^2+y^2+z^2)},
\end{align}
where $l,m,n$ are the angular momentum quantum numbers in the $x,y,$ and $z$ directions respectively, $\gamma$ is a factor determining the width of the Gaussian, and the constant of proportionality is a positive number chosen so that \(g\) is normalized over the real line.
We begin by approximately projecting such a function onto our plane wave basis.
Then we qualitatively explain why we should expect that an efficient description is possible.
Finally, present a formal statement of our technical results, showing that we can represent a single Cartesian component of our primitive Gaussian basis function using a compact matrix product state encoding.

Let us consider the projection of a primitive Gaussian basis function onto a plane wave basis set.
Recall that we can write the plane wave basis function corresponding to the momentum vector \(\bm{k}\) as
\begin{align}
  \varphi_{\bm{k}}(x,y,z) & = \frac{1}{L^{3/2}} e^{-ik_xx}e^{-ik_yy}e^{-ik_zz}
  \label{eq:cartesian_plane_wave_def}
\end{align}
in Cartesian coordinates.
Because both \(\varphi_{\bm{k}}(x, y, z)\) and \(g(x, y, z)\) are separable across \(x\), \(y\), and \(z\) coordinates, we can consider the projection onto each Cartesian component separately.
To that end, let
\begin{equation}
  g_x(x) = c_{l, \gamma} x^l e^{- \gamma x^2},
  \label{eq:g_x_def_main_text}
\end{equation}
where we set \(c_{l, \gamma} = \frac{2^l \left(2\gamma\right)^{l/2 + 1/4}\sqrt{l!}}{\pi^{1/4}\sqrt{(2 l) !} }\) so that \(g(x)\) is normalized over the real line.
Additionally, for any \(k \in \mathbb{K}\), we let \(\phi_k\) denote a plane wave in one spatial dimension with momentum \(k\), normalized over the interval \(\left[ -L/2, L/2 \right]\),
\begin{equation}
  \phi_k(x) = \frac{1}{L^{1/2}} e^{-ikx}.
\end{equation}

As mentioned previously, in order to simplify our analysis, we make the assumption that \(L\) is large enough that the support of \(g_x\) outside of the interval \(\left[ -L/2, L/2 \right]\) is negligible.
We discuss this assumption and the other considerations involved in choosing \(L\) in more detail in \Cref{app:box_size_considerations}.
We choose to make use of this assumption rather than quantitatively analyze the minimum required \(L\) for a combination of reasons.
First of all, the tails of \(g_x\) are suppressed exponentially, so it is clear that the we do not require a very large value of \(L\) in any case.
Secondly, simulating a molecular system using a plane wave basis will demand that we use an \(L\) that is sufficiently large in order to suppress the interactions of the system with its periodic images.
In practice, we expect that any choice of \(L\) that satisfies this condition will also satisfy our assumption.

Formally, the projection of \(g_x\) onto our plane wave basis (within the computational unit cell) is given by \(\sum_{k \in \mathbb{K}} \left( \int_{u=-L/2}^{L/2} \phi_k^*(u) g_x(u) du \right) \phi_k\).
We use our approximation to replace the integral over the interval \(\left[ -L/2, L/2 \right]\) with an integral over \(\mathbb{R}\), defining the function
\begin{equation}
  \tilde{g}_x
  =
  \frac{1}{\widetilde{\mathcal{N}}_x} \sum_{k \in \mathbb{K}} \left(\int_{\mathbb{R}} \phi_k^*(u) g_x(u) du\right) \phi_k,
  \label{eq:tilde_g_x_def}
\end{equation}
where we set \(\widetilde{\mathcal{N}}_x\) to be the positive number such that \(\tilde{g}_x\) is normalized over the interval.
Qualitatively, the assumption that \(g_x\) approximately vanishes outside of this interval implies that \(\widetilde{\mathcal{N}}_x \approx 1\) and \(\tilde{g}_x \approx g_x\).

Having accepted this approximation, our aim is to find a function \(f_x\) that satisfies three criteria:
First of all, \(f_{x}\) should be a good approximation to \(\tilde{g}_x\) (and therefore to \(g_x\)) over the interval \(\left[ -L/2, L/2 \right]\).
Secondly, \(f_x\) should have no support on plane waves above some momentum cutoff \(K\) so that we can truncate to a finite basis.
Thirdly, we should be able to represent the quantum state corresponding to \(f_{x}\) in the usual first quantized plane wave representation with a low bond dimension matrix product state.

Before presenting our formal results, let us consider why these goals should be achievable.
A finite momentum cutoff is reasonable because, like a simple Gaussian function, any given \(g_x\) only has significant support on a particular range of frequencies.
A finite momentum cutoff explains why we can expect a representation of \(\tilde{g}_x\) that is efficient in the limit where we increase the size of our basis by increasing \(K\).

We would also like a representation that scales well when we increase \(L\) while keeping \(K\) fixed.
In this limit, we are essentially representing \(\tilde{g}_x\) using an increasingly fine discretization of the same interval in momentum space.
Intuitively, we make use of the fact that the Fourier transform of \(g_x\) is a ``well-behaved'' function that does not need to be evaluated on an arbitrarily fine grid to be well characterized.
More specifically, the higher-order derivatives of the Fourier transform of \(\tilde{g}_x\) are small enough that we can approximate the function using low-degree polynomial.
We can therefore appeal to the results of Ref.~\citenum{Grasedyck2010-pn} to construct a low bond dimension tensor train that encodes the values of this polynomial evaluated on a grid in momentum space.

\Cref{thm:main_technical_result}, which we state below and prove in \Cref{app:primitive_gaussian_mps_lemma_proof}, formalizes the conclusions of these arguments.
In \Cref{app:3d_proof}, we state and prove a straightforward generalization of this lemma to the case of a primitive Gaussian basis function defined in three spatial dimensions and centered at an arbitrary position.
\begin{restatable}[Efficient MPS representation of one-dimensional primitive Gaussian basis functions\IfAppendix{, formal statement}{, informal statement}]{lemma}{onedmpstheorem}
  For arbitrary \(l\in\mathbb{Z}_{\geq 0}\), \(\gamma \in \mathbb{R}_{>0}\), \(L \in \mathbb{R}_{>0}\), let \(\tilde{g}_{x}\) and \(\widetilde{\mathcal{N}}_x\) be as defined in \Cref{eq:tilde_g_x_def}.

  \IfAppendix{
  For an arbitrary \(\epsilon \in \left( 0, 1 \right)\), let \(K \in \mathbb{R}\) and \(m \in \mathbb{Z}\) be arbitrary numbers that satisfy
  \begin{align}
    K
      & \geq
    2\sqrt{2 \gamma}\sqrt{
      2 \log \left( 2\epsilon^{-1} \right)
      + \log \left( 45 \right)
      + \log \left( 1 + \frac{2 \sqrt{\pi}}{L\sqrt{\gamma}} \right)
      + l \log(4 l)
    },
    \label{eq:K_choice}
    \\
    m & \geq
    \frac{e^2 K^2}{2 \gamma},
    \label{eq:m_choice}
  \end{align}
  with the convention that \(0 \log (0) = 0\).
    
  Let \(\mathbb{K}_{cut}\) be defined as in \Cref{eq:K_cut_def}, i.e., \(\mathbb{K}_{cut} = \left\{ k \in \mathbb{K} : k \leq K \right\}\) with \(\mathbb{K} = \left\{ \frac{2 \pi p}{L} : p \in \mathbb{Z}\right\}\).

  If \(\widetilde{\mathcal{N}}_x \geq 2/3\), then a degree \(m-1\) polynomial \(p(k)\) exists such that the following statements are true:
  }{
  If \(\widetilde{\mathcal{N}}_x \geq 2/3\), then there exists numbers \(K \in \mathbb{R}\), \(m \in \mathbb{Z}\) with
  \begin{equation}
    K = \tilde{\mathcal{O}}\left( \sqrt{\gamma \left(\log(1/\epsilon)  + l \right)}  \right),
  \end{equation}
  \begin{equation}
    m = \tilde{\mathcal{O}}\left(\log(1/\epsilon) + \log\left(1 + \frac{1}{L \sqrt{\gamma}}\right)+ l\right),
  \end{equation}
  and a degree \(m-1\) polynomial \(p(k)\) such that the following statements are true:
  }

  Let
  \begin{equation}
    f_x = \sum_{k \in \mathbb{K}_{cut}} p(k) \phi_k,
  \end{equation}
  where \(\mathbb{K}_{cut} = \left\{ k \in \mathbb{K}: \abs{k} \leq K\right\}\).
  Then
  \begin{equation}
    \norm{f_x} = 1
  \end{equation}
  and 
  \begin{equation}
    D(\tilde{g}_{x}, f_x) \leq \epsilon.
  \end{equation}

  Let \(\ket{f_x}\) denote the function \(f_x\) encoded in the standard representation we use for first-quantized simulation in a plane wave basis, i.e.,
  \begin{equation}
    \ket{f_x} = \sum_{k \in \mathbb{K}_{cut}} p(k) \ket{\frac{kL}{2 \pi}},
  \end{equation}
  where \(\ket{\frac{kL}{2\pi}}\) denotes the computational basis state corresponding to the binary encoding of the signed integer \(\frac{kL}{2 \pi}\).
  For any \(n \in \mathbb{Z}\) such that \(n \geq  \log_2 \left|\mathbb{K}_{cut}\right|\), there exists an \(n\)-qubit matrix product state representation of \(\ket{f_x}\) with bond dimension at most \(2m + 3\).
  \label{thm:main_technical_result}
\end{restatable}

At a technical level, the proof of this lemma begins by considering the Fourier transform of a primitive Gaussian basis function and bounding the error incurred by imposing a momentum cutoff.
We then go on to show that it is possible to approximate the Fourier-transformed function as a low-degree polynomial on the region where the absolute value of the momentum is less than the cutoff.
Ref.~\citenum{Grasedyck2010-pn} allows us to represent the resulting piecewise polynomial function as a tensor train with low bond dimension, which implies the statement about the matrix product state representation.

\subsection{Representing molecular orbitals with matrix product states}
\label{sec:individual_orbitals_as_MPS}

We leverage the fact that primitive Gaussian basis functions have efficient matrix product state representations to show that the same is true for molecular orbitals composed of linear combinations of these functions.
Let us consider a molecular orbital \(\chi\) formed from a linear combination of \(N_g\) primitive Gaussian basis functions,
\begin{equation}
  \chi = \sum_j^{N_g} c_j g_j.
\end{equation}
Note that we are skipping a standard intermediate step where fixed linear combinations of the primitive Gaussian basis functions are used to construct ``atomic orbitals,'' which are then combined to form molecular orbitals.
We can use \Cref{cor:3d_basis_functions} (the generalization of \Cref{thm:main_technical_result} to three spatial dimensions) to approximate each \(g_j\) as a matrix product state.
Matrix product states can be summed to obtain a new matrix product state whose bond dimension is at most the sum of the original bond dimensions, so our main task is to bound the errors carefully.

Similarly to \Cref{sec:primitive_gaussian_basis_functions_as_tensor_trains}, we make the assumption that the computational unit cell is sufficiently large that none of our primitive Gaussian basis functions have significant support outside of \(\left[ -L/2, L/2 \right]^3\).
This allows us to focus on quantifying the error with respect to an approximately projected molecular orbital, defined as
\begin{equation}
  \tilde{\chi} \propto \sum_j^{N_g} c_j \tilde{g}_j,
  \label{eq:chi_tilde_main_text}
\end{equation}
where each \(\tilde{g}_j\) is obtained from the corresponding \(g_j\) by approximately projecting each Cartesian component of the \(g_j\), as in \Cref{eq:tilde_g_x_def}.
As was the case for a single primitive basis function, this assumption simplifies our analysis and introduces a negligible amount of error.
We discuss the motivation for this approximation more substantially in \Cref{app:box_size_considerations}.

One difficulty that we face when constructing a normalized function by taking linear combinations of primitive Gaussian basis functions is the fact that the Gaussian basis functions are generally not orthogonal.
In the limit where the basis functions are linearly dependent, small errors can be amplified arbitrarily.
This same issue can lead to numerical instability in classical calculations, where it is usually controlled by means of an orthogonalization procedure that filters out the subspace corresponding to the most ill-conditioned linear combinations.
We follow the same approach and assume that our molecular orbitals have been constructed using the canonical orthogonalization procedure (which we review in \Cref{app:canonical_orthogonalization}).

Besides the subtlety involved in accounting for the non-orthogonality of the primitive Gaussian basis functions, the rest of the technical details are straightforward.
We present the formal statement of our result below, and provide its proof in \Cref{app:MO_error_bound_proof}.
\begin{restatable}[Efficient MPS representations of molecular orbitals\IfAppendix{, formal statement}{, informal statement}]{lemma}{MOasMPS}
  Let \(\chi\) denote a normalized linear combination of \(N_g\) primitive Gaussian functions obtained by use of the canonical orthogonalization procedure with an eigenvalue cutoff of \(\sigma\) (as reviewed in \Cref{app:canonical_orthogonalization}).
  Let \(\tilde{\chi}\) denote the corresponding approximately projected molecular orbital, as defined in \Cref{eq:chi_tilde_main_text}.
  Let \(\ell\) denote the largest angular momentum quantum number among the Cartesian components of any of the primitive Gaussian basis functions that compose \(\chi\), and let \(\Gamma\) denote the largest value of the width parameter \(\gamma\).
  For a sufficiently large computational unit cell, \(\left[ -L/2, L/2 \right]^3\), and for any \(\epsilon \in \left( 0, 1 \right)\), there exists a function \(\tau\) with the following properties:

  \(\tau\) is normalized over the computational unit cell and
  \begin{equation}
    D(\tau, \tilde{\chi}) \leq \epsilon.
  \end{equation}

  \(\tau\) is a linear combination of plane waves with momenta \(\bm{k} \in \mathbb{K}_{cut}^3\), where \(\mathbb{K}_{cut}^3\) is defined by the momentum cutoff \IfAppendix{
  \begin{equation}
    K
    =
    2\sqrt{2 \Gamma}\sqrt{
      2 \log \left( \frac{288\sqrt{3}N_g}{\epsilon^4 \sigma^2} \right)
      + \ell \log(4 \ell)
      + \log \left( 45 \right)
    }.
    \label{eq:K_choice_for_MO}
  \end{equation}
  }{
  \begin{equation}
    K
    =
    \mathcal{O} \left( \sqrt{\Gamma} \sqrt{\log  N_g  + \log  \epsilon^{-1}  + \log\sigma^{-1} + \ell \log  \ell}\right).
    \label{eq:K_choice_for_MO_big_o}
  \end{equation}
  }

  Let \(\ket{\tau}\) denote the quantum state that encodes \(\tau\) in the usual first-quantized plane wave representation (where each computational basis state is associated with a triplet of signed integers indexing a particular plane wave).

  We can represent \(\ket{\tau}\) as a matrix product state on \(n \sim 3 \log_2 \left( KL \right)\) qubits with bond dimension at most \IfAppendix{
  \begin{equation}
    M \leq 8e^2 N_g \left(
    2 \log \left( \frac{288\sqrt{3}N_g}{\epsilon^4 \sigma^2} \right)
    + \ell \log(4 \ell)  + 4\right).
    \label{eq:MO_mps_bound}
  \end{equation}
  }{
  \begin{equation}
    M = \mathcal{O}\left( N_g \left( \log  N_g  + \log \epsilon^{-1}  + \log\sigma^{-1} + \ell \log  \ell   \right) \right).
  \end{equation}
  }
  \label{lemma:mo_as_mps}
\end{restatable}

For the most part, this lemma is a direct consequence of \Cref{thm:main_technical_result} together with the standard rules for manipulating matrix product states (specifically, that the bond dimension of the sum is at most the sum of the bond dimensions).
Notably, the factor of \(N_g\) in the bond dimension bound arises from adding the matrix product states for \(N_g\) primitive Gaussian basis functions together.
The only technical difficulty that must be addressed is that the primitive Gaussian basis functions are not orthonormal, which presents an obstacle in bounding the overall error.
By making the assumption that the molecular orbitals are obtained by means of a canonical orthogonalization procedure, we avoid this issue by limiting the extent to which the sum is poorly conditioned.

From the asymptotic bounds on the momentum cutoff and the bond dimension, it is easy to see which parameters might have a significant impact on the overall cost of our procedure.
Qualitatively, the momentum cutoff mostly depends on the sharpness of the sharpest primitive Gaussian basis function.
Except for the angular momentum quantum number (which can be considered a small constant factor for any reasonable basis set), the rest of the parameters enter only logarithmically into the bounds on \(K\).
Similarly, we find that the bounds on the bond dimension depends mostly on the number of primitive Gaussian basis functions.
Interestingly, our numerical results (which we present in \Cref{sec:MO_MPS_numerics}) reveal that the actual scaling with \(N_g\) might be sub-linear in practice.

\subsection{Implementing the many-body change of basis and state preparation}
\label{sec:overall_cost}

In this section, we show how the results of \Cref{lemma:mo_as_mps} allow us to accomplish a variety of state preparation tasks.
We begin by describing how we implement the change of basis that maps a first-quantized wavefunction from a molecular orbital basis of size \(N_{mo}\) to a plane wave basis of size \(N \sim \left( \frac{K L}{\pi} \right)^3\).
Using this change of basis, we construct efficient state preparation protocols for Slater determinants, and for more sophisticated wavefunctions.
We briefly discuss the gate complexity and error bounds for our protocols, deferring the detailed analysis to \Cref{app:gadget_appendix}.

Assume that \(\left\{ \ket{\chi_m} \right\}\) is a collection of orthonormal quantum states representing \(N_{mo}\) molecular orbitals in a plane wave basis.
Then 
\begin{equation}
  V = \sum_{m}^{N_{mo}} \ketbra{\chi_m}{m}
  \label{eq:idealized_isometry}
\end{equation}
defines a single-particle change of basis mapping from the molecular orbital basis to the plane wave basis.
The many-body operator that implements this change of basis in first quantization is \(V^{\otimes \eta}\).

Given access to quantum circuits that prepare each \(\ket{\chi_m}\), it is straightforward to implement \(V\) using the general approach for unitary synthesis outlined in Ref.~\citenum{Low2018-uu} and Ref.~\citenum{Kliuchnikov2013-il}, which we discuss further in \Cref{app:unitary_synthesis_description}.
Specifically, their approach constructs a unitary \(W\) that acts as \(\ketbra{0}{1} \otimes V + \ketbra{1}{0} \otimes V^\dagger\) on the appropriate input states.
Approaching the implementation of \(V\) this way is convenient because we can construct \(W\) as a product of reflections,
\begin{align}
  W = &\prod_{m=1}^{N_{mo}} \mathbb{I} - 2 \ketbra{w_m},
  \;\;\;\;\;\;\;\;
  \nonumber \\ 
  \ket{w_m} =& \frac{1}{\sqrt{2}} \left( \ket{1}\ket{m} - \ket{0}\ket{\chi_m} \right).
  \label{eq:W_def}
\end{align}
We can implement each of these reflections using two calls to the (controlled) state preparation circuits for each \(\ket{\chi_m}\), a reflection about the \(\ket{0}\) state, and a few Clifford gates.
The zero state reflection can be performed using a number of Toffoli gates that scales linearly in the number of qubits, which is typically much smaller than the cost of state preparation, so we can neglect this factor and express the overall cost as
\begin{equation}
  \textsc{Toffoli}\left(V^{\otimes \eta}\right) \sim 2 \eta \sum_{m=1}^{N_{mo}} T_m,
  \label{eq:toff_cost_v_eta}
\end{equation}
where \(T_m\) denotes number of Toffoli gates required to prepare \(\ket{\chi_m}\) conditioned on an ancilla qubit.

In practice, while the overall strategy remains the same, we do not exactly implement our change of basis using the operation \(W\).
To begin with, we use \Cref{lemma:mo_as_mps} to approximately represent the molecular orbitals with a collection of matrix product states \(\ket{\tau_1}, \ket{\tau_2}, \cdots, \ket{\tau_{N_{mo}}}\).
As we discuss in \Cref{app:mps_state_prep_details}, we also make an additional approximation when compiling the MPS state preparation circuits into a discrete gate set.
Following the approach of Fomichev \textit{et al.} in Appendix E of Ref.~\citenum{Fomichev2023-vs}, the number of Toffoli gates required to prepare an \(n\)-qubit MPS with bond dimension \(m\) to within an error \(\epsilon_2\) in the trace distance scales as
\begin{equation}
  \textsc{Toffoli}\left( \ket{\psi} \right) = \mathcal{O}(n m^{3/2} \log \left( \epsilon_2^{-1} \right)).
  \label{eq:toff_cost_mps}
\end{equation}
Let \(\ket{\tilde{\tau}_1}, \ket{\tilde{\tau}_2}, \cdots, \ket{\tilde{\tau}_{N_{mo}}}\) denote the quantum states prepared by the compiled MPS preparation circuits.
In our actual implementation of the approximate basis change, we follow the basic strategy described above, except that we approximate \(W\) by replacing each \(\ket{\chi_m}\) with the corresponding \(\ket{\tilde{\tau_m}}\), yielding
\begin{align}
  X =& \prod_{m=1}^{N_{mo}} \mathbb{I} - 2 \ketbra{x_m},
  \nonumber \\
  \ket{x_m} =& \frac{1}{\sqrt{2}} \left( \ket{1}\ket{m} - \ket{0}\ket{\tilde{\tau}_m} \right).
  \label{eq:X_def}
\end{align}

Before writing a single expression that encapsulates the overall cost of our approximate change of basis, we recall some notation.
As above, we let \(N_{mo}\) denote the number of molecular orbitals, \(N_g\) denote the number of primitive Gaussian basis functions, and \(\eta\) denote the number of particles.
We let \(\ell\) denote the largest angular momentum quantum number and \(\Gamma\) denote the largest value of the width parameter \(\gamma\) among the primitive basis functions used in the original Gaussian basis set.
Additionally, we let \(\sigma\) denote the eigenvalue cutoff of the canonical orthogonalization procedure, \(\epsilon_1\) denote the precision parameter for the individual MPS from \Cref{lemma:mo_as_mps}, and \(\epsilon_2\) denote the bound on the error incurred from compiling to a finite gate set.
The number of qubits in each particle register, \(n \sim 3 \log_2(KL)\), is determined by the momentum cutoff \(K\) and the linear size of the computational unit cell in real space \(L\).
We assume that \(K \propto \sqrt{\Gamma} \sqrt{\log N_g + \log \epsilon_1^{-1} + \log \sigma^{-1} + \ell \log \ell}\) in order to satisfy the demands of \Cref{lemma:mo_as_mps}.

Using the bounds on the bond dimensions from \Cref{lemma:mo_as_mps}, combined with the expressions for the non-Clifford gate complexities in \Cref{eq:toff_cost_v_eta} and \Cref{eq:toff_cost_mps}, we find that
\begin{widetext}
\begin{equation}
  \textsc{Toffoli}\left(X^{\otimes \eta}\right) = \widetilde{\mathcal{O}}\left( \eta N_{mo} N_g^{3/2} \left(\log L + \log \Gamma\right)\left(\ell + \log \epsilon_1^{-1} + \log \sigma^{-1}\right)^{3/2} \log  \epsilon_2^{-1} \right),
  \label{eq:overall_scaling}
\end{equation}
\end{widetext}
where the notation \(\widetilde{\mathcal{O}}\left( \cdot \right)\) indicates that we neglect subdominant logarithmic factors.
From this equation, we can see that the dominant factors in the scaling are the number of particles, the number of molecular orbitals, and the number of Gaussian orbitals in the original basis.
Since the number of plane waves only enters into the cost logarithmically (through the linear scaling in \(n\) in \Cref{eq:toff_cost_mps}), the overall scaling is relatively insensitive to the size of the computational unit cell (\(L\)) and the sharpness of the most localized Gaussian orbitals (\(\Gamma\)).
If desired, we could derive a similar expression that explicitly shows the logarithmic dependence of the Toffoli complexity on the total number of plane waves.
This expression would be identical to \Cref{eq:overall_scaling}, except with \(\log L + \log \Gamma\) replaced by \(\log N\).\footnote{While it is true that \(\log L + \log \Gamma \neq \log N\), the asymptotic notation we use would hide the other differences.}

In addition to bounding its overall cost, we would like to quantify the error between the ideal change of basis and our approximation.
Most of the sources of error are accounted for by \(\epsilon_1\) (the error from \Cref{lemma:mo_as_mps}) and \(\epsilon_2\) (The error from compiling to a finite gate set).
However, there is one additional approximation that we have made, related to the finite size of the computational unit cell.
As we discussed in \Cref{sec:individual_orbitals_as_MPS}, \Cref{lemma:mo_as_mps} does not directly bound the error between the true molecular orbital and its matrix product state approximation.
Instead, the error bound is calculated with respected to an approximately projected molecular orbital, defined in \Cref{eq:chi_tilde_main_text}.
By construction, the approximately projected molecular orbitals and the true molecular orbitals would be exactly equal (over the computational unit cell) if the primitive Gaussian basis functions had no support outside of the computational unit cell.
However, while the primitive Gaussian basis functions decay exponentially, they have some small support outside of the computational unit cell.

Rather than bound this error precisely, we will make the assumption that it is negligible.
Specifically, as in \Cref{sec:primitive_gaussian_basis_functions_as_tensor_trains} and \Cref{sec:individual_orbitals_as_MPS}, we make the assumption that the primitive Gaussian basis functions and the molecular orbitals have negligible support outside of the computational unit cell.
We justify this assumption qualitatively in \Cref{app:box_size_considerations} by arguing that any reasonable choice of \(L\) for a practical calculation will lead to Gaussian basis functions with vanishingly small support outside of the computational unit cell.
Using this assumption, we can conflate the true molecular orbitals and the approximately projected molecular orbitals, and therefore interpret \Cref{lemma:mo_as_mps} as bounding the overall error between a molecular orbital and its approximation as a matrix product state.
As we discuss in \Cref{app:sd_preparation}, we can therefore bound the overall error for our change of basis operation (up to our approximation) by
\begin{equation}
  \norm{W^{\otimes \eta} - X^{\otimes \eta}} \leq 2^{3/2} \eta N_{mo} \left(\epsilon_1 + \epsilon_2 \right),
  \label{eq:overall_error_bound}
\end{equation}
where \(\norm{\cdot}\) denotes the spectral norm.
Because the cost of implementing \(X^{\otimes \eta}\) scales logarithmically with \(\epsilon_1\) and \(\epsilon_2\), we can control the overall error with only a logarithmic overhead. 

Our strategy for mapping from a molecular orbital basis to a plane wave basis enables efficient state preparation for both mean-field states, such as the Hartree-Fock state, and more sophisticated correlated wavefunctions.
We consider the Hartree-Fock state first.
When expressed in the molecular orbital basis in first quantization, the Hartree-Fock wavefunction is simply an antisymmetric superposition over the first \(\eta\) computational basis states,
\begin{equation}
  \ket{\mathcal{A}} = \frac{1}{\eta!
  } \sum_{\pi \in S_\eta} (-1)^{|\pi|} \ket{\pi(1)}\ket{\pi(2)}\cdots{\ket{\pi(\eta)}}.
\end{equation}
This state can be efficiently prepared with the algorithm of Ref.~\citenum{Berry2018-ey}, using a number of Toffoli gates that scales as \(\eta \log \eta \log N\).
To complete the state preparation in the plane wave basis, we apply our change of basis algorithm to \(\ket{\mathcal{A}}\), setting \(N_{mo} = \eta\).
This rotates each \(\ket{m}\) to an approximation of the corresponding molecular orbital (expressed in the plane wave basis),
\begin{align}
  &\left(\bra{0}\otimes \mathbb{I}\right) X^{\otimes \eta} \ket{1} \ket{\mathcal{A}} \approx 
  \nonumber \\
  &\frac{1}{\eta!
  } \sum_{\pi \in S_\eta} (-1)^{|\pi|} \ket{\pi(\chi_1)}\ket{\pi(\chi_2)}\cdots{\ket{\pi(\chi_\eta)}},
\end{align}
preparing the desired state.
The overall cost of preparing the Hartree-Fock state in the plane wave is given by \Cref{eq:overall_scaling} with \(N_{mo}\) set to \(\eta\), since the cost of preparing the antisymmetric initial state is negligible.
More explicitly, the Toffoli cost for Slater determinant state preparation scales as \(\eta^2 N_g^{3/2} \ell\) when we neglect all logarithmic factors.

To efficiently prepare states beyond a single Slater determinant, we can combine our change of basis algorithm with the strategy outlined in Ref.~\citenum{Babbush2023-ud} for mapping from second to first quantization.
Consider a circuit that prepares an arbitrary \(\eta\)-particle correlated wavefunction \(\ket{\Psi}\) in a molecular orbital basis of size \(N_{mo}\), such as a density matrix renormalization group approximation to the ground state.  
Following Appendix G of Ref.~\citenum{Babbush2023-ud}, we can convert this wavefunction to first quantization.
This conversion requires a number of Toffoli gates scaling as \(\widetilde{\mathcal{O}}(\eta N_{mo})\). 
After the conversion, we will have encoded the same \(\ket{\Psi}\) as a first-quantized wavefunction in a molecular orbital basis.
In other words, the first-quantized representation of \(\ket{\Psi}\) will be an antisymmetric linear combination of states of the form \(\ket{b_1}\ket{b_2}\ket{b_3} \cdots \ket{b_\eta}\), where the each \(b_k\) is a computational basis state representing one of the \(N_{mo}\) molecular orbitals, e.g., an integer from \(1\) to \(N_{mo}\).

To complete the desired state preparation task, we need to implement the change of basis that transforms the wavefunction from a first-quantized molecular orbital representation to a first-quantized plane wave one.
The procedure is exactly the same as for a single Slater determinant, except that we need to transform \(N_{mo}\) molecular orbitals rather than \(\eta\) (for some \(N_{mo} > \eta\)).
The operation that we would like to apply is \(V^{\otimes \eta}\), with \(V\) as defined in \Cref{eq:idealized_isometry}.
We approximate this operation by encoding the orbitals as matrix product states, as described above.
The cost of mapping from second-quantization to first-quantization (without changing the basis) is negligible compared to the cost of implementing the basis transformation from a molecular orbital basis to a plane wave one.
Therefore, neglecting the cost of preparing the second-quantized wavefunction on \(N_{mo}\) molecular orbitals, the scaling for the cost of first-quantized state preparation is simply given by \Cref{eq:overall_scaling}.
The only source of error is due to the change of basis, which is bounded by \Cref{eq:overall_error_bound}.
Both our method and a naive approach scale linearly in \(N_{mo}\), but we reduce the naive linear scaling in the number of plane waves to a logarithmic scaling.

\section{Numerical results}
\label{sec:MO_MPS_numerics}

As we have explained in the preceding sections, our main theoretical result is that single-particle wavefunctions constructed in a Gaussian basis set and then projected onto a plane wave basis can be efficiently represented using a MPS or TT factorization.
The bond dimension required for an accurate approximation scales weakly with the grid size parameters $L$ and $K$, and with the desired accuracy. 
This allows for an exponential reduction in the gate count as a function of the number of grid points with increasing $L$ compared to the prior state-of-the-art for first-quantized Slater determinant state preparation in a plane wave basis. 
We now demonstrate that our approach can provide a significant practical benefit, even for small molecular systems, with numerical calculations performed on water molecules, benzene, and a series of larger acenes. 
These results complement similar numerical observations that have been obtained for real-space grids, which should behave qualitatively similarly~\cite{Jolly2023-ft}.
The details of our calculations are described in \Cref{app:computational_methods}, including both the construction of the quantum chemical Hamiltonians and the subsequent generation of the matrix product state representations of the orbitals. 

We perform several different numerical experiments to establish the utility of our approach.
First, we demonstrate the efficacy of plane wave basis sets for the representation of smooth molecular orbitals. Second, we show that our asymptotic results for the weak scaling of the bond dimension with the number of grid points are realized in practical calculations on small molecules. Finally, we provide numerical evidence that the analytic bound we have derived in terms of the number of primitive Gaussian basis functions ($N_g$) overestimates the true cost, and in fact the bond dimension scales very weakly, if at all, with $N_g$ in practice. We conclude with a detailed comparison of the quantum resource estimates of our Slater determinant state preparation method against the previous state-of-the-art, achieving a reduction of several orders of magnitude in the estimated Toffoli count.

We begin in \Cref{fig:plane_wave_projection} by studying the infidelity between the occupied molecular orbitals of a water molecule and their MPS/TT-factorized projections onto a finite plane wave basis.
The occupied orbitals are first calculated in a cc-pVDZ Gaussian basis set using restricted Hartree-Fock (RHF) as implemented by the PySCF software package~\cite{Sun2017-sx}.
As explained in \Cref{app:computational_methods}, we approximate their projection onto a plane wave basis by taking the direct summation over matrix product states representing the appropriate primitive Gaussian basis functions.
To reduce the computational overhead, the MPS representing the partial sum is truncated after each addition with a singular value threshold of $10^{-9}$, which incurs some additional error in the wavefunction. 
The plateauing of the error in the lower right side of \Cref{fig:plane_wave_projection} is due to the accumulated errors from this intermediate compression.
After constructing the MPS for a given molecular orbital and truncating it again with some larger singular value cutoff, we use the methodology described in \Cref{app:computational_methods} to estimate the error (in the trace distance) between the approximated MO and the exact MO.\footnote{In calculating the error, we make the approximation that the primitive Gaussian basis functions vanish outside of $[-L/2,L/2]$.
This assumption is true up to machine precision for the values of $L$ that we consider in this section.}

\begin{figure*}
  \centering
  \begin{subfigure}[t]{.49\textwidth}
    \centering
    \includegraphics[width=\textwidth]{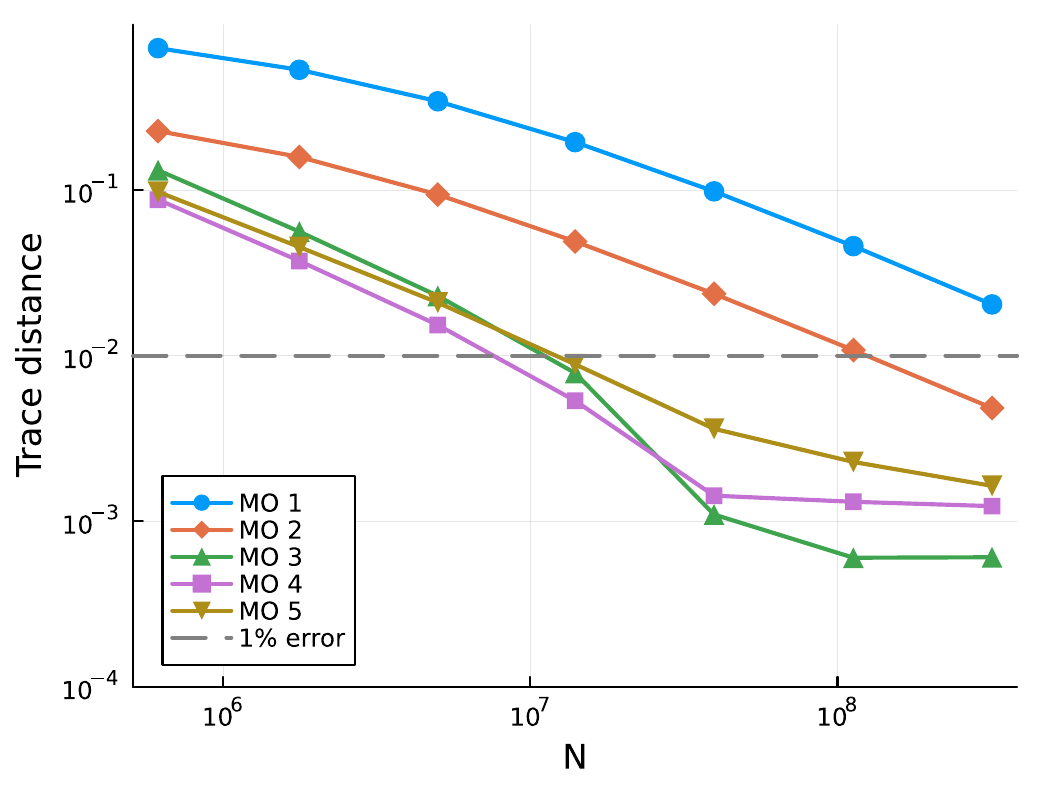}
    \caption{Fixed singular value threshold, varying momentum cutoff.}
    \label{fig:plane_wave_projection_left}
  \end{subfigure}
  \begin{subfigure}[t]{.49\textwidth}
    \centering
    \includegraphics[width=\textwidth]{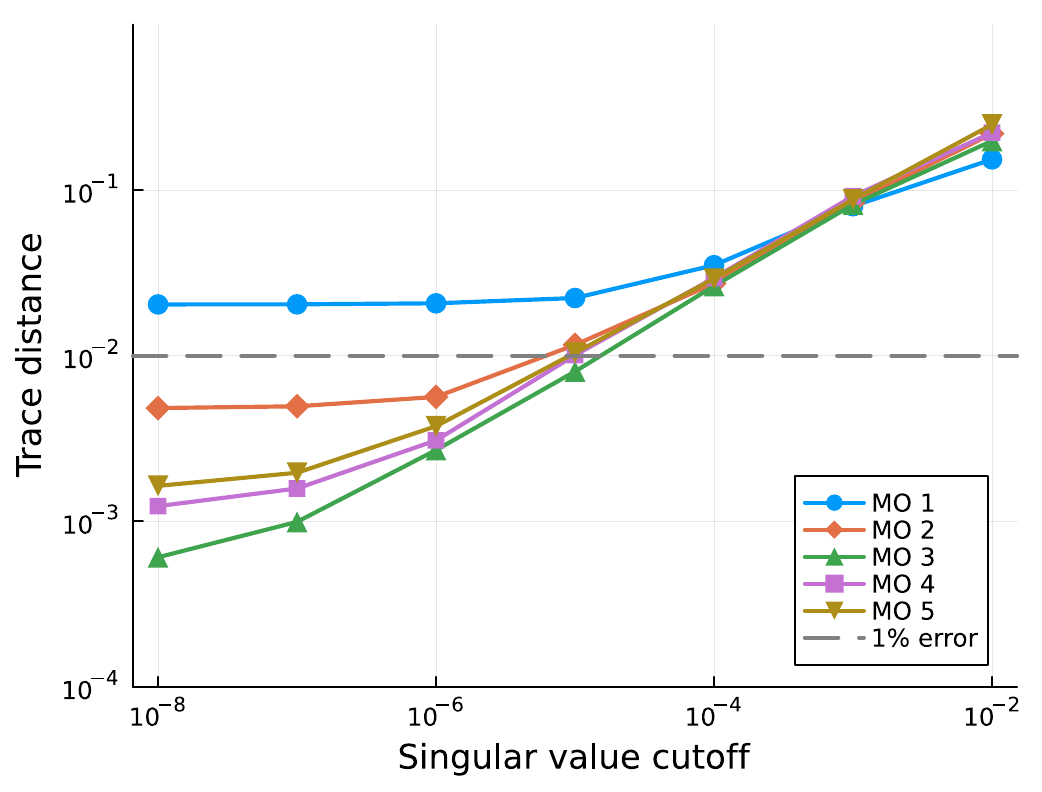}
    \caption{Fixed momentum cutoff, varying singular value threshold.}
    \label{fig:plane_wave_projection_right}
  \end{subfigure}
  \begin{subfigure}[t]{\textwidth}
    \vspace{3pt}
  \includegraphics[width=.9\textwidth]{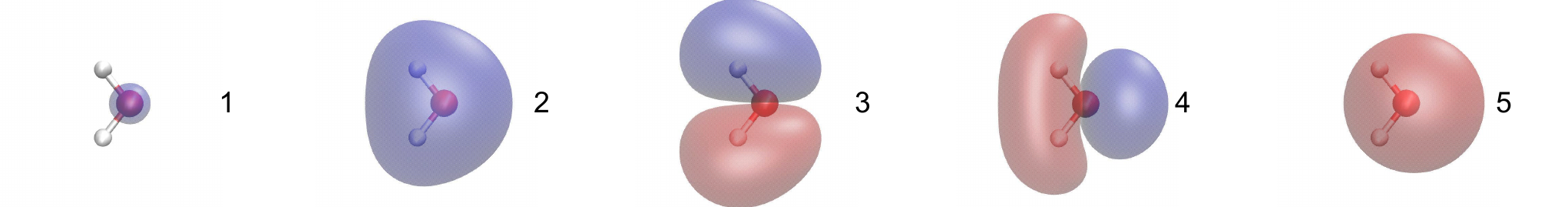}
  \end{subfigure}
  \caption{The error (measured in the trace distance) for the occupied RHF molecular orbitals of a water molecule under projection onto a plane wave basis and subsequent tensor train factorization. The RHF orbitals were computed using the cc-pVDZ basis set and the size of the computational unit cell was fixed to \(L=60\) Bohr.
  Panel (a) shows the error incurred using a constant singular value truncation threshold of $10^{-8}$ for the MPS compression while varying the number of plane waves via an increasing momentum cutoff $K$. 
  Panel (b) shows the error obtained with a constant number of plane waves ($N=3.19\times10^8$) over a range of singular value cutoffs.
  The accuracy of the delocalized orbitals converges quickly with an increasing plane wave basis set size, while the most localized orbital requires a large number of plane waves for an accurate representation.
  }
  \label{fig:plane_wave_projection}
\end{figure*}

In \Cref{fig:plane_wave_projection_left}, we plot the error as a function of the number of plane waves in the basis set (\(N\)) for each of the five occupied molecular orbitals.
In this panel, we use a small predefined singular value threshold of \(10^{-8}\) and a fixed computational unit cell with a side length of \(L=60\) Bohr.
We vary the number of plane waves by varying the momemtum cutoff over a range that corresponds to kinetic energies from \(10\) to \(640\) Hartee.
As expected, the highly spatially localized molecular orbital dominated by the \(1s\) atomic Gaussian orbital of the oxygen atom requires a much larger momentum cutoff to achieve a small error than the more delocalized orbitals involving the oxygen \(2s\) and \(2p\) subshells. Classical simulations in a plane wave basis typically use pseudopotentials to avoid the need for the high-frequency plane waves required to explicitly represent highly localized core orbitals~\cite{Hamann1979-gq, Bachelet1982-xg, Schwerdtfeger2011-jr}.
Rather than including the core electrons in the wavefunction of the system, pseudopotential approaches remove them from the simulation and replace the all-electron potential with an effective potential that treats the nuclei and their core electrons together.
Recent studies have begun to explore the combination of pseudopotentials with first-quantized quantum simulation~\cite{Zini2023-zt, Berry2023-nn}. 
It would be a useful subject for future work to adapt our state preparation approach for use in concert with these methods.

In \Cref{fig:plane_wave_projection_right}, we fix the number of plane waves (\(N=3.19 \times 10^8\)) and the size of the computational unit cell (\(L=60\) Bohr) while varying the final singular value cutoff.
We find that we can systematically reduce the overall error by reducing the singular value threshold down to a floor that is set by the ability of the underlying plane wave basis set to represent the MOs.
Below the plots in \Cref{fig:plane_wave_projection}, we show visualizations of the five occupied MOs.
We note that the fifth occupied MO has a nodal plane that is obscured by the top-down view.

The key idea behind our approach is that we can efficiently represent molecular orbitals (projected into a plane wave basis) using matrix product states.
In \Cref{fig:benzene_orbitals}, we show the bond dimensions obtained when numerically compressing four different molecular orbitals of the benzene molecule, selected to convey the full range of observed behavior under different approximation conditions.
In all cases, we use $N=3.97\times10^7$ plane waves and a side length of $L=60$ Bohr for the computational unit cell.
These parameters are chosen so that we have a sufficiently high momentum cutoff to represent the localized molecular orbitals and a sufficiently large \(L\) to ensure that the primitive Gaussian basis functions have negligible support outside of the computational unit cell.
We consider two different factors that affect the accuracy of the molecular orbital representation: the singular value threshold used for the numerical compression of the matrix product states (as described in \Cref{app:computational_methods}), and the choice of the original primitive Gaussian basis set.
Although the bond dimensions required are higher for the small cutoff (blue triangular markers) than they are for the large cutoff (orange circular markers), for both choices of threshold we see that the compressed orbitals have a bond dimension that is two to three orders of magnitude lower than in the worst case, i.e., an MPS that can represent an arbitrary state (grey square markers).

\begin{figure*}
  \includegraphics[width=.49\textwidth]{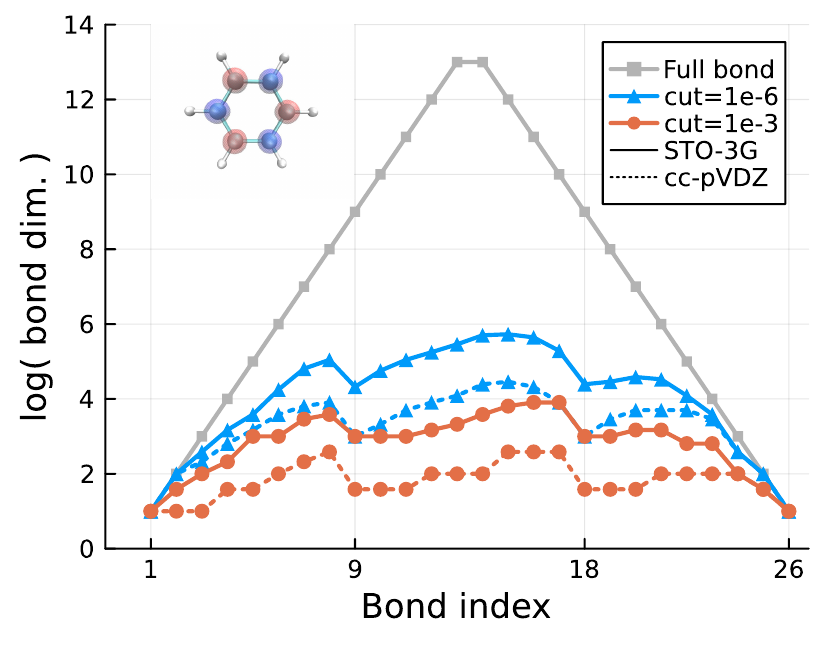}
  \includegraphics[width=.49\textwidth]{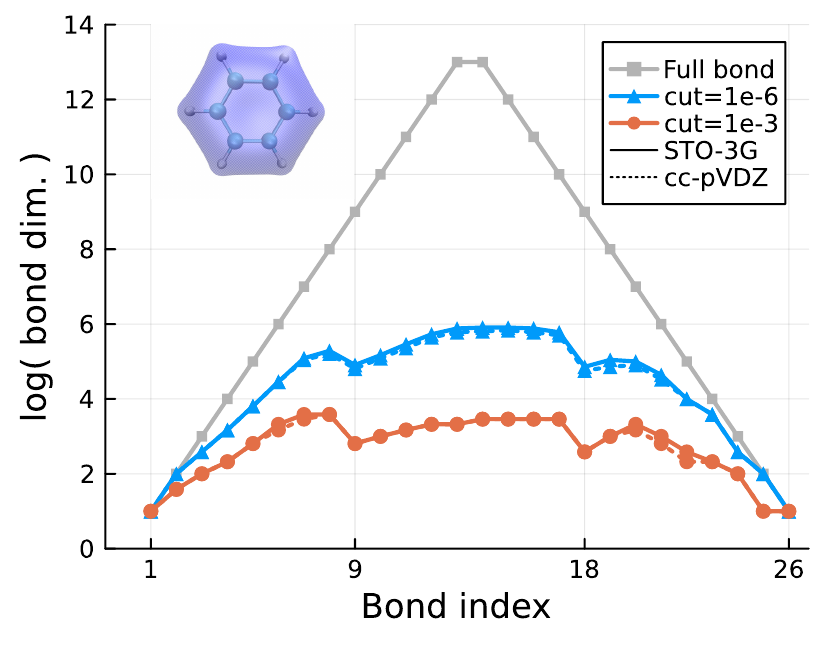}
  \includegraphics[width=.49\textwidth]{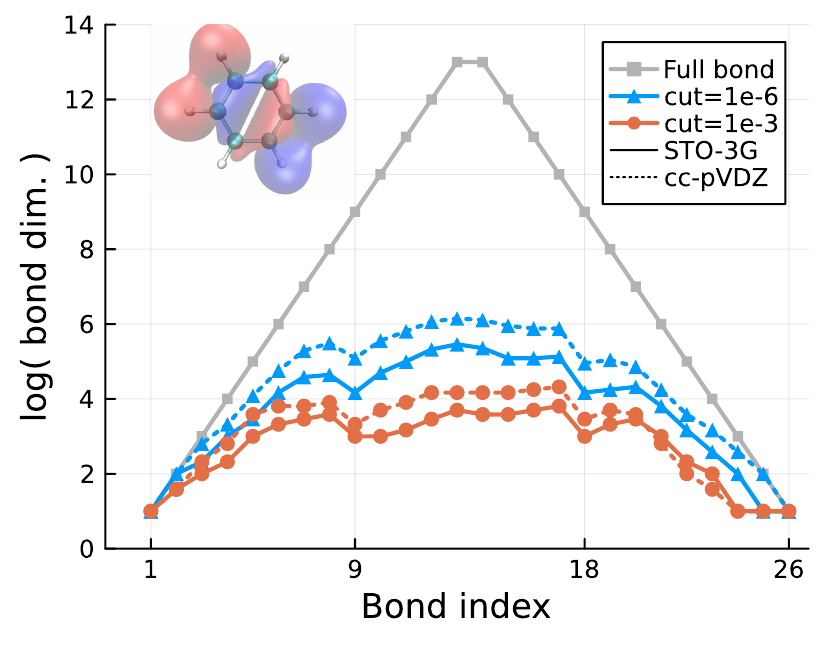}
  \includegraphics[width=.49\textwidth]{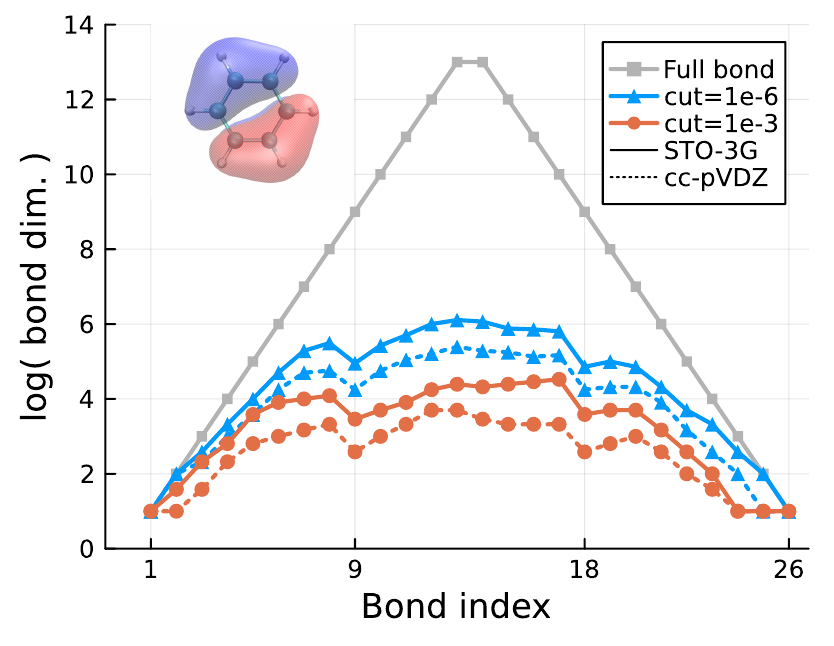}
  \caption{
    The bond dimension required encode selected MOs of benzene projected onto a plane wave basis.
    We plot the logarithm of the bond dimension on the y axis for each of the virtual indices of the MPS representation (x axis).
    The MOs are first calculated in the STO-3G ($N_g=108$ primitive basis functions, solid lines) and cc-pVDZ ($N_g=204$, dotted lines) basis sets.
    They are then projected onto a plane wave basis consisting of $N=3.97\times10^7$ plane waves defined on a computational unit cell whose side length is $L=60$ Bohr.
    We show results for two different singular value cutoffs, \(10^{-6}\) (blue triangles) and \(10^{-3}\) (orange circles).
    We find that there is no consistent increase (and often a decrease) in bond dimension when results from the minimal STO-3G basis set are compared with the larger and more accurate cc-pVDZ basis set, despite doubling the number of Gaussian primitives.
    For all choices of basis set and singular value cutoff we find a large reduction in bond dimension when compared with the bond dimension required to represent an arbitrary state (grey squares).
  }
  \label{fig:benzene_orbitals}
\end{figure*}

Interestingly, we find that the bond dimension depends only very weakly on the size of the Gaussian basis set. 
In fact, in two of the four cases we study, we can represent the wavefunctions obtained using the richer cc-pVDZ basis set (dotted lines) with a smaller bond dimension than is required for the minimal STO-3G basis set (solid lines), despite the fact that the cc-pVDZ basis has 204 primitive basis functions compared with the 108 of the STO-3G basis. 
We rationalize this observation as follows. 
First, the larger Gaussian basis sets typically contribute additional basis functions for the valence electrons which are more diffuse, corresponding to single-particle wavefunctions of higher total energy, and therefore usually contribute to the occupied molecular orbitals with small coefficients.

Secondly, while the upper bound on the bond dimension that we present in \Cref{lemma:mo_as_mps} scales linearly with the number of primitive Gaussian basis functions, this upper bound may be too pessimistic.
Specifically, the linear scaling with \(N_g\) does not take into account the possibility that the MPS representation of a molecular orbital may be further compressed after being constructed from a linear combination of the MPSs that encode the primitive Gaussian basis functions.
As we discussed briefly in \Cref{sec:MPS_and_TT}, some multivariate functions (such as low-degree polynomials in a few variables) have efficient representations as tensor trains.
We conjecture that a numerical study of the molecular orbitals might reveal that they can be well-approximated directly by such functions, even in the infinite basis set limit (where \(N_g \rightarrow \infty\)).
When the addition of more diffuse Gaussians does nontrivially alter the shape of the molecular orbital solutions, this alteration may help the MOs converge towards a `well-behaved' function that can be efficiently represented by a tensor train.

Next, we analyze the number of Toffoli gates required to approximate the highest occupied molecular orbital (HOMO) state for a series of four acene molecules, shown in \Cref{fig:mo_preparation_costs}.
This series, which consists of an increasing number of fused benzene rings, presents an opportunity to study the behavior of our approach as we systematically increase the size of a molecular system.
In this analysis, we keep the size of the computational unit cell fixed by setting \(L=60\) Bohr and vary the number of plane waves by varying the momentum cutoff (x axis in both panels).
The molecular orbitals are obtained from calculations using the cc-pVDZ Gaussian basis set before being approximately projected into a plane wave basis and compressed using a TT factorization.
We plot data for three different choices of the singular value cutoff used to compress the MPS representation of the HOMOs, \(10^{-4}\) (green line with short dashes), \(10^{-6}\) (orange line with long dashes), and \(10^{-8}\) (solid blue lines).
We calculate the number of Toffoli gates using the methodology outlined in \Cref{app:mps_state_prep_details}, choosing the number of bits so that the error in the rotation gates used for unitary synthesis is bounded by $10^{-16}$.
More specifically, we use the expression \Cref{eq:mps_state_prep_app_toff_count}, together with the numerically determined values for the bond dimensions required to represent the orbitals.
As above, we approximate the error using the approach described in \Cref{app:computational_methods}.

\begin{figure*}
  \centering
  \begin{subfigure}[t]{.49\textwidth}
    \centering
    \captionsetup{skip=-16pt,justification=raggedright,singlelinecheck=false}
    \includegraphics[width=.95\textwidth]{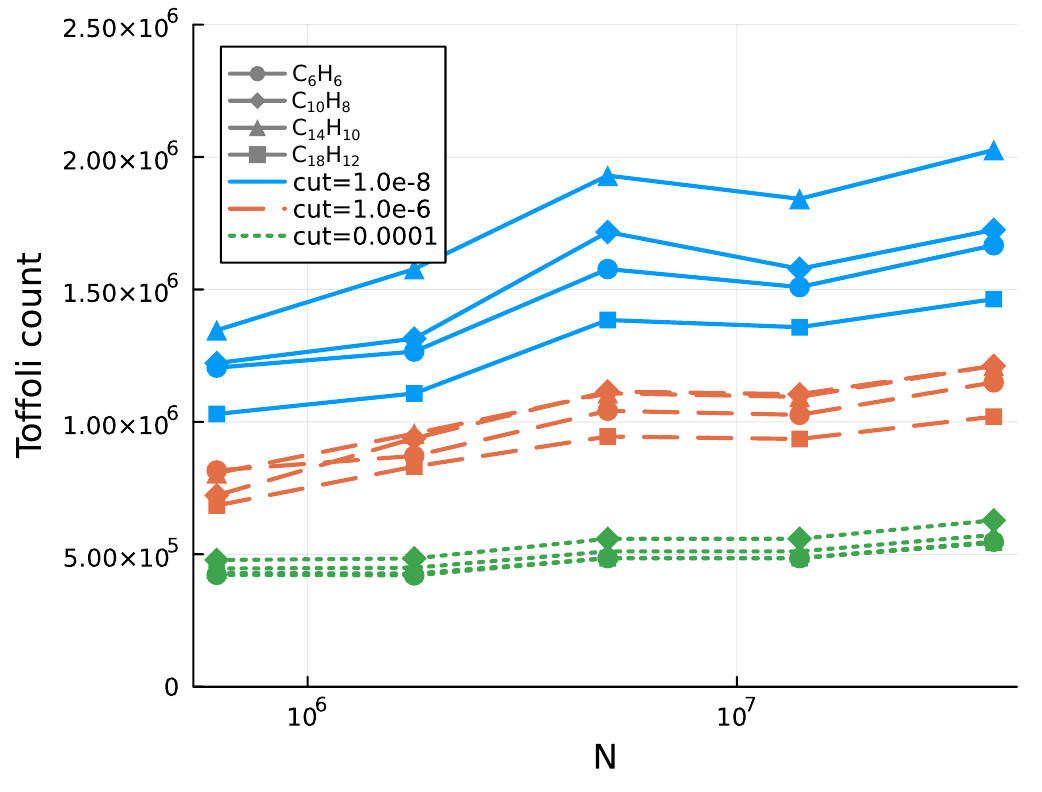}
    \caption{}
    \label{fig:mo_preparation_costs_left}
  \end{subfigure}
  \begin{subfigure}[t]{.49\textwidth}
    \centering
    \captionsetup{skip=-16pt,justification=raggedright,singlelinecheck=false}
    \includegraphics[width=.95\textwidth]{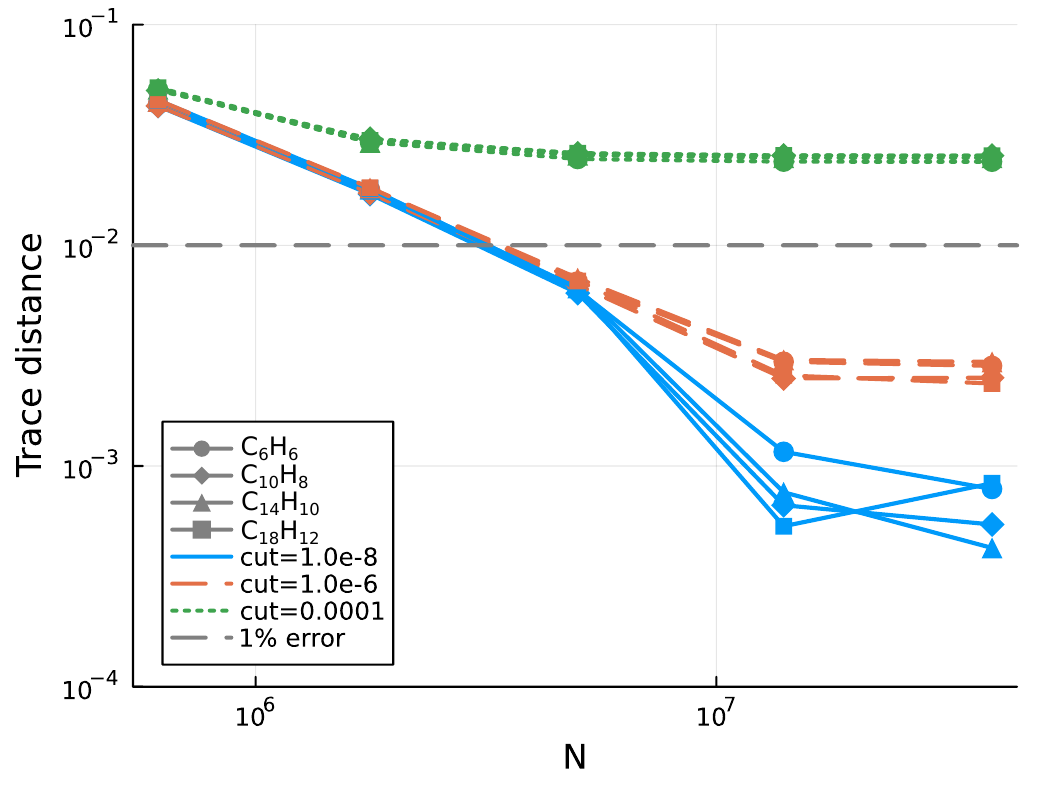}
    \caption{}
    \label{fig:mo_preparation_costs_right}
  \end{subfigure}
  \begin{subfigure}[t]{\textwidth}
    \vspace{3pt}
    \includegraphics[width=.9\textwidth]{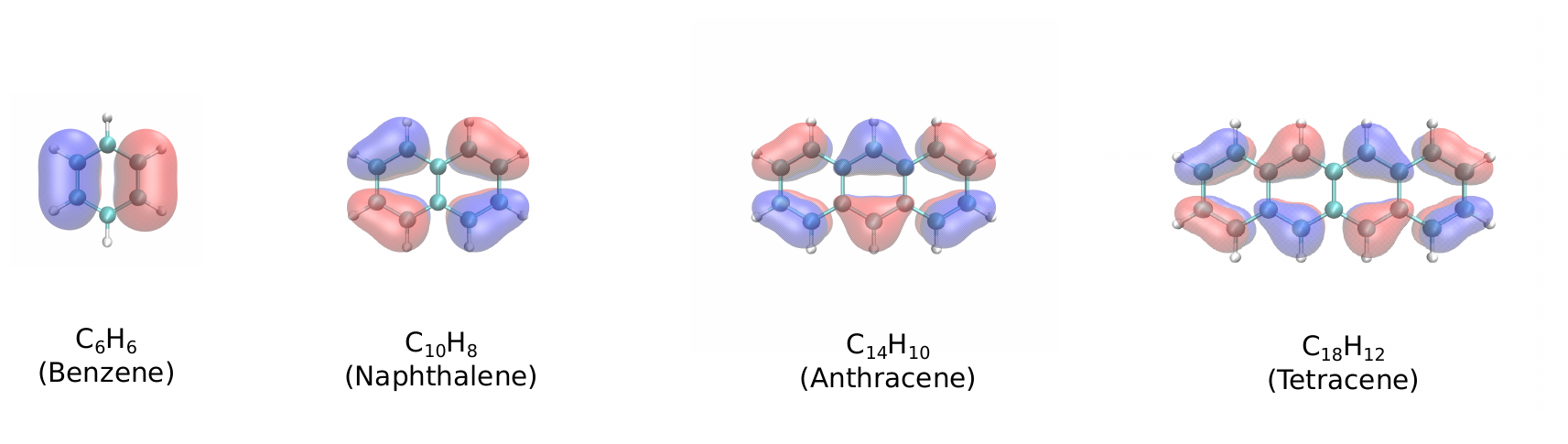}
  \end{subfigure}
  \caption{Toffoli counts (panel a) and error (in terms of the trace distance, panel b) for state preparation of the HOMOs of benzene (circles), naphthalene (diamonds), anthracene (triangles) and tetracene (squares), all in the cc-pVDZ basis set ($N_g=$ 204, 326, 448, and 570 respectively).
    The size of the computational unit cell is fixed by setting $L=60$ Bohr, and we vary the momentum cutoff so that the number of plane waves ranges from $N=6.14\times10^5$ up to $N=3.96\times10^7$.
    For all four molecules, we plot the number of Toffoli gates and the error (in terms of the trace distance) for three different levels of MPS compression, fixing the singular value cutoff to $10^{-8}$ (solid lines), $10^{-6}$ (long dashes), or $10^{-4}$ (short dashes).
  }
  \label{fig:mo_preparation_costs}
\end{figure*}

Examining the relative ordering of the lines in \Cref{fig:mo_preparation_costs_left}, we find that the resources required for state preparation do not monotonically increase as we increase the number of benzene rings.
Focusing on the solid blue lines, which represent the most accurate approximations to the molecular orbitals, we see that the expected number of Toffoli gates required for state preparation increases modestly as we increase in size from C$_6$H$_6$ (circular markers), to C$_{10}$H$_8$ (diamond markers), to C$_{14}$H$_{10}$ (triangular markers).
However, the HOMO of the largest molecule, Tetracene (square markers), is the one that requires the fewest resources for state preparation.
This again suggests that the the linear dependence on \(N_g\) from \Cref{lemma:mo_as_mps} may be too pessimistic. 
Indeed, the only consistent feature of our results with increasing \(N_g\), either by increasing the basis size or the number of atom centers, is a lack of any consistent trend in the bond dimension. 
This provides further evidence that the cost of state preparation is much more strongly dependent on the overall `shape' of the probability amplitudes described by the molecular orbital than on the number of constituent Gaussian functions.
Note that while these results are for the HOMOs, \Cref{fig:benzene_orbitals} suggests that we would expect similar results to hold for more localized orbitals as well.

In \Cref{fig:mo_preparation_costs_right}, we show the estimated error (in terms of the trace distance) of the approximated molecular orbitals with the exactly projected ones.
When a sufficiently large plane wave basis is used, we find that the error drops sharply as we decrease the singular value threshold of the matrix product state approximation.
Examining both the left and right panels together, we find that the error in the molecular orbital can be decreased by almost two orders of magnitude by using roughly three times the number of Toffoli gates for state preparation.
At the bottom of \Cref{fig:mo_preparation_costs}, we illustrate the isosurfaces of the HOMO for the four acene molecules.

\begin{figure*}
  \includegraphics[width=.49\textwidth]{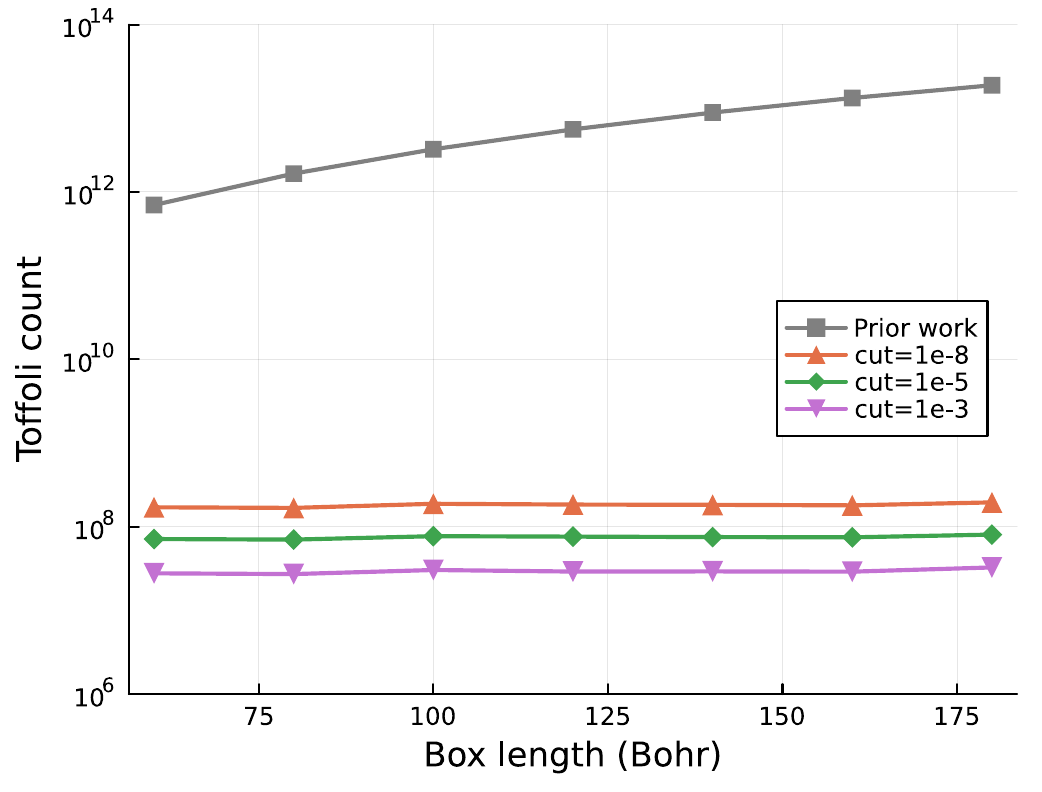}
  \includegraphics[width=.49\textwidth]{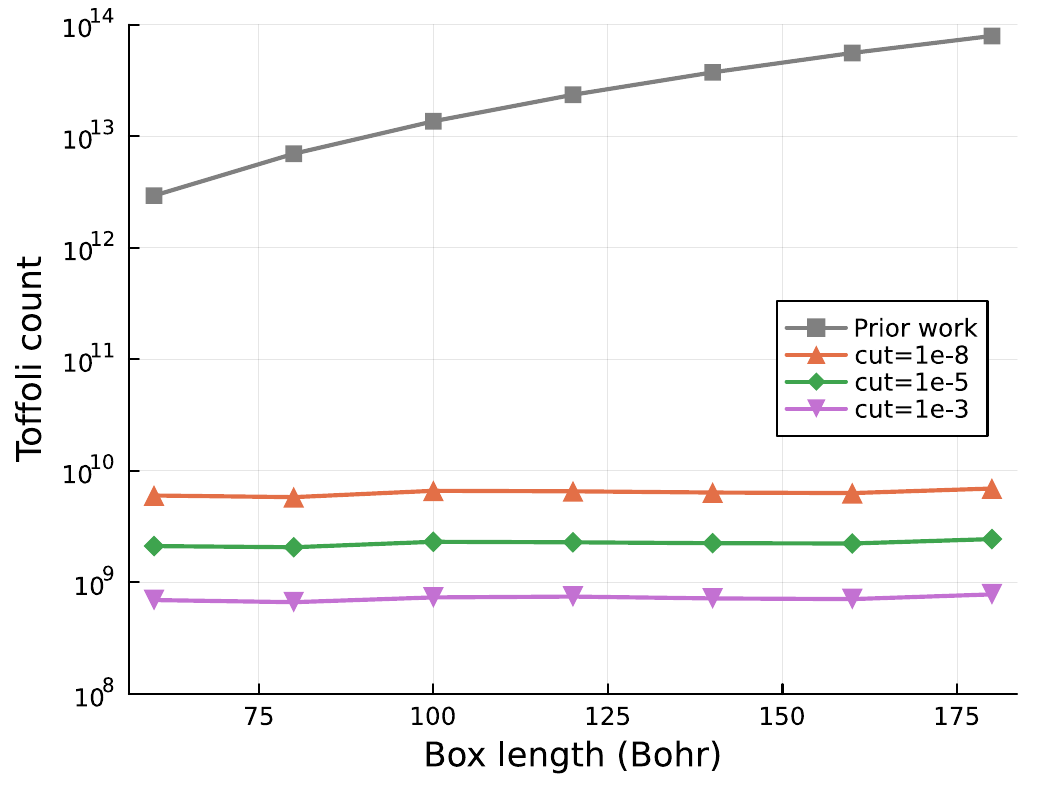}
  \caption{The number of Toffoli gates required to prepare an approximation to the Hartree-Fock state of water (left) and benzene (right).
  We obtain the Hartree-Fock state from calculations in a cc-pVDZ Gaussian basis set.
  We plot the number of Toffoli gates required as a function of the size of the computational unit cell, keeping the momentum cutoff fixed.
  We present data for three different levels of MPS compression for the molecular orbitals, with singular value cutoffs of \(10^{-3}\) (purple triangles pointing downwards), \(10^{-5}\) (green diamonds), and \(10^{-8}\) (orange triangles pointing upwards).
  For both molecules, the number of Toffoli gates required by our approach increases only marginally as the cutoff is lowered and stays nearly constant as the size of the computational unit cell is varied.
  For comparison, we also plot the number of Toffoli gates required for state preparation using the approach of Ref.~\citenum{Babbush2023-ud}.
  Compared with this prior work, our method of state preparation reduces the number of Toffoli gates required by two to five orders of magnitude for the range of parameters we consider.
  }
  \label{fig:overall_toff_count}
\end{figure*}

Finally, in \Cref{fig:overall_toff_count}, we study the overall cost of preparing the Hartree-Fock state of a water molecule (panel a) and a benzene molecule (panel b).
We compare the number of Toffoli gates required by our method with the number of Toffoli gates required using the previous state-of-the-art approach from Ref.~\citenum{Babbush2023-ud} (grey square markers).
For both molecules, we first calculate the occupied molecular orbitals in a cc-pVDZ Gaussian basis set before approximately projecting them onto a plane wave basis and factorizing them.
We plot the number of Toffoli gates as a function of the size of the computational unit cell (for a fixed momentum cutoff), varying \(L\) between \(60\) Bohr and \(180\), which corresponds to a number of plane waves that ranges between \(N = 3.19 \times 10^8\) and \(N =8.6 \times 10^9\). For comparison, the terminal hydrogens of the water molecule are separated by about 3 Bohr, while the terminal hydrogens of the benzene molecule are roughly 10 Bohr apart. 
In order to study the dependence of the overall cost on the accuracy of the individual molecular orbitals, we calculate this data for three different choices of singular value cutoff, \(10^{-3}\) (purple triangles pointing downwards), \(10^{-5}\) (green diamonds), and \(10^{-8}\) (orange triangles pointing upwards).

We provide the details of the Toffoli count analysis for the Slater determinant state preparation in \Cref{app:sd_preparation}, and for the method of Ref.~\citenum{Babbush2023-ud} in \Cref{app:naive_sd_prep}.
The specific Toffoli counts we quote come from \Cref{eq:toff_count_state_prep_mps_sd} and \Cref{eq:toff_count_standard_old_way}.
For both cases, we neglect the cost of preparing the initial antisymmetric state, which is negligible compared with the cost of implementing the other steps.

For the parameters ranges we study in \Cref{fig:overall_toff_count}, we find that our approach to state preparation requires between \(10^2\) and \(10^5\) fewer Toffoli gates than the prior state-of-the-art.
Moreover, unlike the previously proposed approach, the cost of our method is almost entirely independent of the size of the computational unit cell.
The cost of our approach depends on the singular value cutoff, which can be set to achieve a desired level of accuracy.
We find that we can reduce the singular value cutoff by five orders of magnitude at the cost of using at most an order of magnitude more Toffoli gates for state preparation.

In \Cref{tab:errors}, we present approximate upper bounds on the error in the Hartree-Fock states prepared by our method.
As in \Cref{fig:overall_toff_count}, we consider a water molecule and a benzene molecule.
The Hartree-Fock states are obtained from calculations in a cc-pVDZ basis set and we approximately project the occupied molecular orbitals onto a plane wave basis with \(L=60\) and \(N=8.6 \times 10^9\).
We consider three different levels of singular value truncation for the individual molecular orbitals, and we approximate the error in the trace distance as the sum of the errors estimated for each individual occupied molecular orbital (doubled since each orbital is occupied by two electrons), a calculation that we justify in \Cref{app:sd_preparation}.

Even though the individual molecular orbitals are generally accurate at small truncation thresholds, the errors accumulate significantly when we take into account the \(5\) doubly-occupied orbitals of the water molecule, or the \(21\) doubly-occupied orbitals of the benzene molecule.
For some practical applications where small errors in the initial state preparation are required, we may need to use a combination of smaller singular value cutoffs and higher plane wave counts (or pseudopotentials) to increase the accuracy of our approach.
As we have demonstrated in \Cref{fig:overall_toff_count}, it is possible to decrease the singular value cutoff with only a modest increase in the number of Toffoli gates required for state preparation.
Using a larger momentum cutoff is also achievable without a serious increase in cost for our state preparation method (as demonstrated in \Cref{fig:mo_preparation_costs}), although the cost of simulation itself scales less favorably with the momentum cutoff.

\begin{table}
  \centering
  \begin{tabular}{@{}llll@{}}
  \toprule
  Cutoff               & \(10^{-3}\)   & \(10^{-5}\)   & \(10^{-8}\)   \\ \midrule
  H\(_2\)O       & 0.851 & 0.125 & 0.057   \\
  C\(_6\)H\(_6\) & 3.757 & 0.510 & 0.171   \\ \bottomrule
  \end{tabular}%
  \caption{Approximate upper bounds for the error (in the trace distance) for the Hartree-Fock state preparation of water and benzene (\(L=60\), \(N=8.6 \times 10^9\). We estimate the error at three different levels of matrix product state compression, with singular value cutoffs set to \(10^{-3},\) \(10^{-5}\), and \(10^{-8}\). 
  While modest cutoffs may suffice for some applications (e.g., approximating ground state energy using phase estimation), smaller cutoffs, higher plane wave counts, or additional improvements (such as the use of pseudopotentials) may be necessary for quantitative accuracy in other cases.}
  \label{tab:errors}
\end{table}

\section{Discussion and conclusion}
\label{sec:discussion}

In this paper, we introduced a new approach for preparing the initial states required for first-quantized quantum simulations of molecular systems in a plane wave basis.
We demonstrated how to implement a single-particle basis transformation from a Gaussian-type orbital basis to a plane wave basis using a number of quantum gates that scales logarithmically with the size of the plane wave basis.
Our approach enables the approximate preparation of Hartree-Fock states using exponentially fewer quantum gates (in terms of the number of plane waves) than was possible with prior approaches.
It also makes it possible to use any efficiently preparable second-quantized wavefunction in a molecular orbital basis as an initial state for a first-quantized plane wave calculation, without incurring a cost that is linear in the number of plane waves.
This reduction is particularly significant because the state-of-the-art algorithms for time evolution in this context also have a cost that scales sublinearly with the basis set size, making state preparation a potential bottleneck in computational cost.

Our approach exploits the observation that it is possible to approximately represent molecular orbitals in a plane wave basis as matrix product states (or tensor trains) with very low bond dimension.
Our proof of this statement begins by showing that the coefficients of a primitive Gaussian basis function in a finite plane wave basis can be well-approximated by low degree polynomials.
This implies that the first-quantized wavefunctions corresponding to these primitive basis functions have efficient matrix product state descriptions.
We can therefore use standard tools for manipulating matrix product states and converting them into quantum circuits to obtain efficient circuits for preparing states that encode the molecular orbitals.
Using these circuits, it is straightforward to implement the desired change of basis.

In a series of numerical experiments on small molecular systems, we demonstrate that our approach is highly efficient in practice.
We find that the bond dimension required to represent a typical molecular orbital depends only weakly on the choice of the original Gaussian basis set, the size of the target plane wave basis set, and the accuracy parameters of the matrix product state compression.
Furthermore, the classical computational resources required by our approach are modest.
We compare the overall number of non-Clifford gates required to prepare the Hartree-Fock states of water and benzene using our approach with the number required by the prior state-of-the-art from Ref.~\citenum{Babbush2023-ud}.
Even for these small systems, we show that we can reduce the cost by several orders of magnitude for reasonable choices of parameters at the expense of introducing some error into the initial state preparation.

It is likely that the logarithmic scaling of our approach (with the size of the plane wave basis) could be obtained by other means, such as a modification of the state preparation works of Zalka in Ref.~\citenum{Zalka1998-zt} and Ward \textit{et al.} in Ref.~\citenum{Ward2009-ln} (generically known as ``Grover-Rudolph'' schemes).
However, unlike our approach, such a scheme would have high overheads stemming from the need to coherently perform a large number of integrals and would scale explicitly with the size of the underlying Gaussian basis sets.
By contrast, we have demonstrated that the overhead of our state preparation method is low in practice.
Furthermore, while the size of the original Gaussian basis set enters polynomially into our analytical upper bounds, empirically we find that the cost of our approach can be relatively independent of this quantity.
This highlights a strength of matrix product state and tensor train factorizations as a tool for quantum algorithm construction; the fact that they can be efficiently compressed numerically can allow for better-than-expected performance in certain situations.

As we have demonstrated, the flexibility and ease of use offered by matrix product state methods makes them an appealing tool for state preparation in the context of first-quantized plane wave algorithms for molecular systems.
We expect that they will also prove useful in a variety of other contexts.
Most immediately, it would be interesting to see if they can be applied to prepare initial states for periodic systems.
Although the details of our analytical results made heavy use of the properties of a Gaussian type basis set, we expect that the same intuition would apply to, e.g., maximally localized Wannier functions~\cite{Marzari2011-vj}.
One could also imagine directly solving the Hartree-Fock equations in a plane wave basis using a tensor train ansatz for the individual molecular orbitals, as Ref.~\citenum{Khoromskij2011-mn} did in a real space representation.
More generally, many quantum algorithms rely on efficiently loading or reading out classical data.
Matrix product state and tensor train factorizations of functions are a powerful and underexplored tool for these tasks.

The quantum simulation of real molecules and materials is one of the most highly anticipated applications of quantum computing, and first-quantized approaches to this problem promise to scale well in the asymptotic limit.
We hope that our results will help make the first-quantized quantum simulation of large molecular systems more practical, particularly in regimes where the state preparation costs might dominate.
We anticipate that the improved scaling of our proposed method will be especially beneficial when it is desirable to increase the size of the computational unit cell.
Our approach to state preparation and the state-of-the-art methods for Hamiltonian simulation both scale logarithmically with the basis set size in this limit, so our method allows for an exponential improvement in the end-to-end complexity of such a simulation.
This limit is especially relevant for molecular simulations, where it is often necessary to use a large computational unit cell in order to  minimize the interaction of an aperiodic system with its periodic images (which are artificially introduced due by the underlying periodicity of a plane wave basis set).

\textit{Note added:}
The application of matrix product state and tensor train techniques to various problems in quantum computing and quantum chemistry is an active area of research.
In the time between the initial release of our manuscript and its publication, a number of interesting related works have appeared.
We highlight a few of these here.
Ref.~\citenum{Rodriguez-Aldavero2024-lo} explores the use of Chebyshev interpolation for construction MPS/QTT approximations to multivariate functions in a more general context.
Ref.~\citenum{Berry2024-qe} proposed a variation of the matrix product state preparation techniques of Ref.~\citenum{Fomichev2023-vs} with significantly reduced constant factors. We do not incorporate their optimizations here, but we expect they would further reduce the costs of our approach.

\section*{Acknowledgments}
We thank Ryan Babbush, Nick Rubin, Elliot Eklund, Joonho Lee, Joshua Goings, and Subhayan Roy Moulik for helpful discussions.
O.L. and K.B.W. acknowledge support from the NSF QLCI program through grant number QMA-2016345 and from Dow.
T.F.S and K.B.W acknowledge support from a Google academic research award. T.F.S. was
supported as a Quantum Postdoctoral Fellow at the Simons Institute for the Theory of Computing, supported
by the U.S. Department of Energy, Office of Science, National Quantum Information Science Research Centers,
Quantum Systems Accelerator.

\section*{Code Availability}
Code which implements the numerical experiments of \Cref{sec:MO_MPS_numerics} is available under an open source license at \url{https://www.github.com/oskar-leimkuhler/orb2mps-fq}.
See \Cref{app:computational_methods} for more details.

%

\appendix
\widetext

\section{The size of the computational unit cell}
\label{app:box_size_considerations}

Our goal in this work is to enable the more efficient treatment of aperiodic molecular systems using quantum algorithms that are particularly efficient in a periodic plane wave basis.
We do so by mapping single-particle states in a Gaussian basis set to a plane wave basis set.
In order to simplify our analysis, we make the assumption that the Gaussian basis functions have negligible support outside of the computational unit cell.
In this appendix, we briefly explain why we believe that this assumption will be satisfied by any reasonable choice for the computational unit cell.

In practice, we expect that a computational unit cell that is significantly larger than the extent of the molecular wavefunction will be required for accurate calculations in a plane wave basis.
We expect the electron density of any bound state wavefunction to decay exponentially in the continuum limit~\cite{Katriel1980-xe}.
Furthermore, any linear combination of a fixed set of Gaussians will exhibit a Gaussian decay profile at long range.

In contrast, the strength of the Coulomb interaction decays with a power law.
Although we intend to represent an aperiodic system, by using a periodic single-particle basis we enforce periodic boundary conditions.
Because of these periodic boundary conditions, our molecular system will experience unphysical interactions with its periodic images.
Because the Coulomb interaction is fundamentally long range, these interactions can introduce complications at much larger length scales than the actual size of the system.

Naively, one could choose the size of the computational unit cell to be so large that the quantities of interest are well converged, or can be extrapolated effectively.
Ref.~\citenum{Makov1995-gk} showed that error in the ground state energy that arises from a finite cubic computational unit cell be made to decay as \(\mathcal{O}(L^{-5})\) by properly setting up the calculation.
One could combine a numerical analysis of the convergence with an extrapolation technique in order to accurately estimate the ground state energy in the \(L=\infty\) limit~\cite{Makov1995-gk, Martyna1999-av, Fusti-Molnar2002-ch}.
There is less prior work to draw on for estimating other quantities (although many static properties can be obtained by taking derivatives of the ground state energy with respect to parameters of the Hamiltonian), but one might hope that similarly quick convergence could be obtained.
Analyzing these approaches is outside of the scope of this work, but they would likely demand significantly larger values of \(L\) than we require for our assumptions to hold.

Alternatively, one could attempt to modify the construction of the first-quantized simulation algorithms for the electronic structure Hamiltonian in order to rigorously truncate the range of the Coulomb operator, following Ref.~\citenum{Fusti-Molnar2002-ch} or related works.
This possibility was alluded to in Ref.~\citenum{Su2021-uj}, but the necessary modifications have never been explored in detail.
This kind of approach relies on the assumption that the computational unit cell is roughly twice as large (in each direction) as the spatial extent of the system, so it would place a more stringent constraint on \(L\) than anything required by our technical results.

\section{Proof of \Cref{thm:main_technical_result} (\nameref*{thm:main_technical_result})}
\label{app:primitive_gaussian_mps_lemma_proof}

In this section, we prove \Cref{thm:main_technical_result}. The version of the lemma in the main text uses asymptotic notation to discuss the constraints on \(K\) and \(m\), whereas the version we state below provides more concrete inequalities.
\onedmpstheorem*

To begin, it will be useful to rewrite \(\tilde{g}_x\) in a different form from that presented in \Cref{eq:tilde_g_x_def}.
Specifically, it will be convenient to rewrite \(g_x\) in terms of Hermite-Gaussian functions, defined as
\begin{equation}
  \psi_n(x) = c_n  e^{-x^2/2} H_n(x),
  \label{eq:hermite_gaussian_def}
\end{equation}
and analytically evaluate the integrals that define \(\tilde{g}_x\).
Here \(H_n(x)\) are the so-called physicist's Hermite polynomials,
\begin{equation}
  H_n(x) = (-1)^{n} e^{x^2} \frac{d^n}{dx^n} e^{x^2},
  \label{eq:hermite_polynomial_def}
\end{equation}
and \(c_n\) are coefficients chosen so that the integral of \(|\psi_n(x)|^2\) over the whole real line is one,
\begin{equation}
  c_n = \left( 2^n n!
  \sqrt{\pi} \right)^{-1/2}.
  \label{eq:hermite_gaussian_c_def}
\end{equation}
The Hermite-Gaussian functions have a number of useful properties.
For one, they form an orthonormal basis for \(L^2(\mathbb{R})\), the set of square integrable functions on the real line.
They are also eigenfunctions of the Fourier transform.

After a change of variables to \(u = \sqrt{2 \gamma} x\), we can represent \(g_x\) as a linear combination of the Hermite-Gaussian functions for \(n \in \left[ 0..l \right]\),
\begin{equation}
  g_x(x)
  =
  g_x(\frac{u}{\sqrt{2 \gamma}})
  =
  \frac{c_{l, \gamma}}{\left( 2\gamma \right)^{l/2}} u^l e^{-u^2 / 2} = (2 \gamma)^{1/4} \sum_{n \in \left[ 0..l \right]} h_n \psi_n(u).
\end{equation}
In the implicit definition of the coefficients \(h_n\), we choose the normalization factor so that we have
\begin{align}
  \int_{x \in \mathbb{R}} |g_x(x)|^2 dx
   & =
  \frac{1}{\sqrt{2 \gamma}} \int_{u \in \mathbb{R}} \left|g_x(\frac{u}{\sqrt{2\gamma}})\right|^2 du
  \nonumber
  \\
   & =
  \int_{u \in \mathbb{R}} \left| \sum_{n \in \left[ 0..n \right]} h_n \psi_n(u)\right|^2 du
  \nonumber
  \\
   & =
  \sum_{n \in \left[ 0..n \right]} \sum_{m \in \left[ 0..n \right]} h_n^* h_m \int_{u \in \mathbb{R}} \psi_n(u)^* \psi_m(u) du
  \nonumber
  \\
   & =
  \sum_{n \in \left[ 0..n \right]} |h_n|^2.
  \label{eq:h_n_normalization}
\end{align}
Since \(g_x\) is \(L^2\)-normalized over the whole real line, the vector whose entries are \(h_n\) is also normalized (in the \(l^2\) norm).
The fact that the Hermite-Gaussian functions are eigenfunctions of the Fourier transform makes it simple to evaluate
\begin{align}
  \int_{x=-\infty}^{\infty} \phi_k^*(x) g_x(x) \phi_k(x) dx
   & =
  \frac{1}{\left(2 \gamma\right)^{1/4} } \sum_{n \in \left[ 0 .. l \right]} h_n \int_{u=-\infty}^{\infty} \phi^*_k\left(\frac{u}{\sqrt{2 \gamma}}\right)  \psi_n(u) du
  \nonumber
  \\
   & =
  \frac{2^{1/4}\sqrt{\pi}}{\gamma^{1/4}\sqrt{L}} \sum_{n \in \left[ 0 .. l \right]} \left(-i\right)^{n} h_n \psi_n\left(\frac{k}{\sqrt{2 \gamma}}\right).
\end{align}
Therefore, we can represent \(\tilde{g}_x\) as
\begin{equation}
  \tilde{g}_x(x)
  =
  \frac{2^{1/4}\sqrt{\pi}}{\widetilde{\mathcal{N}}_x\gamma^{1/4}\sqrt{L}} \sum_{k \in \mathbb{K}} \sum_{n \in \left[ 0 .. l \right]} \left(-i\right)^{n} h_n \psi_n\left(\frac{k}{\sqrt{2 \gamma}}\right) \phi_k(x).
  \label{eq:tilde_g_hermite_gaussian_def}
\end{equation}

Recall the definition \(\mathbb{K}_{cut} = \left\{ k \in \mathbb{K} : k \leq K \right\}\), and let
\begin{equation}
  g_{trunc}(x)
  =
  \frac{2^{1/4}\sqrt{\pi}}{\widetilde{\mathcal{N}}_x\mathcal{N}_t\gamma^{1/4}\sqrt{L}} \sum_{k \in \mathbb{K}_{cut}} \sum_{n \in \left[ 0 .. l \right]} \left(-i\right)^{n} h_n \psi_n\left(\frac{k}{\sqrt{2 \gamma}}\right) \phi_k(x),
  \label{eq:g_trunc_def}
\end{equation}
where \(\mathcal{N}_t \in \mathbb{R}_{> 0}\) is implicitly defined by demanding that \(\norm{g_{trunc}}=1\).\footnote{
Recall that we use the notation \(\norm{f}\) to denote \(\sqrt{\inp{f}{f}}\), where \(\inp{g}{f}\) denotes \(\int_{x=-L/2}^{L/2} g^*(x) f(x) dx\).}
Informally, we claim that the trace distance between \(\tilde{g}_x\) and \(g_{trunc}\) is negligible for a sufficiently large cutoff \(K\) because \(g_x\) does not have significant support on plane waves with very high frequencies.
This is captured formally in \Cref{lemma:momentum_cutoff}, which we prove in \Cref{sec:momentum_cutoff_proof}.
\begin{restatable}{lemma}{cutofflemma}
  Consider \(\tilde{g}_x\) as defined in \Cref{eq:tilde_g_hermite_gaussian_def} and \(g_{trunc}\) as defined in \Cref{eq:g_trunc_def}.

  For any \(K \in \mathbb{R}_{>0}\) and \(\epsilon_t \in \left( 0, 1 \right]\) that satisfy
  \begin{equation}
    K
    \geq
    2\sqrt{2 \gamma}\sqrt{
      2 \log \left( \epsilon_t^{-1} \right)
      + \log \left( 20 \right)
      + \log \left( 1 + \frac{2 \sqrt{\pi}}{L\sqrt{\gamma}} \right)
      - \log \left( \widetilde{\mathcal{N}}_x^2 \right)
      + l \log(4 l)
    }
    .
    \label{K_condition}
  \end{equation}
  then
  \begin{equation}
    D(\tilde{g}_x, g_{trunc}) = \sqrt{1 - |\inp{\tilde{g}_x}{g_{trunc}}|^2} \leq \epsilon_{t}
    \label{eq:g_whole_g_trunc_distance}
  \end{equation}
  and
  \begin{equation}
    1 - \epsilon_t \leq \mathcal{N}_t \leq 1.
  \end{equation}
  \label{lemma:momentum_cutoff}
\end{restatable}

The trace distance between \(\tilde{g}_x\) and \(g_{trunc}\) is one of two components of our overall error.
In order to bound the overall error by \(\epsilon\), we seek to bound this component by \(\epsilon/2\) by setting \(\epsilon_t = \epsilon/2\).
The resulting condition on \(K\) is that
\begin{equation}
  K
  \geq
  2\sqrt{2 \gamma}\sqrt{
    2 \log \left( 2 \epsilon^{-1} \right)
    + \log \left( 20 \right)
    + \log \left( 1 + \frac{2 \sqrt{\pi}}{L\sqrt{\gamma}} \right)
    - \log \left( \widetilde{\mathcal{N}}_x^2 \right)
    + l \log(4 l)
  }
  .
  \label{eq:K_epsilon_inequality}
\end{equation}
Observe that this inequality is implied by the combination of our assumption on \(K\) from the statement of the lemma together with our assumption that \(\widetilde{\mathcal{N}}_x \geq 2/3\).
Therefore, we have that
\begin{equation}
  D(\tilde{g}_x, g_{trunc}) \leq \frac{\epsilon}{2}
  \label{eq:g_whole_trunc_bound_proof}
\end{equation}
and
\begin{equation}
  1 - \frac{\epsilon}{2} \leq \mathcal{N}_t \leq 1.
\end{equation}

Our next goal is to approximate \(g_{trunc}\) by a normalized function of the form
\begin{equation}
  f_{x}(x) = \sum_{k \in \mathbb{K}_{cut}} p(k) \phi_k(x),
\end{equation}
where the function \(p\) is a polynomial of some reasonable degree.
Once this is accomplished, we will be able to use existing machinery for encoding tensor train representations of polynomials to complete the proof.
To this end, we prove the following lemma in \Cref{sec:polynomial_approximation_proof},
\begin{restatable}{lemma}{polynomiallemma}
  Consider \(g_{trunc}\) defined as in \Cref{eq:g_trunc_def}.
  For an arbitrary \(\epsilon_{p} \in \left( 0, 1 \right)\), \(m \in \mathbb{Z}\) that satisfy the following conditions
  \begin{align}
    m \geq &
    2,
    \nonumber
    \\
    m \geq &
    \frac{eK}{2\sqrt{\gamma}} \left( \frac{eK}{2\sqrt{\gamma}} + \sqrt{2l + 1} \right),
    \label{eq:m_conditions}
    \\
    m \geq &
    \frac{ 4 \log \left( \epsilon_{p}^{-1}\right) - 2\log \left( \widetilde{\mathcal{N}}_x \right) - 2\log \left(\mathcal{N}_t\right) + \log(\frac{K}{\sqrt{\gamma}} + \frac{\pi}{L \sqrt{\gamma}}) + \log \left( l + 1 \right)}{\log 2} + 4.5,
    \nonumber
  \end{align}
  there exists a degree \(m-1\) polynomial \(p\) such that the following holds:

  Let
  \begin{equation}
    f_{x}(x) = \sum_{k \in \mathbb{K}_{cut}} p(k) \phi_k(x).
  \end{equation}
  Then
  \begin{equation}
    \norm{f_{x}}^2 = 1,
  \end{equation}
  and
  \begin{equation}
    D(g_{trunc}, f_{x}) \leq \epsilon_p.
  \end{equation}
  \label{lemma:polynomial_approximation}
\end{restatable}

Recall the condition on \(K\) from the statement of the lemma,
\begin{equation}
  K
  \geq
  2\sqrt{2 \gamma}\sqrt{
    2 \log \left( 2\epsilon^{-1} \right)
    + \log \left( 45 \right)
    + \log \left( 1 + \frac{2 \sqrt{\pi}}{L\sqrt{\gamma}} \right)
    + l \log(4 l)
  }
  \label{eq:K_condition_recap}
  .
\end{equation}
We can apply \Cref{lemma:polynomial_approximation}, setting \(\epsilon_{p} = \epsilon/2\) and use this inequality to simplify the conditions on \(m\).
By inspection, \(\frac{eK}{2\sqrt{\gamma}} \gg 1\).
Therefore, the first condition is redundant because the second condition already implies that \(m\geq2\).
Examining the second condition, we can see from \Cref{eq:K_condition_recap} that the inequality \(\frac{eK}{2 \sqrt{\gamma}} > \sqrt{2l + 1}\) is manifestly true, so for our purposes we can simplify the second condition to the more stringent requirement that
\begin{equation}
  m \geq \frac{e^2 K^2}{2 \gamma}.
  \label{eq:m_condition_recapped}
\end{equation}
Now we need to show that this inequality implies the third condition as well.
Rewriting \Cref{eq:m_condition_recapped} as the sum of two equal terms and substituting in the smallest possible value of \(K\) from \Cref{eq:K_condition_recap}, we see that \Cref{eq:m_condition_recapped} implies that
\begin{equation}
  m \geq 2e^2 \left(   2 \log \left( 2\epsilon^{-1} \right)
  + \log \left( 45 \right)
  + \log \left( 1 + \frac{2 \sqrt{\pi}}{L\sqrt{\gamma}} \right)
  + l \log(4 l) \right) + \frac{e^2 K^2}{4 \gamma}.
  \label{eq:m_second_condition_explicit}
\end{equation}

Now we would like to show that any \(m\) that satisfies \Cref{eq:m_second_condition_explicit} also satisfies the third condition.
Taking \(\epsilon_p\) to be \(\epsilon/2\) and making use of the assumption that \(\widetilde{\mathcal{N}}_x \geq 2/3\), the fact that \(\mathcal{N}_t \geq 1 - \epsilon/2\), and the assumption that \(0 < \epsilon < 1\), we can write a more stringent version of the third condition as 
\begin{equation}
    m \geq 
    \frac{ 4 \log \left( 2\epsilon^{-1}\right) + 2\log \left( 3/2 \right) + 2\log \left(2\right) + \log(\frac{K}{\sqrt{\gamma}} + \frac{\pi}{L \sqrt{\gamma}}) + \log \left( l + 1 \right)}{\log 2} + 4.5.
\end{equation}
By inspection of \Cref{eq:K_condition_recap}, \(\frac{K}{\sqrt{\gamma}} > 4\), so we can use the fact that \(ab > a + b\) for \(a, b > 2\) to establish that 
\begin{equation}
\log(\frac{K}{\sqrt{\gamma}} + \frac{\pi}{L \sqrt{\gamma}}) = 
\log\left( \left(\frac{K}{\sqrt{\gamma}} - 2 \right) + \left( 2 + \frac{\pi}{L \sqrt{\gamma}} \right)\right) <
\log\left( \frac{K}{\sqrt{\gamma}} - 2 \right) +  \log \left( 2 + \frac{\pi}{L \sqrt{\gamma}} \right).
\end{equation}
Therefore, with a bit of algebraic manipulation, we can write an even more stringent version of the third condition as
\begin{equation}
    m \geq 
    \frac{ 4 \log \left( 2\epsilon^{-1}\right) + 2\log \left( 3\right) + \log\left( \frac{K}{\sqrt{\gamma}} - 2 \right) +  \log \left( 1 + \frac{\pi}{2 L \sqrt{\gamma}} \right) + \log \left( l + 1 \right)}{\log 2} + 5.5.
    \label{eq:third_m_condition_tweaked}
\end{equation}
By inspection, we have the following inequalities,
\begin{align}
  4 e^2 \log(2 \epsilon^{-1}) & > \frac{4 \log(2 \epsilon^{-1})}{\log(2)}
  \nonumber \\
  2e^2 \log(45) & > \frac{2 \log(3)}{\log (2)}  + 5.5
  \nonumber \\
  2e^2 \log \left( 1 + \frac{2 \sqrt{\pi}}{L \sqrt{\gamma}} \right) &> \frac{\log\left( 1 + \frac{\pi}{2 L \sqrt{\gamma}} \right)}{\log \left( 2 \right)}
  \nonumber \\
  2e^2 l \log(4l) \geq \frac{\log(l + 1)}{\log(2)}
  \nonumber \\
  \frac{e^2 K^2}{4 \gamma} &> \frac{\log \left( \frac{K}{\sqrt{\gamma}} - 2\right)}{\log(2)},
\end{align}
with the last inequality following from the fact that \(\frac{K}{\sqrt{\gamma}}\) is large.
Therefore, we see that \Cref{eq:m_second_condition_explicit} implies that \Cref{eq:third_m_condition_tweaked} is satisfied.

We have therefore shown that any \(m\) satisfying \Cref{eq:m_condition_recapped} satisfies all three requirements from the statement of \Cref{lemma:polynomial_approximation} (with \(\epsilon_p\) set to \(\epsilon/2\)).
By the assumption in the statement of \Cref{lemma:mo_as_mps}, \(m\) satisfies \Cref{eq:m_condition_recapped}.
Thus, \Cref{lemma:polynomial_approximation} guarantees that there exists a degree \(m-1\) polynomial \(p\) such that
\begin{equation}
  f_{x}(x) = \sum_{k \in \mathbb{K}_{cut}} p(k) \phi_k(x)
\end{equation}
satisfies
\begin{equation}
  \norm{f_{x}}^2 = 1,
\end{equation}
and
\begin{equation}
  D(g_{trunc}, f_{x}) \leq \epsilon/2.
  \label{eq:g_trunc_poly_bound_proof}
\end{equation}
Applying the triangle inequality to \Cref{eq:g_whole_trunc_bound_proof} and \Cref{eq:g_trunc_poly_bound_proof}, we have that
\begin{equation}
  D(\tilde{g}_{x}, f_{x}) \leq \epsilon.
\end{equation}

To complete our proof, we will use the following result, which we recall and rephrase from Ref.~\citenum{Grasedyck2010-pn}.
\begin{proposition}[Adapted from Corollary 28 of Ref.~\citenum{Grasedyck2010-pn}]
  Let \(p: \mathbb{R} \rightarrow \mathbb{C}\) be a polynomial of degree at most \(d\).\footnote{
    \scriptsize
    Note that the statement of Corollary 28 and the related theorems in Ref.~\citenum{Grasedyck2010-pn} are made in terms of polynomials with real coefficients and vectors with real entries, but we claim that the same arguments hold for the complex case.
    If we were to hold ourselves strictly to the text of the corollary as proved, we could obtain a similar result with a doubled bond dimension by writing \(p(x) = a(x) + i b(x)\) for two real-valued polynomials and adding the resulting tensor trains.
  }
  For an arbitrary \(a, b \in \mathbb{R}\) with \(a<b\), \(N \in \mathbb{Z}\) with \(N > 1\), define the vector \(x\) by
  \begin{equation}
    x_i = a + \frac{b-a}{N-1}i, \;\;\; i \in \left[ 0..N-1 \right].
  \end{equation}
  For any \(n \in \mathbb{Z}\) such that \(n \geq \log_2 N\), let \(T\) be the rank-\(n\) tensorization of \(p\) over the grid specified by \(x\), as defined in \Cref{def:tensorization}.
  In other words, let \(T\) be defined by
  \begin{equation}
    T^{s_1s_2\cdots s_n} =
    \begin{cases}
      p(x_i) & \text{ for } i < N
      \\
      0      & \text{ for } i \geq N
    \end{cases}
    ,
    \;\;\; i = \sum_{j=1}^n 2^{n-j} s_j.
  \end{equation}
  Then \(T\) can be represented as a tensor train with bond dimension at most \(d+2\).
  \label{prop:mps_dimension}
\end{proposition}

The standard representation for first-quantized simulation indexes the momentum eigenstates using signed integers rather than unsigned integers, so \Cref{prop:mps_dimension} does not apply directly.
However, it is straightforward to prove the following corollary.
\begin{corollary}[Signed tensor train representation of a polynomial]
  Let \(p: \mathbb{R} \rightarrow \mathbb{C}\) be a polynomial of degree at most \(d\).
  For an arbitrary \(a \in \mathbb{R}\) with \(a>0\), \(N \in \mathbb{Z}\) with \(N > 1\) and \(N\) odd, define the vector \(x\) by
  \begin{equation}
    x_i = \frac{2a}{N-1}i, \;\;\; i \in \left[ -\frac{N-1}{2}..\frac{N-1}{2} \right].
  \end{equation}
  For any \(n \in \mathbb{Z}\) such that \(n \geq \log_2 N\), let \(T\) be defined by
  \begin{equation}
    T^{s_1s_2\cdots s_n} =
    \begin{cases}
      p(x_i) & \text{ for } i \in \left[ -\frac{N-1}{2}..\frac{N-1}{2} \right]
      \\
      0      & \text{ for } i \notin \left[ -\frac{N-1}{2}..\frac{N-1}{2} \right]
    \end{cases}
    ,
    \;\;\; i = \left( -1 \right)^{s_1} \sum_{j=2}^n 2^{n-j} s_j,
    \label{eq:signed_tensor_T_def}
  \end{equation}
  except that \(T^{1 0 \cdots 0} = 0\) (so that we don't encode ``\(-0\)'').
  Then \(T\) can be represented as a tensor train with bond dimension at most \(2d+5\).

  \label{cor:signed_mps_dimension}
\end{corollary}

\begin{proof}
  Let \(T_{\geq 0}\) be the tensor whose entries are identical to \(T\) as defined in \Cref{eq:signed_tensor_T_def} when the index \(i\) is non-negative, and zero otherwise.
  We can invoke \Cref{prop:mps_dimension} to guarantee the existence of a tensor train for a rank \(n-1\) tensor that encodes \(p\) on the appropriate grid over the interval \(\left[ 0, a \right]\) (corresponding to the non-negative values of \(i\) in \Cref{eq:signed_tensor_T_def}).
  This tensor train will have a bond dimension of at most \(d + 2\).
  The tensor product of this tensor train with \(\left( 1, 0 \right)\) yields a tensor train for \(T_{\geq 0}\) with the same bond dimension bounds.

  We also define \(T_{< 0}\) in a similar fashion, letting the entries be identical to \(T\) when the index \(i\) is negative, and zero otherwise.
  Invoking \Cref{prop:mps_dimension}, we can guarantee the existence of a tensor train that encodes \(p(-x)\) on a grid over \(\left[ 0, a\right]\).
  As above, the tensor product of \(\left( 0, 1 \right)\) with this tensor train yields a new tensor train with bond dimension at most \(d + 2\).
  Let us denote the tensor encoded by this tensor train as \(T_{\leq 0}\).
  By construction, every entry of \(T_{\leq 0}\) is equal to the corresponding entry of \(T_{<0}\), except that \(T_{\leq 0 }^{10\cdots 0} = p(0)\) instead of \(0\).

  Let \(T_{-0}\) be the tensor with the same shape as \(T\) and a single nonzero element, \(T^{10\cdots0} = -p(0)\).
  Because \(T_{-0}\) contains a single non-zero entry, the maximum rank of any matricization of \(T\) is one, and therefore there exists a bond dimension one tensor train representation of \(T_{-0}\).
  By construction, we have that \(T = T_{\geq 0} + T_{\leq 0} + T_{-0}\).
  We can therefore obtain a tensor train for \(T\) whose bond dimension is bounded by the sum of the bond dimensions of the tensor trains for \(T_{\geq 0}\), \(T_{\leq 0}\), and \(T_{-0}\), completing the proof.
\end{proof}

Because the coefficients of \(f_{x}\) with respect to the basis \(\left\{ \phi_k : k \in \mathbb{K}_{cut}  \right\}\) are given by the polynomial \(p\), \Cref{cor:signed_mps_dimension} tells us that we can encode them using a (signed) tensor train with bond dimension at most \(2m+3\).
Consider the matrix product state
\begin{equation}
  \ket{f_x} =
  \sum_{\left\{ s \right\}} \sum_{\left\{ \alpha \right\}} A^{s_1}_{\alpha_1} A^{s_2}_{\alpha_1 \alpha_2} \cdots A^{s_n}_{\alpha_{n-1}} \ket{s_1 s_2 \cdots s_n},
\end{equation}
where the \(A\) tensors are the components of the signed tensor train encoding \(p\).
Note that the computational basis states \(\ket{s_1s_2 \cdots s_n}\) are binary encodings of the signed integers (\(p = \frac{kL}{2 \pi}\)) indexing the momenta \(k \in \mathbb{K}_{cut}\),
\begin{equation}
  p \in \left[ -\frac{\left| \mathbb{K}_{cut} \right| - 1}{2}, \frac{\left| \mathbb{K}_{cut} \right| - 1}{2} \right] = \left[ -\frac{N - 1}{2}, \frac{N - 1}{2} \right] \subset \mathbb{Z}.
\end{equation}
Therefore, we have
\begin{equation}
  \ket{f_x} = \sum_{k \in \mathbb{K}_{cut}} p(k) \ket{\frac{kL}{2\pi}}
\end{equation}
as desired, completing the proof.

\subsection{Proof of \Cref{lemma:momentum_cutoff}}
\label{sec:momentum_cutoff_proof}

Let us recall the statement of the lemma,
\cutofflemma*

Let us recall the definitions of \(\ket{\tilde{g}_x}\) and \(\ket{g_{trunc}}\) from \Cref{eq:tilde_g_hermite_gaussian_def} and \Cref{eq:g_trunc_def},
\begin{align}
  \tilde{g}_x(x)
  = &
  \frac{2^{1/4}\sqrt{\pi}}{\widetilde{\mathcal{N}}_x\gamma^{1/4}\sqrt{L}} \sum_{k \in \mathbb{K}} \sum_{n \in \left[ 0 .. l \right]} \left(-i\right)^{n} h_n \psi_n\left(\frac{k}{\sqrt{2 \gamma}}\right) \phi_k(x),
  \label{eq:g_whole_def_in_proof}
  \\
  g_{trunc}(x)
  = &
  \frac{2^{1/4}\sqrt{\pi}}{\widetilde{\mathcal{N}}_x\mathcal{N}_t\gamma^{1/4}\sqrt{L}} \sum_{k \in \mathbb{K}_{cut}} \sum_{n \in \left[ 0 .. l \right]} \left(-i\right)^{n} h_n \psi_n\left(\frac{k}{\sqrt{2 \gamma}}\right) \phi_k(x).
  \label{eq:g_trunc_def_in_proof}
\end{align}
Using the definitions of these two wavefunctions, we can write \(\abs{\inp{\tilde{g}_x}{g_{trunc}}}^2\) as
\begin{align}
  \abs{\inp{\tilde{g}_x}{g_{trunc}}}^2
   & =
  \left| \sum_{k\in \mathbb{K}} \inp{\tilde{g}_x}{\phi_k} \inp{\phi_k}{g_{trunc}} \right|^2
  \nonumber
  \\
   & =
  \abs{
    \sum_{k \in \mathbb{K}_{cut}} \inp{\tilde{g}_x}{\phi_k} \inp{\phi_k}{g_{trunc}}
  }^2
  \nonumber
  \\
   & =
  \frac{1}{\mathcal{N}_t^2}
  \abs{
    \sum_{k \in \mathbb{K}_{cut}} \inp{\tilde{g}_x}{\phi_k} \inp{\phi_k}{\tilde{g}_x}
  }^2.
  \label{eq:whole_trunc_squared}
\end{align}
Similarly, from the definition of \(\tilde{g}_x\), we can show that
\begin{align}
  1 & = \inp{g_{trunc}}{g_{trunc}}
    & =
  \nonumber
  \\
    & =
  \sum_{k\in \mathbb{K}} \inp{g_{trunc}}{\phi_k} \inp{\phi_k}{g_{trunc}}
  \nonumber
  \\
    & =
  \sum_{k \in \mathbb{K}_{cut}} \inp{g_{trunc}}{\phi_k} \inp{\phi_k}{g_{trunc}}
  \nonumber
  \\
    & =
  \frac{1}{\mathcal{N}_t^2}
  \sum_{k \in \mathbb{K}_{cut}} \inp{\tilde{g}_x}{\phi_k} \inp{\phi_k}{\tilde{g}_x}.
  \label{eq:N_t_equation_equals_1}
\end{align}
Combining this with \Cref{eq:whole_trunc_squared}, we find that
\begin{align}
  1 -   \abs{\inp{\tilde{g}_x}{g_{trunc}}}^2
   & =
  1 -
  \sum_{k \in \mathbb{K}_{cut}} \inp{\tilde{g}_x}{\phi_k} \inp{\phi_k}{\tilde{g}_x}
  \nonumber
  \\
   & =
  \sum_{k \in \mathbb{K} : |k| > K} \inp{\tilde{g}_x}{\phi_k} \inp{\phi_k}{\tilde{g}_x},
  \label{eq:whole_trunc_overlap_to_complement}
\end{align}
where the last equality follows from the fact that \(\left\{ \phi_k \right\}_{|k| \leq K}\) and \(\left\{ \phi_k \right\}_{|k| > K}\) are complementary subspaces.

Therefore, in order to bound the trace distance, we first focus on bounding the quantity
\begin{equation}
  \sum_{k \in \mathbb{K} : |k| > K }
  \abs{\inp{\tilde{g}_x}{\phi_k}}^2
  =
  \sum_{k \in \mathbb{K} : |k| > K }
  \frac{\pi \sqrt{2}}{\widetilde{\mathcal{N}}_x^2 L \sqrt{\gamma}}
  \left| \sum_{n \in \left[ 0 .. l \right]} \left(-i\right)^{n} h_n \psi_n\left(\frac{k}{\sqrt{2 \gamma}}\right) \right|^2.
\end{equation}
We begin by using the triangle inequality, combined with the fact that \(|h_n| \leq 1\) for all \(n\) due to normalization to yield
\begin{align}
  \sum_{ k \in \mathbb{K} : |k| > K}
  \abs{\inp{\tilde{g}_x}{\phi_k}}^2
   & \leq
  \sum_{ k \in \mathbb{K} : |k| >  K }
  \frac{\pi \sqrt{2}}{\widetilde{\mathcal{N}}_x^2 L \sqrt{\gamma}}
  \sum_{n \in \left[ 0 .. l \right]} \sum_{m \in \left[ 0 .. l \right]} \left| \psi_n\left(\frac{k}{\sqrt{2 \gamma}}\right) \psi_m\left(\frac{k}{\sqrt{2 \gamma}}\right)\right|
  \nonumber
  \\
   & =
  \frac{\pi \sqrt{2}}{\widetilde{\mathcal{N}}_x^2 L \sqrt{\gamma}}
  \sum_{n \in \left[ 0 .. l \right]} \sum_{m \in \left[ 0 .. l \right]}
  \sum_{ k \in \mathbb{K} : |k| >  K }
  \left| \psi_n\left(\frac{k}{\sqrt{2 \gamma}}\right) \psi_m\left(\frac{k}{\sqrt{2 \gamma}}\right)\right|.
  \label{eq:g_phi_sum_bound}
\end{align}
We define the quantity \(R\) as
\begin{equation}
  R = \sum_{k \in \mathbb{K} : |k| > K }
  \left| \psi_n\left(\frac{k}{\sqrt{2 \gamma}}\right) \psi_m\left(\frac{k}{\sqrt{2 \gamma}}\right)\right|.
\end{equation}
For simplicity of notation, we suppress the dependence of \(R\) on the various parameters.
We will work to bound \(R\) and then return to \Cref{eq:g_phi_sum_bound}.

It is helpful to make some further simplifications.
The Hermite-Gaussian functions are either even or odd, so any product of them has definite parity.
This allows us to consider only positive values of \(k\),
\begin{equation}
  R =
  2 \sum_{k \in \mathbb{K}: k > K }
  \left| \psi_n\left(\frac{k}{\sqrt{2 \gamma}}\right) \psi_m\left(\frac{k}{\sqrt{2 \gamma}}\right)\right|.
\end{equation}
Without loss of generality we can assume that \(n \geq m\) since the case where \(m \geq n \) is symmetric and covers the remaining possibilities.
Note that our assumptions on \(K\) imply that \(K \geq \sqrt{4 \gamma l}
\geq \sqrt{4 \gamma n}\).
As a result, \(\frac{k}{\sqrt{2 \gamma}} > \sqrt{2 n}\).
We can therefore apply \Cref{lemma:hermite_gauss_bound}, which we state and prove below, to bound \(R\) by the expression
\begin{equation}
  R \leq 2 \sum_{k \in \mathbb{K}: k > K }
  c_n c_m 2^{n + m} \left( \frac{k}{\sqrt{2 \gamma}} \right)^{n + m} e^{-\frac{k^2}{4\gamma}}.
  \label{eq:R_bound_without_abs}
\end{equation}

\begin{lemma}
  Let \(\psi_n(x)\) be the \(n\)th Hermite-Gaussian function and let \(c_n\) be the normalization coefficient, as defined in \Cref{eq:hermite_gaussian_def} and \Cref{eq:hermite_gaussian_c_def}.
  Then, for any \(x \in \mathbb{R}\) such that \(x \geq \sqrt{2n}\),
  \begin{equation}
    0 \leq \psi_n(x) \leq c_n 2^n x^n e^{-x^2/2}.
  \end{equation}
  \label{lemma:hermite_gauss_bound}
\end{lemma}

\begin{proof}
  From the definition of the Hermite-Gaussian functions, it is clear that the lemma is equivalent to the claim that \(0 \leq H_n(x) \leq 2^n x^n\) for \(x \geq \sqrt{2n}\), where \(H_n(x)\) is the \(n\)th physicist's Hermite polynomial.
  We now recall a few useful properties of these Hermite polynomials.
  The Hermite polynomials have \(n\) distinct real roots, all of which can be shown to lie in the interval \(\left[- \sqrt{2n - 2}, \sqrt{2n - 2}\right] \subset \left[ -\sqrt{2n}, \sqrt{2n} \right]\).
  They also have definite parity, being odd functions if \(n\) is odd and even if \(n\) is even.
  Finally, the coefficient of their highest degree term is \(2^n\).
  We can use the properties of the Hermite polynomials to rewrite them in a more convenient form.

  Specifically, we can express \(H_n\) as
  \begin{equation}
    H_n(x) =
    \begin{cases}
      2^n x \left(\prod_{r \in R^+_n} \left( x+r \right)\left( x-r \right)\right) & \text{if}\ n\ \text{odd}
      \\
      2^n \left(\prod_{r \in R^+_n} \left( x+r \right)\left( x-r \right)\right)   & \text{if}\ n\ \text{even},
    \end{cases}
    \label{eq:hermite_polynomial_bound_expansion}
  \end{equation}
  where \(R^{+}_{n}\) denotes the set of positive roots of the \(n\)th Hermite polynomial.
  Notice that \(0 <  x^2 - r^2 = \left( x+r \right)\left(x - r\right) < x^2\) for all \(r \in R^+_n\), \(x > r\).
  The condition that \(x \geq \sqrt{2n}\) implies that \(x>r\) for all roots of \(H_n(x)\).
  Therefore, we can upper bound \(H_n(x)\) by replacing each of the terms \(\left( x+r \right)\left( x-r \right)\) in the product in \Cref{eq:hermite_polynomial_bound_expansion} with \(x^2\).
  This yields the desired upper bound,
  \begin{equation}
    H_n(x) < 2^n x^n
  \end{equation}
  for all \(x \geq \sqrt{2n}\).
  The lower bound follows from the fact that each term in the product is non-negative under the assumption on \(x\).
\end{proof}

Returning to the main thread of our proof, we now show that our assumption on \(K\) also guarantees that the sum in \Cref{eq:R_bound_without_abs} is monotonically decreasing.
This is true by inspection when \(n + m = 0\).
For the general case, consider the derivative of \(f(x) = x^{a}e^{-x^2/2}\) under the assumption that \(a \geq 1\),
\begin{equation}
  f'(x) = (a - x^2)x^{a-1}e^{-x^2/2}.
\end{equation}
The derivative is non-positive for all \(x \geq \sqrt{a}\), implying that \(\left( \frac{k}{\sqrt{2 \gamma}} \right)^{n + m} e^{-\frac{k^2}{4 \gamma}}\) is non-increasing for all \(\frac{k}{\sqrt{2 \gamma}} \geq \sqrt{n + m}\).
By our assumption on \(K\), we have that \(k > K \geq \sqrt{4 \gamma l}\), and therefore \(\frac{k}{\sqrt{2 \gamma}} > \sqrt{2 l} \geq \sqrt{n + m}\).
Therefore, the sum is monotonically decreasing.

In order to proceed, let us convert the sum to a series (summing over the momentum index \(p\) rather than the values of the momentum).
Substituting \(k = \frac{2 \pi p }{L}\) yields
\begin{align}
  R & \leq
  2 c_n c_m 2^{n+m}\sum_{k \in \mathbb{K} : k > K}
  \left( \frac{k}{\sqrt{2 \gamma}} \right)^{n + m} e^{-\frac{k^2}{4\gamma}}
  \nonumber
  \\
    & =
  2 c_n c_m 2^{n+m}\sum_{ p = \frac{K L}{2 \pi} + 1 - s}^{\infty}
  \left( \frac{2 \pi p}{L \sqrt{2 \gamma}} \right)^{n + m} e^{-\frac{(2\pi p)^2}{L^2 4 \gamma}},
\end{align}
where \(s \in \left[0, 1 \right)\) is chosen so that \(\frac{K L}{2 \pi} + 1 - s\) is a whole number (in order to ensure that we sum over all indices in \(\mathbb{K}\) where \(k > K\)).
Since the series is monotonically decreasing, we can bound it by an integral, yielding
\begin{equation}
  R \leq
  2 c_n c_m 2^{n+m}\left( \left( \frac{2 \pi a}{L \sqrt{2 \gamma}} \right)^{n + m} e^{-\frac{(2\pi a)^2}{L^2 4 \gamma}} + \int_{p = a}^\infty
  \left( \frac{2 \pi p}{L \sqrt{2 \gamma}} \right)^{n + m} e^{-\frac{(2\pi p)^2}{L^2 4 \gamma}} dp\right),
\end{equation}
where \(a = \frac{K L}{2 \pi} + 1 - s\).
Since the function in the series is monotonically decreasing but positive for \(k \geq K\),
we can replace \(a\) by \(\frac{KL}{2 \pi}\) and the result still holds.
Therefore, changing variables back to \(k\) and \(K\), we have that
\begin{equation}
  R \leq
  \frac{L}{\pi} c_n c_m 2^{n+m} \left( \frac{2 \pi}{L}\left( \frac{K}{\sqrt{2 \gamma}} \right)^{n + m} e^{-\frac{K^2}{4\gamma}} +  \int_{ k = K}^{\infty} \left( \frac{k}{\sqrt{2 \gamma}} \right)^{n + m} e^{-\frac{k^2}{4\gamma}} dk \right).
  \label{eq:R_bound_messy_integral}
\end{equation}

In order to simplify \Cref{eq:R_bound_messy_integral}, we can perform the substitutions \(w = n + m\), \(a=\frac{K}{\sqrt{2 \gamma}}\), \(u + a = \frac{k}{\sqrt{2\gamma}}\), \(C = \frac{L}{\pi} c_n c_m 2^{n + m}\), to yield
\begin{align}
  R \leq &
  C \left(
  \frac{2 \pi}{L} a^w e^{-a^2/2}
  +
  \sqrt{2 \gamma} \int_{u=0}^{\infty}
  \left(u + a\right)^{w} e^{-\left(u + a\right)^2/2}du
  \right)
  \nonumber
  \\
  =      &
  C \left(
  \frac{2\pi}{L}a^w e^{-a^2/2}
  +
  \sqrt{2 \gamma} e^{-a^2/2}\int_{u=0}^{\infty}
  \left(u + a\right)^{w} e^{-au} e^{-u^2/2}du
  \right).
  \label{eq:R_bound_nice_integral}
\end{align}
We can bound the expression in \Cref{eq:R_bound_nice_integral} by upper bounding the integrand with a simpler function.
Let \(f(u) = \left( u + a \right)^w e^{-au}\).
Let us examine the derivative (noting that \(u + a > 0\) by assumption),
\begin{equation}
  f'(u)=\left( w -a^2 - au\right)\left( u+a \right)^{w-1}e^{-au}.
\end{equation}
We see that \(f(u)\) is non-increasing when \(u \geq \frac{w - a^2}{a}\) since \(u > -a\) by assumption.
The observation that \(K \geq \sqrt{4 \gamma l}\) implies that \(a^2 \geq 2l\), which in turn implies that \(w - a^2\) is non-positive and therefore that \(u \geq 0 \geq \frac{w - a^2}{a}\).
Since \(f(u)\) is non-increasing over the whole region of integration, it is upper bounded by \(f(0) = a^w\).

We can therefore simplify \Cref{eq:R_bound_nice_integral} to yield
\begin{equation}
  R \leq
  C a^w e^{-a^2/2}\left(\frac{2\pi}{L} + \sqrt{2\gamma}\int_{u=0}^{\infty}e^{-u^2/2}du
  \right)
  =
  C a^w e^{-a^2/2}\left(\frac{2\pi}{L} + \sqrt{\pi \gamma}\right).
\end{equation}
Now we can substitute the values of \(a\), \(w\), and \(C\) back in, as well as the concrete expressions for \(c_n\) and \(c_m\),
\begin{align}
  R & \leq
  e^{-\frac{K^2}{4 \gamma}} (\frac{K}{\sqrt{2 \gamma}})^{n + m}
  \frac{L}{\pi} \frac{1}{\sqrt{2^n n!
      \sqrt{\pi}}} \frac{1}{\sqrt{2^m m! \sqrt{\pi}}} 2^{n + m} \left( \frac{2\pi}{L}  + \sqrt{\pi \gamma} \right)
  \nonumber
  \\
    & =
  \frac{L\left(  \frac{2\pi}{L} + \sqrt{\pi \gamma} \right)}{\pi^{3/2}}
  e^{-\frac{K^2}{4 \gamma}} \left(\frac{K}{\sqrt{\gamma}}\right)^{n + m}
  \frac{1}{\sqrt{n! }} \frac{1}{\sqrt{m! }}
  \nonumber
  \\
    & \leq
  \frac{L\left(  \frac{2\pi}{L} + \sqrt{\pi \gamma} \right)}{\pi^{3/2}}
  e^{-\frac{K^2}{4 \gamma}} \left(\frac{K}{\sqrt{\gamma}}\right)^{2 l}
  \frac{1}{\sqrt{n! }} \frac{1}{\sqrt{m! }},
\end{align}
where the last inequality holds due to the fact that our assumption on \(K\) implies that \(\frac{K}{\sqrt{\gamma}} \geq 1\) unless \(l=0\).
We can plug this expression back into the inequality from \Cref{eq:g_phi_sum_bound} and simplify, finding that
\begin{align}
  \sum_{k \in \mathbb{K} : |k| > K }
  \abs{\inp{\tilde{g}_x}{\phi_k}}^2
   & \leq
  \frac{\pi \sqrt{2}}{\widetilde{\mathcal{N}}_x^2 L \sqrt{\gamma}}
  \sum_{n \in \left[ 0 .. l \right]} \sum_{m \in \left[ 0 .. l \right]} R
  \\
   & \leq
  \frac{\sqrt{2}\left( 1 + \frac{2 \sqrt{\pi}}{L\sqrt{\gamma}} \right)}{\widetilde{\mathcal{N}}_x^2}
  e^{-\frac{K^2}{4 \gamma}} \left(\frac{K}{\sqrt{\gamma}}\right)^{2 l}
  \sum_{n \in \left[ 0 .. l \right]} \sum_{m \in \left[ 0 .. l \right]}
  \frac{1}{\sqrt{n!
    }} \frac{1}{\sqrt{m! }}.
\end{align}

We can bound the contribution from the sums,
\begin{align}
  \sum_{n \in \left[ 0 .. l \right]} \sum_{m \in \left[ 0 .. l \right]}
  \frac{1}{\sqrt{n!
    }}\frac{1}{\sqrt{m!}}
   & =
  \left(\sum_{n \in \left[ 0 .. l \right]} \frac{1}{\sqrt{n!}} \right)^2
  \nonumber
  \\
   & \leq
  \left(\sum_{n=0}^{\infty} \frac{1}{\sqrt{n!}} \right)^2
  \nonumber
  \\
   & =
  \left( 2 + \frac{1}{\sqrt{2}} + \sum_{n=3}^{\infty} \frac{1}{\sqrt{n!}} \right)^2
  \nonumber
  \\
   & \leq
  \left(2 + \frac{1}{\sqrt{2}} + \sum_{n=1}^{\infty} 2^{-k}\right)^2
  \nonumber
  \\
   & \leq \frac{20}{\sqrt{2}}.
\end{align}
Therefore, we have that
\begin{equation}
  \sum_{ k \in K : |k| > K }
  \abs{\inp{\tilde{g}_x}{\phi_k}}^2
  \leq
  \frac{20\left( 1 + \frac{2 \sqrt{\pi}}{L\sqrt{\gamma}} \right)}{\widetilde{\mathcal{N}}_x^2}
  e^{-\frac{K^2}{4 \gamma}} (\frac{K}{\sqrt{\gamma}})^{2 l}.
\end{equation}
Recall that we are ultimately interesting in establishing sufficient conditions to guarantee that the trace distance is bounded by \(\epsilon_t\), i.e,
\begin{equation}
  D(\tilde{g}_x, g_{trunc}) = \sqrt{1 -   \abs{\inp{\tilde{g}_x}{g_{trunc}}}^2} =
  \sqrt{  \sum_{ k \in K : |k| > K }
    \abs{\inp{\tilde{g}_x}{\phi_k}}^2}
  \leq \epsilon_t.
  \label{eq:infidelity_bound_whole_trunc}
\end{equation}
We can ensure that this holds by choosing a \(K\) such that
\begin{equation}
  \frac{20\left( 1 + \frac{2 \sqrt{\pi}}{L\sqrt{\gamma}} \right)}{\widetilde{\mathcal{N}}_x^2}
  e^{-\frac{K^2}{4 \gamma}} \left(\frac{K}{\sqrt{\gamma}}\right)^{2 l} \leq \epsilon_t^2.
  \label{eq:bound_for_K_getting_close}
\end{equation}

In the case where \(l > 0\), taking the logarithm of both sides and rearranging yields the inequality
\begin{equation}
  \frac{K^2}{4 \gamma l}
  - 2 \log \left( \frac{K}{\sqrt{\gamma}} \right)
  \geq
  \frac{
    2 \log \left( \epsilon_t^{-1} \right)
    + \log \left( 20 \right)
    + \log \left( 1 + \frac{2 \sqrt{\pi}}{L\sqrt{\gamma}} \right)
    - \log \left( \widetilde{\mathcal{N}}_x^2 \right)
  }{l}.
\end{equation}
We can further manipulate to yield the sufficient condition
\begin{equation}
  \frac{K^2}{4 \gamma l}
  - 2 \log \left( \frac{K}{2\sqrt{l \gamma}} \right)
  \geq
  \frac{
    2 \log \left( \epsilon_t^{-1} \right)
    + \log \left( 20 \right)
    + \log \left( 1 + \frac{2 \sqrt{\pi}}{L\sqrt{\gamma}} \right)
    - \log \left( \widetilde{\mathcal{N}}_x^2 \right)
    + 2 l \log(2 \sqrt{l})
  }{l}.
\end{equation}
The expression \(\frac{1}{2}x^2 - 2 \log \left( x \right) \geq 0\) holds for all positive \(x\), so we can subtract the positive quantity \(\frac{1}{2} \left( \frac{K}{2 \sqrt{\gamma l}} \right)^2 - 2 \log \left( \frac{K}{2 \sqrt{\gamma l}} \right)\) from the left-hand side and simplify to establish the more stringent condition that
\begin{equation}
  K
  \geq
  2\sqrt{2 \gamma}\sqrt{
    2 \log \left( \epsilon_t^{-1} \right)
    + \log \left( 20 \right)
    + \log \left( 1 + \frac{2 \sqrt{\pi}}{L\sqrt{\gamma}} \right)
    - \log \left( \widetilde{\mathcal{N}}_x^2 \right)
    + l \log(4 l)
  }
  .
\end{equation}
A simpler computation verifies that this expression is also a sufficient condition in the \(l = 0\) case under the convention that \(l \log(l)\) is defined to be zero when \(l = 0\).

We satisfy this condition by assumption, and so \Cref{eq:infidelity_bound_whole_trunc} holds, showing that
\begin{equation}
  D(\tilde{g}_x, g_{trunc}) \leq \epsilon_t.
\end{equation}
From \Cref{eq:N_t_equation_equals_1} and \Cref{eq:whole_trunc_overlap_to_complement}, we can see that
\(\mathcal{N}_t \leq 1\) and
\begin{equation}
  \sqrt{1 -  \mathcal{N}_t^2} =
  \sqrt{  \sum_{ k \in K : \abs{k} > K }
    \abs{\inp{\tilde{g}_x}{\phi_k}}^2}.
\end{equation}
Because \(1 - x \leq \sqrt{1 - x^2}\) for \(x \in \left[ 0, 1 \right]\), \Cref{eq:infidelity_bound_whole_trunc} implies that \(1 - \mathcal{N}_t \leq \epsilon_t\).
Therefore,
\begin{equation}
  1 - \epsilon_t \leq \mathcal{N}_t \leq 1,
\end{equation}
completing the proof.

\subsection{Proof of \Cref{lemma:polynomial_approximation}}
\label{sec:polynomial_approximation_proof}

Let us recall the statement of the lemma,
\polynomiallemma*
To proceed, we begin by proving the lemma below,
\begin{restatable}{lemma}{approximationlemma}
  Let \(\psi_n(x)\) denote the \(n\)th Hermite-Gaussian function.
  Let \(C, \epsilon_{cheb} \in \mathbb{R}_{>0}\).
  For any \(m \in \mathbb{Z}\) such that
  \begin{equation}
    m \geq \frac{eC}{\sqrt{2}} \left( \frac{eC}{\sqrt{2}} +  \sqrt{2n + 1}\right)
  \end{equation}
  and
  \begin{equation}
    m \geq \frac{2 \log \epsilon_{cheb}^{-1}}{\log 2},
  \end{equation}
  and
  \begin{equation}
    m \geq 1
  \end{equation}
  there exists a degree \(m - 1\) polynomial \(p(x)\) such that
  \begin{equation}
    |\psi_n(x) - p(x)| \leq \epsilon_{cheb}
  \end{equation}
  for all \(x \in \left[ -C, C \right]\).
  \label{lemma:chebyshev_approximation}
\end{restatable}
\begin{proof}

  To prove this lemma, we will use a standard result from approximation theory, which we state below,
  \begin{proposition}[Adapted from  Ref.~\citenum{Stewart1996-za}, Lecture 20 by a change of variables]
    Let \(m \in \mathbb{Z}_{>0}\), \(a, b \in \mathbb{R}\) such that \(a < b\), and \(f: \left[ a, b \right] \rightarrow \mathbb{R}\) be a function whose first \(m\) derivatives exist at all points in \(\left[ a, b \right]\).
    Let \(x_i\) for \(i \in \left[ 0..m-1 \right]\) be the Chebyshev nodes of the first kind, i.e., let
    \begin{equation}
      x_i = \cos \left( \frac{2i + 1}{2m + 2} \pi \right),\;\; i \in \left[ 0 .. m-1 \right]
    \end{equation}
    and let
    \begin{equation}
      t_i = \frac{a + b}{2} + \frac{b - a}{2} x_i.
    \end{equation}

    Let \(p(x)\) be the unique degree \(m-1\) polynomial that satisfies \(p(t_i) = f(t_i)\) for all \(i \in \left[ 0..m-1 \right]\).
    Then for all \(x \in \left[ a, b \right]\),
    \begin{equation}
      |f(x) - p(x)| \leq \frac{(b-a)^{m}}{2^{2m-1} (m)!
      } \max_{\xi \in \left[ a, b \right]}\Big|f^{(m)}(\xi)\Big|.
      \label{eq:chebyshev_error_bound}
    \end{equation}
    \label{lemma:standard_chebyshev_approximation_citation}
  \end{proposition}

  In order to apply this result, we shall proceed by bounding \(\big| \psi_n^{(m)}(x)\big|\) as a function of \(m\) and \(n\).
  Recall that the derivatives of the Hermite-Gaussian functions obey the recursion relation
  \begin{equation}
    \psi_n'(x) = \sqrt{\frac{n}{2}}\psi_{n-1}(x) - \sqrt{\frac{n+1}{2}}\psi_{n+1}(x),
  \end{equation}
  except where \(n=0,\) in which case the first term is zero.
  Consider the outcome of applying this recursion relation \(m\) times to obtain an expression for \(\psi_n^{(m)}(x)\).
  The resulting expression would contain at most \(2^{m}\) non-zero terms.
  Each term would consist of some \(\psi_j(x)\) multiplied by a coefficient whose absolute value is at most \(\prod_{i=1}^{m} \sqrt{\frac{n+i}{2}} = 2^{-m/2} \sqrt{\frac{(n+m)!
    }{n!}}\).
  Therefore, taking the absolute value and applying the triangle inequality, we find that
  \begin{equation}
    \big|\psi_n^{(m)}(x)\big|
    \leq
    2^{m/2} \sqrt{\frac{(n+m)!
      }{n!}} \max_{j}\big|\psi_j(x)\big|.
    \label{eq:bound_after_recursive_argument}
  \end{equation}

  The Hermite-Gaussian functions are known to satisfy an inequality known as Cramér's inequality,
  \begin{equation}
    |\psi_j(x)| \leq \pi^{-1/4},
  \end{equation}
  \(\)
  for all \(j\) and all \(x \in \mathbb{R}\).
  We can use the observation that \((n+1+j)(n+m-j)/2 \leq (n+m/2+1/2)^2\) for \(j \in \mathbb{R}\) to bound the first part of the expression from \Cref{eq:bound_after_recursive_argument}, finding that
  \begin{equation}
    2^{m/2} \sqrt{\frac{(n+m)!
      }{n!}}
    \leq
    2^{m/2} (n + m/2 + 1/2)^{m/2}
    =
    (2n+m + 1)^{m/2}.
  \end{equation}
  Combining this with Cramér's inequality, we have that
  \begin{equation}
    \big|\psi^{(m)}_n(x)\big|
    \leq
    (2n+m + 1)^{m/2} \pi^{-1/4}.
    \label{eq:hermite_gaussian_derivative_bound}
  \end{equation}

  We are now ready to apply \Cref{lemma:standard_chebyshev_approximation_citation}, setting \(a=-C\) and \(b = C\).
  Combining \Cref{eq:hermite_gaussian_derivative_bound} with \Cref{eq:chebyshev_error_bound} and simplifying tells us that we can find a degree \(m-1\) polynomial \(p(x)\) such that
  \begin{align}
    |\psi_n(x) - p(x)| \leq \frac{2}{\pi^{1/4} (m)!
    }  \left(\frac{C^2}{4}\left(2n+m + 1\right)\right)^{m/2}.
  \end{align}
  Applying Robbins' version of Stirling's formula, we have that
  \begin{align}
    |\psi_n(x) - p(x)|
     & \leq
    \frac{\sqrt{2}}{\pi^{3/4} \sqrt{m}} \left( \frac{e}{m} \right)^{m}
    \left(\frac{C^2}{4}\left(2n+m + 1\right)\right)^{m/2}
    \nonumber
    \\
     & \leq
    \left(\frac{e^2 C^2}{4m^2}\left(2n+m + 1\right)\right)^{m/2}.
  \end{align}

  Now we would like to show that \(\frac{e^2 C^2}{4m^2}\left(2n+m + 1\right) \leq \frac{1}{2}\) in order to establish that
  \begin{equation}
    |\psi_n(x) - p(x)|
    \leq
    \left(\frac{1}{2}\right)^{m/2}.
    \label{eq:nice_degree_bound}
  \end{equation}
  We can rearrange the desired inequality,
  \begin{equation}
    \frac{e^2 C^2}{4m^2}\left(2n+m + 1\right) \leq \frac{1}{2},
  \end{equation}
  finding that it is satisfied if and only if
  \begin{equation}
    2m^2 - e^2C^2m - e^2C^2(2n + 1)  \geq 0.
  \end{equation}
  Solving this inequality for \(m\), we find that the inequality holds when
  \begin{equation}
    m \geq \frac{1}{4}e^2C^2 \left( 1 + \sqrt{1 + \frac{16n + 8}{e^2C^2}} \right).
  \end{equation}
  Because \(\sqrt{a + b} \leq \sqrt{a} + \sqrt{b}\) for non-negative \(a\) and \(b\), we can establish a more stringent sufficient condition on \(m\) by requiring that
  \begin{align}
    m & \geq
    \frac{e^2 C^2}{4} \left( 2 + \sqrt{\frac{16n + 8}{e^2C^2}} \right)
    \nonumber
    \\
      & =
    \frac{eC}{\sqrt{2}} \left( \frac{eC}{\sqrt{2}} +  \sqrt{2n + 1}\right).
  \end{align}
  This condition is satisfied by assumption, so \Cref{eq:nice_degree_bound} must hold.

  By assumption, we have
  \begin{equation}
    m \geq \frac{2 \log \epsilon_{cheb}^{-1}}{\log 2}.
  \end{equation}
  Rearranging and exponentiating both sides yields
  \begin{equation}
    \left(\frac{1}{2}\right)^{m/2} \leq \epsilon_{cheb}.
  \end{equation}
  Therefore, by \Cref{eq:nice_degree_bound}, we have that
  \begin{equation}
    |\psi_n(x) - p(x)| \leq \epsilon_{cheb},
  \end{equation}
  completing the proof of \Cref{lemma:chebyshev_approximation}.
\end{proof}

Let us recall the definition of \(g_{trunc}\) from \Cref{eq:g_trunc_def},
\begin{equation}
  g_{trunc}(x)
  =
  \frac{2^{1/4}\sqrt{\pi}}{\widetilde{\mathcal{N}}_x\mathcal{N}_t\gamma^{1/4}\sqrt{L}} \sum_{k \in \mathbb{K}_{cut}} \sum_{n \in \left[ 0 .. l \right]} \left(-i\right)^{n} h_n \psi_n\left(\frac{k}{\sqrt{2 \gamma}}\right) \phi_k(x).
  \label{eq:g_trunc_def_in_proof_again}
\end{equation}
We now seek to apply \Cref{lemma:chebyshev_approximation} to bound the error incurred by replacing each \(\psi_n(\frac{k}{\sqrt{2 \gamma}})\) in \Cref{eq:g_trunc_def_in_proof_again} with a polynomial \(p_n(\frac{k}{\sqrt{2 \gamma}})\).
Using this lemma, we can find a set of \(l + 1\) degree \(m-1\) polynomials \(\left\{ p_n \right\}\) such that
\begin{equation}
  \left|\psi_n\left(\frac{k}{\sqrt{2\gamma}}\right) - p_n\left(\frac{k}{\sqrt{2 \gamma}}\right)\right| \leq \epsilon_{cheb}
  \label{eq:p_n_error_bounds}
\end{equation}
for an arbitrary \(\epsilon_{cheb} \in \mathbb{R}_{>0}\) and for all  \(n \in \left[ 0..l \right]\) and \(k \in \left[ -K, K \right]\), so long as we take \(m\) to satisfy
\begin{align}
  m \geq & \frac{eK}{2\sqrt{\gamma}} \left( \frac{eK}{2\sqrt{\gamma}} + \sqrt{2l + 1} \right)
  \nonumber
  \\
  m \geq & \frac{2 \log \epsilon_{cheb}^{-1}}{\log 2}
  \label{eq:m_K_bound}
  \\
  m \geq & 2 \nonumber.
\end{align}
Using these polynomials, we define the function \(f_{x \; (unnormalized)}\) as
\begin{equation}
  f_{x \; (unnormalized)}(x) =
  \frac{2^{1/4}\sqrt{\pi}}{\widetilde{\mathcal{N}}_x\mathcal{N}_t\gamma^{1/4}\sqrt{L}} \sum_{k \in \mathbb{K}_{cut}} \sum_{n \in \left[ 0 .. l \right]} \left(-i\right)^{n} h_n p_n \left(\frac{k}{\sqrt{2 \gamma}}\right) \phi_k(x).
  \label{eq:g_poly_def_in_proof}
\end{equation}
Combining \Cref{eq:p_n_error_bounds} and the fact that \(\sum_{n \in \left[0..l \right]} |h_n|^2 = 1\) by construction, we can bound the infinity norm distance (in the basis of the \(\left\{ \phi_k \right\}s\)) between \(g_{trunc}\) and \(f_{x \; (unnormalized)}\),
\begin{align}
  \norm{g_{trunc} - f_{x \; (unnormalized)}}_{\infty} & =
  \frac{2^{1/4}\sqrt{\pi}}{\widetilde{\mathcal{N}}_x\mathcal{N}_t\gamma^{1/4}\sqrt{L}}
  \max_{k \in \mathbb{K}_{cut}} \left|
  \sum_{n \in \left[ 0..l \right]}
  \left(-i\right)^{n} h_n\left( \psi_n(\frac{k}{\sqrt{2 \gamma}}) - p_n(\frac{k}{\sqrt{2 \gamma}}) \right)
  \right|
  \nonumber
  \\
                                                   & \leq
  \frac{2^{1/4}\sqrt{\pi}}{\widetilde{\mathcal{N}}_x\mathcal{N}_t\gamma^{1/4}\sqrt{L}}
  \max_{k \in \mathbb{K}_{cut}}
  \sum_{n \in \left[ 0..l \right]}
  \left|
  \left(-i\right)^{n} h_n\left( \psi_n(\frac{k}{\sqrt{2 \gamma}}) - p_n(\frac{k}{\sqrt{2 \gamma}}) \right)
  \right|
  \nonumber
  \\
                                                   & \leq
  \frac{2^{1/4}\sqrt{\pi}}{\widetilde{\mathcal{N}}_x\mathcal{N}_t\gamma^{1/4}\sqrt{L}}
  \epsilon_{cheb}
  \sum_{n \in \left[ 0..l \right]}
  |
  h_n
  |
  \nonumber
  \\
                                                   & \leq
  \frac{2^{1/4}\sqrt{\pi}}{\widetilde{\mathcal{N}}_x\mathcal{N}_t\gamma^{1/4}\sqrt{L}}
  \epsilon_{cheb}
  \sqrt{l + 1}.
  \label{eq:g_trunctaed_poly_infinity_norm_bound}
\end{align}

From this notion of error in the infinity norm, we would like to obtain a bound on the trace distance between \(g_{trunc}\) and the normalized function
\begin{equation}
  f_{x} = \frac{1}{\mathcal{N}_p} f_{x \; (unnormalized)}, \;\;\; \text{where } \mathcal{N}_p = \norm{f_{x \; (unnormalized)}}.
  \label{eq:normalized_g_poly_def_in_proof}
\end{equation}
We will proceed by making use of the following lemma, which we state and prove below.
\begin{restatable}{lemma}{fidelityboundlemma}
  Let \(a\) and \(b\) be two elements of a vector space endowed with an inner product, \(\inp{\cdot}{\cdot}\), and let \(\norm{\cdot}\) denote the norm induced by this inner product.
  Let \(\Delta = b - a\).

  Then the following inequality holds:
  \begin{equation}
    \frac{\norm{a}^2 - 2 \norm{a} \norm{\Delta}}{\norm{b}^2}
    \leq
    \frac{\left| \inp{a}{b} \right|^2}{\norm{a}^2 \norm{b}^2}
    \leq 1.
  \end{equation}
  \label{lemma:simple_fidelity_bound}
\end{restatable}

\begin{proof}
  Using linearity, we can write the absolute value of the inner product as
  \begin{equation}
    \big|\innerproductcomma{a}{b}\big|
    =
    \big|\innerproductcomma{a}{a} + \innerproductcomma{a}{\Delta}\big| =
    \big| ||a||^2 + \innerproductcomma{a}{\Delta} \big|.
  \end{equation}
  Applying the reverse triangle inequality yields
  \begin{equation}
    \Big| ||a||^2 - \big| \innerproductcomma{a}{\Delta} \big| \Big|
    \leq
    \big|\innerproductcomma{a}{b}\big|.
  \end{equation}
  Applying the Cauchy-Schwarz inequality to \(\big| \innerproductcomma{a}{\Delta} \big|\) shows us that
  \begin{equation}
    \Big| ||a||^2 - ||a|| \; ||\Delta|| \Big|
    \leq
    \big|\innerproductcomma{a}{b}\big|.
  \end{equation}

  We can square both sides to show that
  \begin{equation}
    \left(||a||^2 - ||a|| \; ||\Delta||\right)^2
    \leq
    \big|\innerproductcomma{a}{b}\big|^2.
  \end{equation}
  Expanding the square and dropping a positive term, we find that
  \begin{equation}
    ||a||^4 - 2 ||a||^3 ||\Delta|| \leq
    \big|\innerproductcomma{a}{b}\big|^2.
  \end{equation}
  Dividing by \(||a||^2||b||^2\) and using the Cauchy-Schwarz inequality to produce an upper bound yields
  \begin{equation}
    \frac{||a||^2 - 2 ||a|| \; ||\Delta||}{||b||^2}
    \leq
    \frac{\big|\innerproductcomma{a}{b}\big|^2}{||a||^2 \; ||b||^2}
    \leq 1,
  \end{equation}
  completing the proof.
\end{proof}

We proceed by applying \Cref{lemma:simple_fidelity_bound}, setting \(a = g_{trunc}\), \(b = f_{x \; (unnormalized)}\), and \(\Delta = g_{trunc} - f_{x \; (unnormalized)}\).
This yields the following inequality,
\begin{equation}
  \frac{
    1 - 2 || \Delta ||
  }{
    \mathcal{N}_p^2
  }
  \leq
  \abs{\inp{g_{trunc}}{{{f_{x}}}}}^2
  \leq 1.
\end{equation}
Taking the norm of both sides of the expression \(f_{x \; (unnormalized)} = g_{trunc} - \Delta\) and applying the triangle inequality
yields
\begin{equation}
  \mathcal{N}_p \leq 1 + \norm{\Delta}.
\end{equation}
Therefore,
\begin{equation}
  \frac{
    1 - 2 \norm{ \Delta }
  }{
    \left( 1 + \norm{ \Delta }\right)^2
  }
  \leq
  \abs{\inp{g_{trunc}}{{{f_{x}}}}}^2
  \leq 1.
\end{equation}
Making use of the fact that \(1 - 4b \leq \frac{1 - 2b}{\left( 1 + b \right)^2}\) for any positive \(b\), \footnote{
  This identity can be verified by starting with the inequality \(-4b^3 - 7b^2 \leq 0\), rewriting it as \((1 - 4b)(1+b)^2 - (1 - 2b) \leq 0\), adding \((1 - 2b)\) and dividing by \((1 + b)^2\).
}
we arrive at the inequality
\begin{equation}
  1 - 4 \norm{ \Delta }
  \leq
  \abs{\inp{g_{trunc}}{{{f_{x}}}}}^2
  \leq 1,
\end{equation}
which implies that
\begin{equation}
  \sqrt{1 -
    \abs{\inp{g_{trunc}}{{{f_{x}}}}}^2
  }
  \leq 2 \sqrt{\norm{ \Delta}}.
  \label{eq:trunc_poly_infideliy_two_norm}
\end{equation}

In order to proceed, we need to convert our infinity norm bound from \Cref{eq:g_trunctaed_poly_infinity_norm_bound},
\begin{equation}
  \norm{ \Delta }_{\infty} \leq
  \frac{2^{1/4}\sqrt{\pi}}{\widetilde{\mathcal{N}}_x \mathcal{N}_t\gamma^{1/4}\sqrt{L}}
  \epsilon_{cheb}
  \sqrt{l + 1}.
\end{equation}
into a two norm bound.
Note that we can conflate the \(L^2\) norm of the function \(\Delta\) and the \(l^2\) norm of \(\Delta\) understood as the vector of coefficients of the basis functions \(\phi_k\) since the \(\phi_k\) are orthonormal functions.
Recalling that \(\mathbb{K} = \left\{ \frac{2 \pi p}{L} : p \in \mathbb{Z} \right\}\), we can see that there are at most \(\frac{KL}{\pi} + 1\) basis functions \(\phi_k\) where \(|k| \leq K\) and \(\Delta\) may have a non-zero coefficient.
Therefore,
\begin{equation}
  \norm{ \Delta } \leq \norm{ \Delta }_{\infty} \sqrt{\frac{KL}{\pi} + 1} ,
\end{equation}
and thus,
\begin{align}
  \sqrt{1 -
    \abs{\inp{g_{trunc}}{{{f_{x}}}}}^2
  }
   & \leq
  2 \sqrt{
    \sqrt{\frac{KL}{\pi} + 1}
    \frac{2^{1/4}\sqrt{\pi}}{\widetilde{\mathcal{N}}_x\mathcal{N}_t\gamma^{1/4}\sqrt{L}}
    \epsilon_{cheb}
    \sqrt{l + 1}
  }
  \nonumber
  \\
   & \leq
  2
  \sqrt{\frac{\epsilon_{cheb}}{\widetilde{\mathcal{N}}_x \mathcal{N}_t}}
  \left(
  \sqrt{2}
  \left( \frac{K}{\sqrt{\gamma}} + \frac{\pi}{L \sqrt{\gamma}} \right)
  \left( l + 1 \right)
  \right)^{1/4}.
\end{align}
If the right-hand side of this inequality is bounded above by \(\epsilon_{p}\), then we have established the desired bound on the trace distance.
This is equivalent to the condition that
\begin{equation}
  \epsilon_{cheb} \leq  \epsilon_{p}^2 \frac{\widetilde{\mathcal{N}}_x \mathcal{N}_t}{4} \sqrt{\frac{1}{\sqrt{2} (\frac{K}{\sqrt{\gamma}} + \frac{\pi}{L \sqrt{\gamma}}) (l+1)}}.
  \label{eq:epsilon_cheb_bound}
\end{equation}

We can set \(\epsilon_{cheb}\) to the largest allowed value by taking \Cref{eq:epsilon_cheb_bound} to be an equality.
Substituting this expression for \(\epsilon_{cheb}\) into \Cref{eq:m_K_bound}, we find that the following conditions on \(m\) are sufficient:
\begin{align}
  m \geq &
  2
  \nonumber
  \\
  m \geq &
  \frac{eK}{2\sqrt{\gamma}} \left( \frac{eK}{2\sqrt{\gamma}} + \sqrt{2l + 1} \right)
  \label{eq:m_conditions_recap}
  \\
  m \geq &
  \frac{ 4 \log \left( \epsilon_{p}^{-1}\right) - 2\log \left( \widetilde{\mathcal{N}}_x \right) - 2\log \left(\mathcal{N}_t\right) + \log(\frac{K}{\sqrt{\gamma}} + \frac{\pi}{L \sqrt{\gamma}}) + \log \left( l + 1 \right)}{\log 2} + 4.5.
  \nonumber
\end{align}
By our assumption in \Cref{eq:m_conditions}, the inequalities in \Cref{eq:m_conditions_recap} are satisfied.
Therefore, there exists a set of degree \(m-1\) polynomials \(p_n\) that allow us to construct a \(f_{x \; (unnormalized)}\) according to \Cref{eq:g_poly_def_in_proof}, and hence, a \(f_{x}\) according to \Cref{eq:normalized_g_poly_def_in_proof}, such that
\begin{equation}
  \norm{f_{x}} = 1,
\end{equation}
and
\begin{equation}
  D(g_{trunc}, f_{x}) = \sqrt{1 - \abs{\inp{g_{trunc}}{{{f_{x}}}}}^2} \leq \epsilon_{p}.
\end{equation}
Note that we can interpret the coefficient of each \(\phi_k\) in \Cref{eq:g_poly_def_in_proof} as a degree \(m-1\) polynomial in \(k\) and then normalize by dividing by \(\mathcal{N}_p\),
i.e.,
\begin{align}
  p(k)        & =
  \frac{2^{1/4}\sqrt{\pi}}{\widetilde{\mathcal{N}}_x\mathcal{N}_t \mathcal{N}_p\gamma^{1/4}\sqrt{L}}
  \sum_{n \in \left[ 0 .. l \right]} i^{n} h_n p_n(\frac{k}{\sqrt{2 \gamma}}),
  \label{eq:p_from_p_n_def}
  \\
  f_{x}(x) & = \sum_{k \in \mathbb{K}_{cut}} p(k) \phi_k(x),
\end{align}
completing the proof.

\section{Statement and proof of \Cref{cor:3d_basis_functions} (\nameref*{cor:3d_basis_functions})}
\label{app:3d_proof}

\Cref{thm:main_technical_result} shows that we can use a matrix product state to succinctly approximate the projection of a primitive Gaussian basis function onto a plane wave basis in one spatial dimension.
In order to make practical use of this result, we need to generalize it to the three-dimensional case, and to the case where the functions are translated away from the origin.
We begin by clarifying our notation.
Let \(g_x\), \(g_y\), \(g_z\) denote one-dimensional Gaussian basis functions with angular momentum quantum numbers \(l, m, n\), and width parameter \(\gamma\), defined so that they are normalized over the real line (as in \Cref{eq:g_x_def_main_text}).
Let \(\varphi_{\bm{k}}\) denote the three-dimensional plane wave with momentum \(\bm{k} \in \mathbb{K}^3\),
\begin{equation}
  \varphi_{\bm{k}}(x, y, z)  = \frac{1}{L^{3/2}} e^{i \left( k_x x + k_y y + k_z z \right)}.
  \label{eq:three_d_plane_wave_proof_def}
\end{equation}
Note that, when discussing functions with three spatial dimensions, we use the notation \(\int_{\Omega} \cdots d\Omega\) to indicate integration over the computational unit cell \(\left[ -L/2, L/2 \right]^3\).
We also overload our previously defined notation for the inner product and norm of one-dimensional functions by letting \(\inp{f}{g}\) denote the inner product given by \(\int_{\Omega} \left(f^* g\right) d\Omega\) and \(\norm{f}\) denote the corresponding induced norm.

Informally, our goal is to represent \(g(x, y, z) = g_x(x) g_y(y) g_z(z)\) using a linear combination of plane waves defined in three spatial dimensions.
As in the one-dimensional case, we shall assume that the computational unit cell in three dimensions is large enough that we can approximately project onto a plane wave basis with negligible error.
Specifically, we define
\begin{align}
  \tilde{g}_{x}(x) & =
  \frac{1}{\widetilde{\mathcal{N}}_x} \sum_{k\in \mathbb{K}} \left( \int_{\mathbb{R}} g_x(u-x_0) \phi_k(u) du \right) \phi_k(x),
  \nonumber
  \\
  \tilde{g}_{y}(y) & =
  \frac{1}{\widetilde{\mathcal{N}}_y} \sum_{k\in \mathbb{K}} \left( \int_{\mathbb{R}} g_y(u-y_0) \phi_k(u) du \right) \phi_k(y),
  \nonumber
  \\
  \tilde{g}_{z}(z) & =
  \frac{1}{\widetilde{\mathcal{N}}_z} \sum_{k\in \mathbb{K}} \left( \int_{\mathbb{R}} g_z(u-z_0) \phi_k(u) du \right) \phi_k(z),
  \label{eq:tilde_g_x_proof_def}
\end{align}
with \(\widetilde{\mathcal{N}}_x, \widetilde{\mathcal{N}}_y, \widetilde{\mathcal{N}}_z \in \mathbb{R}_{>0}\) chosen so that \(\norm{\tilde{g}_x} = \norm{\tilde{g}_y} = \norm{\tilde{g}_z} = 1\).
Then we define
\begin{equation}
  \tilde{g}(x, y, z) = \tilde{g}_x(x) \tilde{g}_y(y) \tilde{g}_z(z).
  \label{eq:tilde_g_3d_def}
\end{equation}
Note that there would be no approximation if the translated Gaussian basis functions were exactly zero outside of the computational unit cell.

Now we are ready to state our result.
Note that the required kinetic energy cutoff and bond dimension are nearly unchanged from \Cref{thm:main_technical_result} except for a small factor in the dependence on \(\epsilon\).
\begin{restatable}[Efficient MPS representation of three-dimensional primitive Gaussian basis functions]{corollary}{threeDimBasisFunction}
  For arbitrary \(l, m, n \in \mathbb{Z}_{\geq 0}\), \(\gamma \in \mathbb{R}_{>0}\), \(L \in \mathbb{R}_{>0}\), let \(\mathbb{K}\) be defined as in \Cref{eq:K_def}, let \(\varphi_{\bm{k}}\) be defined as in \Cref{eq:three_d_plane_wave_proof_def}, let \(\widetilde{\mathcal{N}_x}\), \(\widetilde{\mathcal{N}_y}\), \(\widetilde{\mathcal{N}_z}\) be defined as in \Cref{eq:tilde_g_x_proof_def}, and let \(\tilde{g}\) be defined as in \Cref{eq:tilde_g_3d_def}.
  Let \(\ell = \max\left( l, m, n \right)\).

  For an arbitrary \(\epsilon \in \left( 0, 1 \right)\), set
  \begin{align}
    K
      & =
    2\sqrt{2 \gamma}\sqrt{
      2 \log \left( 2\sqrt{3}\epsilon^{-1} \right)
      + \log \left( 45 \right)
      + \log \left( 1 + \frac{2 \sqrt{\pi}}{L\sqrt{\gamma}} \right)
      + \ell \log(4 \ell)
    },
    \label{eq:K_choice_3d}
    \\
    m &
    = \left\lceil 4e^2 \left(   2 \log \left( 2\sqrt{3}\epsilon^{-1} \right)
    + \log \left( 45 \right)
    + \log \left( 1 + \frac{2 \sqrt{\pi}}{L\sqrt{\gamma}} \right)
    + \ell \log(4 \ell) \right)\right\rceil,
    \label{eq:m_choice_3d}
  \end{align}
  and let \(\mathbb{K}^3_{cut} = \left\{ k \in \mathbb{K} : k \leq K \right\}^3\).

  If \(\widetilde{\mathcal{N}_x} \geq 2/3\), \(\widetilde{\mathcal{N}_y} \geq 2/3\),  and \(\widetilde{\mathcal{N}_z} \geq 2/3\), then there exists degree \(m-1\) polynomials \(p_x\), \(p_y\), and \(p_z\) such that the following statements are true:

  Let
  \begin{equation}
    f(x, y, z) = \sum_{k \in \mathbb{K}_{cut}^3} p_x(k_x) p_y(k_y) p_z(k_z) e^{i \left( k_x x_0 + k_y y_0 + k_z z_0 \right)} \varphi_{\bm{k}}(x, y, z).
    \label{eq:f_corollary_def}
  \end{equation}
  Then
  \begin{equation}
    \norm{f} = 1
  \end{equation}
  and
  \begin{equation}
    D(\tilde{g}, f) \leq \epsilon.
  \end{equation}

  Furthermore, let \(\ket{f}\) denote the function \(f\) encoded in the standard representation used for first quantized simulation in a plane wave basis, i.e.,
  \begin{equation}
    \ket{f} = \sum_{\bm{k} \in \mathbb{K}^3_{cut}} p_x(k_x) p_y(k_y) p_z(k_z) e^{i \left( k_x x_0 + k_y y_0 + k_z z_0 \right)} \ket{\frac{k_x L}{2 \pi}}\ket{\frac{k_y L}{2 \pi}}\ket{\frac{k_z L}{2 \pi}}.
    \label{eq:3d_mps_approx_corollary}
  \end{equation}
  For an arbitrary \(n \in \mathbb{Z}\) such that \(n \geq 3 \lceil \log_2 \left|\mathbb{K}_{cut}\right| \rceil\), there exists an \(n\)-qubit matrix product state representation of \(\ket{f}\) with bond dimension at most \(2m + 3\).
  \label{cor:3d_basis_functions}
\end{restatable}

\begin{proof}
Our basic strategy for proving this corollary is to apply \Cref{thm:main_technical_result} to each Cartesian component with the appropriate error bound.
To that end, let \(f(x, y, z) = f_x(x)f_y(y)f_z(z)\) and \(g(x, y, z) = g_x(x)g_y(y)g_z(z)\) for some normalized \(f_x, f_y, f_z, g_x, g_y, g_z\).
Assume that \(D(f_x, g_x) \leq \epsilon_x\), \(D(f_y, g_y) \leq \epsilon_y\), and \(D(f_z, g_z) \leq \epsilon_z\) for \(\epsilon_x, \epsilon_y, \epsilon_z \in \left( 0, 1 \right)\).
Then, for each \(w \in \left\{ x, y, z \right\}\), we have
\begin{equation}
  \abs{\inp{f_w}{g_w}}^2 \geq 1 - \epsilon_w^2.
\end{equation}
Therefore,
\begin{equation}
  \abs{\inp{f}{g}}^2 =
  \abs{\inp{f_x}{g_x}}^2
  \abs{\inp{f_y}{g_y}}^2
  \abs{\inp{f_z}{g_z}}^2
  \geq
  \left(1 - \epsilon_x^2\right)
  \left(1 - \epsilon_y^2\right)
  \left(1 - \epsilon_z^2\right)
  \geq
  1 - \epsilon_x^2 - \epsilon_y^2 - \epsilon_z^2.
\end{equation}
This implies that
\begin{equation}
  D(f, g) \leq \sqrt{\epsilon_x^2 + \epsilon_y^2 + \epsilon_z^2}.
\end{equation}
Therefore, in order to bound \(D(\tilde{g}, f)\) by \(\epsilon\), it suffices to bound the error in approximating each component by \(\frac{\epsilon}{\sqrt{3}}\).

In order to apply \Cref{thm:main_technical_result}, we have to reduce to the case where we are approximating the untranslated basis functions, so we make the observation that
\begin{align}
  \tilde{g}_{x}(x) & =
  \frac{1}{\widetilde{\mathcal{N}}_x} \sum_{k\in \mathbb{K}} \left( \int_{\mathbb{R}} g_x(v-x_0) \phi_k(v) dv \right) \phi_k(x),
  \nonumber
  \\
                   & =
  \frac{1}{\widetilde{\mathcal{N}}_x} \sum_{k\in \mathbb{K}} \left( \int_{\mathbb{R}} g_x(v) \phi_k(v+x_0) dv \right) \phi_k(x),
  \nonumber
  \\
                   & =
  \frac{1}{\widetilde{\mathcal{N}}_x} \sum_{k\in \mathbb{K}} \left( \int_{\mathbb{R}} g_x(v) \phi_k(v) dv \right) e^{i k x_0} \phi_k(x).
\end{align}
Similarly, the translations in the \(y\) and \(z\) directions manifest as momentum dependent phases.
Now we consider evaluating \(D(\tilde{g}, f)\).
The only terms in the inner product between \(\tilde{g}\) and \(f\) that don't vanish are the ones with the same index \(\bm{k}\), and the phases that come from the translations exactly cancel with the phases in the definition of \(f\).
Therefore, we can take \(x_0 = y_0 = z_0 = 0\) without loss of generality when evaluating \(D(\tilde{g}, f)\).
We can therefore apply \Cref{thm:main_technical_result} to approximate each the basis functions for each Cartesian component, neglecting the desired translation, to within an error \(\frac{\epsilon}{\sqrt{3}}\) using the values of \(K\) and \(m\) set in \Cref{eq:K_choice_3d} and \Cref{eq:m_choice_3d}.
This tells us that \(f\), as defined in \Cref{eq:f_corollary_def}, is indeed an approximation of \(\tilde{g}\) to within a distance \(\epsilon\).
The normalization of the approximate basis functions for each spatial dimension implies that \(\norm{f} = 1\).

The final statement regarding the matrix product state representation follows from the matrix product state representations for the individual Cartesian components, combined with three useful properties of the tensor train decomposition.
First of all, exponential functions on an equispaced grid can be represented by tensor trains with a bond dimension of one~\cite{Oseledets2013-vy}.
Secondly, the Hadamard (element-wise) product of two tensor trains with bond dimensions \(m_1\) and \(m_2\) can be written as a tensor train with bond dimension most \(m_1 m_2\)\cite{Oseledets2011-xm}.
This allows us to freely multiply the matrix product states guaranteed by \Cref{thm:main_technical_result} by the phase factors required to produce \Cref{eq:3d_mps_approx_corollary} without increasing the bond dimension.
Thirdly, the tensor product of multiple tensor trains (or matrix product states) is still a tensor train (MPS) and and the bond dimensions of the new indices are \(1\).
Therefore, we can take the tensor product of the matrix product state representations for each Cartesian component to obtain the claimed matrix product state representation, completing the proof.
\end{proof}

\section{Review of canonical orthogonalization}
\label{app:canonical_orthogonalization}

Atom-centered Gaussian basis functions have a number of convenient features for electronic structure calculations, but they also have the disadvantage that they do not (in general) form an orthonormal basis.
Depending upon the choice of basis set, it is possible for a given collection of Gaussian type orbitals to be nearly linearly-dependent, especially for larger more accurate bases and large system sizes.
Even in purely classical calculations, this can lead to issues with numerical precision.
It is common to address these difficulties by implicitly or explicitly orthogonalizing the orbitals.
There are several orthogonalization procedures that are routinely used.
In this work, we assume the use of canonical orthogonalization, which we review below.

The goal of the canonical orthogonalization procedure is to transform the \(N_g\) normalized basis functions \(\left\{ g_j \right\}\) into a collection of \(N'_g \leq N_g\) orthonormal functions \(\left\{ \chi'_i \right\}\).
Ideally, these new functions would span the same space as the original set.
In practice, canonical orthogonalization may entail discarding some portion of the space where the original set of functions is poorly conditioned.
We set the threshold for this truncation in terms of a parameter \(\sigma \in \mathbb{R}_{>0}\), which we define below.

In general, the set of \(\left\{ g_j \right\}\) is not orthonormal and the overlap matrix \(S\), defined by
\begin{equation}
  S_{ij} = \left\langle g_i, g_j \right\rangle,
\end{equation}
is not the identity.
Following the usual convention, we would like to define our new orbitals \(\left\{ \chi'_i \right\}\) as linear combinations of the old orbitals specified by a matrix \(X\),
\begin{equation}
  \chi'_i = \sum_{j} X_{ji} g_j.
\end{equation}
A quick computation verifies that demanding that the new orbitals are orthonormal, i.e., that \(\left\langle\phi'_j, \phi'_k \right\rangle = \delta_{jk}\), is equivalent to demanding that
\begin{equation}
  X^\dagger S X  = \mathbb{I}.
  \label{eq:X_matrix_condition}
\end{equation}

The matrix \(X\) is not uniquely specified by \Cref{eq:X_matrix_condition} and different procedures for orthogonalization correspond to different choices for its construction.
In the canonical orthogonalization procedure, we take advantage of the fact that we can write
\begin{equation}
  S = U D U^\dagger
\end{equation}
for some unitary matrix \(U\) and some diagonal matrix \(D\) with non-negative entries.
Assuming the original basis functions are linearly independent, we can satisfy \Cref{eq:X_matrix_condition} by setting
\begin{equation}
  X = U D^{-1/2}.
\end{equation}
When the original basis is linearly dependent, then \(S\) has some eigenvalues that are exactly \(0\) and the inverse square root is not well defined.
Small eigenvalues of \(S\) are formally acceptable, but they will contribute to numerical difficulties.
We deal with both cases by defining a truncated \(N_g \times N'_g\) transformation matrix
\begin{equation}
  \tilde{X} = \tilde{U} \tilde{D}^{-1/2},
  \label{eq:tilde_X_def}
\end{equation}
where \(\tilde{D}\) and \(\tilde{U}\) are defined by removing the eigenvalues of \(S\) that are less than some cutoff \(\sigma\), along with their associated eigenvectors.
The resulting collection of \(N'_g\) orbitals defined by
\begin{equation}
  \chi'_i = \sum_{j} \tilde{X}_{ji} g_j
  \label{eq:chi_prime_def}
\end{equation}
form an orthonormal set.

Canonical orthogonalization prevents numerical issues by removing the poorly conditioned parts of the subspace spanned by the \(g_j\).
One implication of this is captured in the lemma below, which provides a useful bound on the coefficients that appear when taking (normalized) linear combinations of outputs from the canonical orthogonalization procedure.

\begin{restatable}[Bound on the coefficients obtained from canonical orthogonalization]{lemma}{canonicalOrthogonalizationBound}
  Let \(\left\{ g_j \right\}\) be a collection of \(N_g\) functions that are \(L^2\)-normalized over \(\mathbb{R}^3\) and let
  \begin{equation}
    \chi = \sum_{j=1}^{N_g} c_j g_j
  \end{equation}
  be a similarly normalized linear combination of the \(g_j\) obtained after using the canonical orthogonalization procedure with an eigenvalue cutoff of \(\sigma\), as reviewed in \Cref{app:canonical_orthogonalization}.

  Then the vector of coefficients, \(\bm{c}\), satisfies
  \begin{equation}
    \norm{\bm{c}}_2 \leq \sigma^{-1}.
  \end{equation}
  \label{lemma:canonical_orthogonalization_coefficient_blowup}
\end{restatable}
\begin{proof}

  By assumption, we \(\chi\) was obtained from the canonical orthogonalization procedure with an eigenvalue cutoff of \(\sigma\).
  As a consequence, we can write \(\chi\) as a linear combination of orthonormal functions \(\left\{ \chi'_j \right\}\), as defined in \Cref{eq:chi_prime_def}.
  In other words,
  \begin{equation}
    \chi = \sum_i c'_i \chi'_i = \sum_{ij} c'_i \tilde{X}_{ji} g_{j} = \sum_{ijk} c'_i \tilde{U}_{jk} \tilde{D}^{-1/2}_{ki}  g_j,
  \end{equation}
  where \(\tilde{X} = \tilde{U}\tilde{D}^{-1/2}\) is as defined in \Cref{eq:tilde_X_def}.
  The coefficients of the \(g_j\) are therefore given by the expression
  \begin{equation}
    c_j = \sum_{ik} c'_i \tilde{U}_{jk} \tilde{D}^{-1/2}_{ki}.
  \end{equation}
  In other words,
  \begin{equation}
    \bm{c} = \tilde{U} \tilde{D}^{-1/2} \bm{c'}.
  \end{equation}
  Since \(||\bm{c'}||_2 = 1\), \(\tilde{U}\) is an isometry, and \(0 < \tilde{D}^{-1/2}_{kk} \leq \sigma^{-1}\), we have that
  \begin{equation}
    \norm{\bm{c}}_2 \leq \sigma^{-1}.
  \end{equation}

\end{proof}

\section{Proof of \Cref{lemma:mo_as_mps} (\nameref*{lemma:mo_as_mps})}
\label{app:MO_error_bound_proof}

In this appendix, we make use of \Cref{eq:3d_mps_approx_corollary} to bound the bond dimension required to accurately represent a normalized linear combination of primitive Gaussian basis functions in three spatial dimensions centered at different atomic coordinates, e.g., a molecular orbital.
We begin with a formal restatement \Cref{lemma:mo_as_mps}, which is identical to the statement of the lemma in the main text except that we provided concrete bounds on \(K\) and \(M\) rather than using asymptotic notation.
\MOasMPS*

We also recall the definition of a molecular orbital \(\chi\),
\begin{equation}
  \chi(x, y, z) = \sum_{j=1}^{N_g} c_j g_j(x, y, z),
\end{equation}
where each \(g_j\) is a function of the form
\begin{equation}
  g_j(x, y, z) \propto \left( x - x_0 \right)^l \left( y - y_0 \right)^m \left( z - z_0 \right)^n e^{- \gamma \left( \left( x - x_0 \right)^2 + \left( y - y_0 \right)^2 + \left( z - z_0 \right)^2 \right)},
\end{equation}
normalized over \(\mathbb{R}^3\).
Following the statement of the lemma, we assume that this linear combination of primitive Gaussian basis functions has been obtained using a canonical orthogonalization procedure.
We review canonical orthogonalization in \Cref{app:canonical_orthogonalization}, but the key idea is that canonical orthogonalization guarantees that we end up with a linear combination of \(g_j\) functions that is not too poorly conditioned.
This implies the following lemma, which we recall from \Cref{app:canonical_orthogonalization},
\canonicalOrthogonalizationBound*

Recall the definition of the approximately projected molecular orbital \(\tilde{\chi}\) from \Cref{eq:chi_tilde_main_text},
\begin{equation}
  \tilde{\chi} \propto \sum_j c_j \tilde{g}_j,
  \label{eq:chi_tilde_main_recalled}
\end{equation}
where the \(g_j\) are the approximately projected Gaussian basis functions,
\begin{equation}
  \tilde{g}_j \propto \sum_{\bm{k} \in \mathbb{K}^3} \left(\int_{\mathbb{K}^{3}} \varphi_{\bm{k}}(x, y, z) g_j(x, y, z) d \Omega\right) \varphi_{\bm{k}}.
  \label{eq:g_tilde_recalled_MO}
\end{equation}
Both constants of proportionality are define by requiring that the corresponding functions be normalized over the computational unit cell.
In order to treat the impact of the normalization rigorously, we implicitly define the coefficients \(\tilde{c}_j\) by the following equation,
\begin{equation}
  \tilde{\chi}(x, y, z)
  =
  \sum_{j=1}^{N_g} \tilde{c}_j \tilde{g}_j(x, y, z).
  \label{eq:chi_tilde_def_redefined}
\end{equation}
Note that each \(\tilde{c}_j\) is obtained from the corresponding \(c_j\) by absorbing the normalization factors for \(\tilde{\chi}\) (in \Cref{eq:chi_tilde_main_recalled}) and the appropriate \(\tilde{g}_j\) (in \Cref{eq:g_tilde_recalled_MO}).
All of these normalization factors arise from the difference between integrating over the computational unit cell and integrating over \(\mathbb{R}^3\), and they approach one in the limit where \(L\) goes to infinity and these two integrals become identical.
\footnote{In fact, they approach one quickly since the \(g_j\) all decay exponentially at long distances.}
Therefore, for a sufficiently large value of \(L\), we must have
\begin{equation}
  \norm{\bm{\tilde{c}}}_2 \leq \sqrt{2}\sigma^{-1}.
  \label{eq:tilde_c_norm_bound}
\end{equation}

We will use our ability to efficiently approximate the individual primitive Gaussian basis functions to approximate their linear combinations.
For any \(\epsilon_{go} \in \left( 0, 1 \right)\), we can apply \Cref{cor:3d_basis_functions} to obtain an \(f_j\) that approximates each \(\tilde{g}_j\) to within \(\epsilon_{go}\) in the trace distance.
Furthermore, note that we can freely multiply the \(f_j\) by a phase without changing the degree of the polynomial approximation or the bond dimensions of the matrix product states we construct, so we can assume without loss of generality that \(\inp{f_j}{\tilde{g}_j} \in \mathbb{R}_{\geq 0}\).
We can then define \(\tau\) by the equations
\begin{align}
  \tau(x, y, z)                & = \norm{\tau_{unnormalized}}^{-1}\tau_{unnormalized}(x, y, z),
  \\
  \tau_{unnormalized}(x, y, z) & = \sum_{j=1}^{N_g} \tilde{c}_j f_j(x, y, z).
  \nonumber
\end{align}
By construction, \(\tau\) will be a normalized linear combination of momentum eigenfunctions for some momentum cutoff.
Our primary task is to understand how to set \(\epsilon_{go}\) in order to guarantee that \(D(\tau, \tilde{\chi}) \leq \epsilon\).
After determining this, we can translate the desired scaling of \(\epsilon_{go}\) into bounds on the momentum cutoff and bond dimension.

In order to bound \(D(\tau, \tilde{\chi})\), we first focus on bounding \(\left| \inp{\tilde{\chi}}{\tau} \right|^2\).
To do so, we use \Cref{lemma:simple_fidelity_bound}.
Taking \(a = \tilde{\chi}\) and \(b = \tau_{unnormalized}\), we have
\begin{equation}
  \left| \inp{\tilde{\chi}}{\tau}\right|^2 =
  \frac{1}{\norm{\tau_{unnormalized}}^2}\left|\inp{\tau_{unnormalized}}{\tilde{\chi}}\right|^2 \geq \frac{1 - 2 \norm{\Delta}}{\norm{\tau_{unnormalized}}^2},
\end{equation}
where \(\Delta = \tau_{unnormalized} - \tilde{\chi}\).
Using the triangle inequality, we can establish that
\begin{equation}
  \norm{\tau_{unnormalized}}^2 = \norm{\tilde{\chi} + \Delta}^2 \leq \left(\norm{\tilde{\chi}} + \norm{\Delta}\right)^2 = \left(1 + \norm{\Delta}\right)^2.
\end{equation}
Therefore, assuming for now that \(1 - 2 \norm{\Delta}\) is not negative, we can establish that
\begin{equation}
  \frac{1 - 2 \norm{ \Delta }}{\norm{ \tau_{unnormalized} }^2} \geq
  \frac{1 - 2 \norm{ \Delta }}{\left(1 + \norm{\Delta}\right)^2}.
\end{equation}
Using the fact that
\begin{equation}
  \frac{1}{\left(1 + x\right)^2} \geq 1 - 2x
  \label{eq:1_plus_x_squared_inequality}
\end{equation}
for \(x \geq 0\),\footnote{This inequality can be proven by starting with \(1 \geq 1 -2x^3 -3x^2\), rearranging to show that \(1 \geq \left(1 + x\right)^2(1-2x)\), and dividing by \((1+x)^2\).}, we can simplify further to find that
\begin{equation}
  \left| \inp{\tilde{\chi}}{\tau}\right|^2
  \geq
  \left(1 - 2 \norm{\Delta}\right)\left(1 - 4 \norm{\Delta}\right)
  \geq
  1 - 6 \norm{\Delta}.
  \label{eq:just_norm_Delta_bound}
\end{equation}
Note that this inequality holds by inspection if \(1 - 2 \norm{\Delta}\) is negative, so \Cref{eq:just_norm_Delta_bound} does not depend on this assumption.

Now that our task has been reduced to bounding the error in \(\norm{\Delta}\), let us expand this quantity as
\begin{equation}
  \norm{\Delta} = \norm{
    \sum_j \tilde{c}_j \left(f_j - \tilde{g}_j\right)
  }.
\end{equation}
Applying the triangle inequality, we have
\begin{equation}
  \norm{\Delta} \leq \sum_j \abs{\tilde{c}_j} \norm{f_j - \tilde{g}_j} =
  \sum_j \abs{\tilde{c}_j} \sqrt{\inp{f_j - \tilde{g}_j}{f_j - \tilde{g}_j}} =
  \sum_j \abs{\tilde{c}_j} \sqrt{2 - \inp{f_j}{\tilde{g}_j} - \inp{\tilde{g}_j}{f_j}}.
  \label{eq:getting_close_delta_norm}
\end{equation}
If we have \(D(f_j, \tilde{g_j}) \leq \epsilon_{go}\), then
\begin{equation}
  \left| \inp{f_j}{\tilde{g}_j} \right|^2 \geq 1 - \epsilon_{go}^2.
\end{equation}
By assumption, we also have that \(\left| \inp{f_j}{\tilde{g}_j} \right|  = \inp{f_j}{\tilde{g}_j}\), allowing us to obtain the expression
\begin{equation}
  \inp{f_j}{\tilde{g}_j} = \inp{\tilde{g}_j}{f_j} \geq \sqrt{1 - \epsilon_{go}^2} \geq 1 - \epsilon_{go}.
  \label{eq:epsilon_j_bound_useful}
\end{equation}
Combining \Cref{eq:getting_close_delta_norm} and \Cref{eq:epsilon_j_bound_useful}, we have
\begin{equation}
  \norm{\Delta}
  \leq
  \sqrt{2\epsilon_{go}} \sum_j \abs{\tilde{c}_j}
  =
  \sqrt{2\epsilon_{go}} \norm{\bm{\tilde{c}}}_1
  \leq
  \sqrt{2\epsilon_{go} N_g} \norm{\bm{\tilde{c}}}_2.
  \label{eq:intermediate_Delta_e_gaus_bound}
\end{equation}

Combining \Cref{eq:tilde_c_norm_bound} with \Cref{eq:intermediate_Delta_e_gaus_bound}, we find that
\begin{equation}
  \norm{\Delta} \leq 2\sigma^{-1}\sqrt{\epsilon_{go} N_g}.
\end{equation}
Returning to \Cref{eq:just_norm_Delta_bound} and the definition of \(D(\tilde{\chi}, \tau)\), we have that
\begin{equation}
  D(\tilde{\chi}, \tau) = \sqrt{1 - \norm{\inp{\tilde{\chi}}{\tau}}^2} \leq
  \sqrt{12\sigma^{-1}\sqrt{\epsilon_{go} N_g}}.
\end{equation}
Therefore, we can guarantee that \(D(\tilde{\chi}, \tau) \leq \epsilon\) by requiring that
\begin{equation}
  \epsilon_{go} \leq \frac{\epsilon^4 \sigma^2}{144 N_g}.
  \label{eq:epsilon_go_requirements}
\end{equation}

We can apply \Cref{cor:3d_basis_functions} to each \(\tilde{g}_j\) to obtain to an \(f_j\) with the desired error bound.
Let \(K_j\) denote the corresponding kinetic energy cutoffs (from \Cref{eq:K_choice_3d}) and let \(m_j\) denote the corresponding values of \(m\) (from \Cref{eq:m_choice_3d}).
Recall the definition of \(K\) from the statement of the lemma,
\begin{equation}
  K
  =
  2\sqrt{2 \Gamma}\sqrt{
    2 \log \left( \frac{288\sqrt{3}N_g}{\epsilon^4 \sigma^2} \right)
    + \log \left( 45 \right)
    + \ell \log(4 \ell)
  }.
  \label{eq:K_choice_3d_recalled}
\end{equation}
Assuming that \(L\) is sufficiently large that \(\log\left( 1 + \frac{2 \sqrt{\pi}}{L \sqrt{\gamma}} \right) \leq 1\), then \(K \geq K_j\) for all \(j\) by inspection, and all of the \(f_j\) are composed of linear combinations of plane waves with momenta \(\bm{k} \in \mathbb{K}_{cut}^3\).

Now, let \(m = \max_j m_j\).
Examining \Cref{eq:m_choice_3d}, we see that
\begin{equation}
  m \leq
  \left\lceil 4e^2 \left(
  2 \log \left( \frac{288\sqrt{3}N_g}{\epsilon^4 \sigma^2} \right)
  + \log \left( 45 \right)
  + \ell \log(4 \ell) \right)\right\rceil.
\end{equation}
\Cref{cor:3d_basis_functions} tells us that we can represent each \(f_j\) as a matrix product state \(\ket{f_j}\) with bond dimension at most \(2m + 3\).
Using standard properties of matrix product states, we can therefore represent \(\tau\) as a matrix product state
\begin{equation}
  \ket{\tau} = \frac{1}{\norm{\tau_{unnormalized}}} \sum_{j=1}^{N_g} \tilde{c}_j \ket{f_j},
\end{equation}
whose bond dimension is at most \(M = N_g\left( 2m + 3 \right)\).
As with the \(\ket{f_j}\) obtained from \Cref{cor:3d_basis_functions}, \(\ket{\tau}\) is in the standard representation we use for first-quantized simulation in a plane wave basis.
Simplifying the inequality slightly, we find that
\begin{equation}
  M = N_g \left( 2m + 3  \right) \leq 8e^2 N_g \left(
  2 \log \left( \frac{288\sqrt{3}N_g}{\epsilon^4 \sigma^2} \right)
  + \ell \log(4 \ell)  + 4\right),
\end{equation}
establishing the desired bound on the bond dimension.

We can determine the number of qubits required by combining the expression for the momentum cutoff (\(K\)) from the statement of the lemma (or \Cref{eq:K_choice_3d_recalled}) with the statement from \Cref{cor:3d_basis_functions} that any number of qubits (\(n\)) such that \(n \geq 3 \left\lceil \log_2 \left| \mathbb{K}_{cut} \right| \right \rceil\) suffices.
Recall the definition of \(\mathbb{K}_{cut}\) in terms of \(K\),
\begin{equation}
  \mathbb{K}_{cut} = \left\{ k \in \mathbb{K} : \abs{k} \leq K \right\},
\end{equation}
with
\begin{equation}
  \mathbb{K} = \left\{ \frac{2 \pi p}{L} : p \in \mathbb{Z}\right\}.
\end{equation}
We see that 
\begin{equation}
  \left| \mathbb{K}_{cut} \right| = 2 \left\lfloor \frac{KL}{2 \pi} \right\rfloor + 1.
\end{equation}
Asymptotically, we can simplify to find that we require
\begin{equation}
  n \sim 3 \log_2 \left( KL \right)
\end{equation}
qubits, as claimed.

\section{Pedagogical summary of state preparation and unitary synthesis}
\label{app:gadget_appendix}

Here we present a pedagogical summary of the method proposed in Ref.~\citenum{Low2018-uu} for arbitrary state preparation and unitary synthesis, which builds on the fundamental concepts introduced in Ref.~\citenum{Gidney2018-xg,Shende2005-vm,Kliuchnikov2013-il}, as well as an application to the preparation of matrix product states developed in Ref.~\citenum{Fomichev2023-vs}. 
We also summarize the method of Ref.~\citenum{Berry2018-ey} for preparing the antisymmetric initial states required for first quantized simulation. 
Ultimately, we combine these techniques, together with results for Gaussian-plane-wave projections, to produce a full cost accounting for the end-to-end implementation of a first-quantized Slater determinant wavefunction. 
Each section of this appendix builds directly on the previous sections to provide a reference for practical implementations. 
Readers familiar with these topics may wish to skip ahead to the later sections, referring back as needed.

\subsection{Circuit primitives}

Two circuit primitives, a multi-controlled-NOT gate and a controlled multi-qubit SWAP gate, are the dominant sources of Toffoli costs in all of the following state preparation and unitary synthesis routines. Here we describe how both of these circuit primitives can be implemented and provide exact Toffoli counts for each operation.

\subsubsection{Multi-controlled-NOT gate}

Suppose that we have a main register of $n$ qubits and that we wish to apply a bit-flip operation to an auxiliary qubit conditioned on the main register being in the computational basis state $\ket{k}$. We denote this multi-controlled-NOT gate as
\begin{align}
C^n_k(X) = \ket{k}\bra{k}\otimes X + \sum_{k'\neq k}\ket{k'}\bra{k'} \otimes \mathbb{I}.
\end{align}
\begin{figure}
    \centering
    \includegraphics[width=\textwidth]{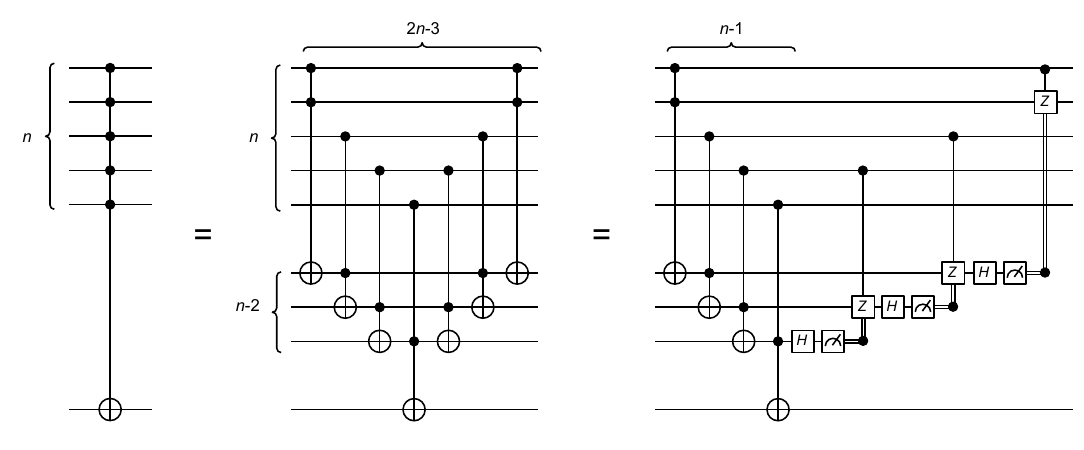}
    \caption{A quantum circuit primitive for the multi-controlled-NOT operation $C^n(X)$, requiring $2n-3$ Toffoli gates and a total of $2n-1$ qubits. The uncomputation staircase is then replaced with measurements and Clifford gates as described in ~\cite{Gidney2018-xg}, reducing the Toffoli cost to $n-1$.}
    \label{fig:cprim1}
\end{figure}
When the bitstring $\ket{k}=\ket{11\cdots1}$ we can simply write $C^n(X)$; any other control bitstring can be obtained by conjugation with Pauli $X$ gates on the main register. A circuit to implement this operation is shown in Figure \ref{fig:cprim1} (a), based on the generic multi-controlled gate construction shown in Figure 4.10 of ~\cite{Nielsen2012-ym}. This circuit uses an additional `work' register of $n-2$ qubits, and requires $n-2$ Toffoli gates for each of the computation and uncomputation steps on the work register, plus a single Toffoli to write to the target register, for a total of $2n-3$ Toffoli gates. The `uncomputation' staircase can also be implemented by measuring each of the work qubits in the $x$-basis and applying Clifford gates dependent on the measurement outcome, as shown in Figure 3 of ~\cite{Gidney2018-xg}, which reduces the Toffoli cost to
\begin{align}
\mathtt{Toffoli}(C^n(X)) = n-1.
\end{align}
Note that a multi-controlled-$Z$ gate can also be implemented by a change of basis on the target qubit, i.e., conjugation of $C^n(X)$ with Hadamard gates on the target qubit.

\subsubsection{Controlled multi-qubit SWAP gate}

Suppose that we have two auxiliary registers of $b$ qubits each, and we wish to swap the states on each register controlled on the value of a single qubit in the main register. We write the action of the multi-qubit SWAP gate on any states $\ket{\phi_0}$ and $\ket{\phi_1}$ of the auxiliary registers as
\begin{align}
S\ket{\phi_0}\otimes\ket{\phi_1} = \ket{\phi_1}\otimes\ket{\phi_0},
\end{align}
\begin{figure}
    \centering
    \includegraphics{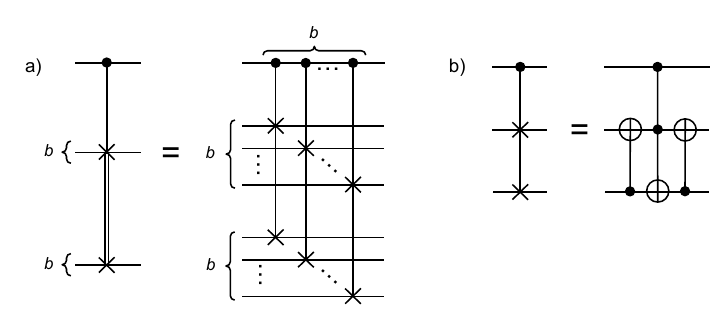}
    \caption{(a) A decomposition of a controlled multi-qubit SWAP gate into a sequence of $b$ controlled SWAP gates (a.k.a. Fredkin gates). We have used a double-line notation to indicate the exchange of multi-qubit registers. A log-depth version of this circuit is described in Ref.~\citenum{Low2018-uu}. (b) The Fredkin gate decomposes into a single Toffoli gate and two CNOT gates.}
    \label{fig:cprim2}
\end{figure}
and the controlled version as 
\begin{align}
C(S) = \ket{0}\bra{0}\otimes\mathbb{I} + \ket{1}\bra{1}\otimes S.
\end{align}
This controlled operation can be performed by swapping pairs of qubits on each register, resulting in the linear sequence of controlled SWAP gates (a.k.a. Fredkin gates) in Figure \ref{fig:cprim2} (a). This can be reduced to logarithmic depth using techniques discussed in Ref.~\citenum{Low2018-uu}, however the construction in Figure \ref{fig:cprim2} (a) is sufficient to obtain the Toffoli count. Each controlled SWAP gate decomposes into a single Toffoli gate and two CNOT gates as in Figure \ref{fig:cprim2} (b), resulting in
\begin{align}
\mathtt{Toffoli}(C(S)) = b.
\end{align}

\subsection{Quantum multiplexor (QMUX) gates}

We now introduce the ``quantum multiplexor'' (QMUX) gates, which are the main engine of arbitrary quantum state preparation. Suppose that we have two registers of qubits -- a main register of $n$ qubits and an auxiliary register of $b$ qubits -- and we wish to apply a different unitary $U_k$ on the auxiliary register dependent on the computational basis state $\ket{k}$ on the main register.
We write this gate as
\begin{align}
    M = \sum_{k=0}^{2^n-1}\ket{k}\bra{k}\otimes U_k.
\end{align}
It is plain to see that this kind of operation can easily generate maximal entanglement between the registers, and these multi-qubit entangling operations are the fundamental building block for state preparation and unitary synthesis. This is because the state preparation routine is built up using multi-controlled rotation multiplexors, the construction of which involves a data lookup oracle multiplexor which writes a different bitstring state to an auxiliary register for each value of $\ket{k}$. As we shall see, this data lookup oracle is the dominant source of Toffoli costs for arbitrary state preparation.

Two possible implementations of a quantum multiplexor gate are summarized here. The first, called the SELECT variant, uses multi-controlled-NOT gates, and the second, called the SWAP variant, uses additional qubit registers and controlled multi-qubit SWAPs. While the arbitrary state preparation routines can be built entirely using either one of these primitives, Low \textit{et al.}~\cite{Low2018-uu} have shown that a reduction in the total Toffoli count can be achieved using a combination of both types of primitives.

\subsubsection{SELECT variant}
\begin{figure}
    \centering
    \includegraphics{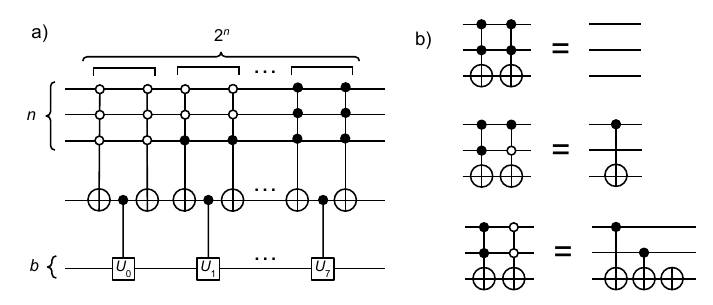}
    \caption{(a) SELECT implementation of $M$, using a single ancilla qubit and $b$ auxiliary qubits. (b) Replacing pairs of Toffoli gates with CNOT and NOT gates in the uncomputation and recomputation steps, the key idea behind unary iteration, reduces the Toffoli cost asymptotically by a factor of $n$.}
    \label{fig:qmux1}
\end{figure}
The first multiplexor implementation, called the SELECT variant, is constructed as follows. Let $C^n_k(X_a)$ denote the multi-controlled-NOT gate flip to an additional ancilla qubit labelled by $a$ given that the main register is in state $\ket{k}$. Following this the controlled-$U_k$ gate $C_a(U_k)$ is applied to the auxiliary register with control from the ancilla qubit, with subsequent uncomputation of $C^n_k(X_a)$. A sequence of $2^n$ such operations implements the QMUX gate,
\begin{align}
    M = \prod_{k=0}^{2^n-1}C^n_k(X_a) C_a(U_k) C^n_k(X_a).
\end{align}
The circuit construction for this operation is shown in Figure \ref{fig:qmux1} (a). Each pair of $C^n(X)$ is implementable with $n-1$ Toffolis, by cancelling Toffoli gates from each staircase and performing the uncomputation steps by the `measure-and-fix-up' approach. The total Toffoli cost is then
\begin{align}
\mathtt{Toffoli}(M) = 2^n\cdot\mathtt{Toffoli}(C^n(X))=2^n(n-1),
\end{align}
in addition to the cost of the controlled-$U_k$ gates. 

However, upon inspection of Figures \ref{fig:qmux1} (a) and \ref{fig:cprim1} (a), one observes further redundancy in the uncomputation and re-computation steps for neighboring pairs of multi-controlled-NOT gates. 
The Toffoli count can be asymptotically improved using the technique in Ref.~\citenum{Babbush2018-tb} known as unary iteration, which boils down to replacing pairs of neighboring Toffoli gates with sequences of CNOT and NOT gates, as shown in Figure \ref{fig:qmux1} (b). 
The Toffoli count is thereby reduced to 
\begin{align}
\mathtt{Toffoli}(M) = 2^n-1.
\end{align}
The total number of qubits required is 
\begin{align}
\mathtt{Qubit}(M) = 2n+b-1
\end{align}
clean qubits.
Given a promise that the main register is never in a state \(\ket{k}\) for any \(k \geq L\), Ref.~\citenum{Babbush2018-tb} also explains how we can reduce the Toffoli cost to \(L-1\).
For simplicity, we do not explain this further optimization here.

\subsubsection{SWAP variant}
\begin{figure}
    \centering
    \includegraphics{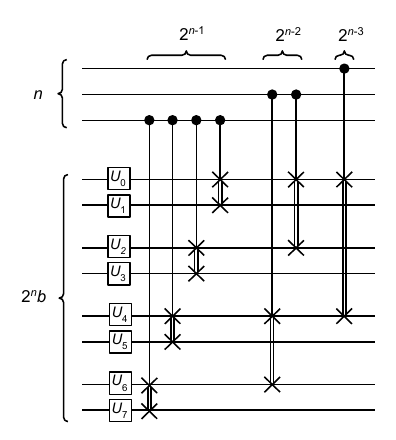}
    \caption{SWAP implementation of $M$, using $2^n$ auxiliary registers of $b$ qubits each.}
    \label{fig:qmux2}
\end{figure}
Now we detail an alternative implementation of a multiplexor gate, called the SWAP variant, which implements the action of $M$ on any initial state $\ket{\phi}$ of the auxiliary register. This variant assumes access to $2^n$ copies of the $b$-qubit state $\ket{\phi}$, requiring $2^nb$ auxiliary qubits. First, each of the unitaries $U_0,...,U_{2^n-1}$ are applied on separate registers. A network of controlled multi-qubit SWAPs is then applied with control from separate qubits of the main register. This brings the state $U_k\ket{\phi}$ from auxiliary register $k$ onto the first auxiliary register on the condition that the main register is in the state $\ket{k}$. More precisely, defining $C_j(S[k,k'])$ to be a controlled multi-qubit swap gate between registers $k$ and $k'$ with control from qubit $j$, we apply 
\begin{align}\label{eq:swap_qmux}
\bar{S}=\prod_{j=1}^n\prod_{k=0}^{2^{j-1}-1}C_j(S[2k\cdot2^{n-j},(2k+1)\cdot2^{n-j}]).
\end{align}
This notation is somewhat opaque. However, the circuit diagram in Figure \ref{fig:qmux2} helps to clarify what this what we mean in Eq.~\eqref{eq:swap_qmux}. The swap network functions somewhat in the spirit of a binary search tree, where the state of each qubit of the main register defines a branching choice over the auxiliary registers. We now have that
\begin{align}
\bar{S}\ket{k}\otimes U_{0}\ket{\phi} \otimes\cdots\otimes U_{2^n-1}\ket{\phi} = \ket{k}\otimes U_{k}\ket{\phi} \otimes\ket{\text{garbage}_k}.
\end{align}
In many cases, the `garbage' on the additional auxiliary registers can be safely ignored when applying an operation in-between applications of $M$ and $M^\dag$, so long as that operation does not interfere with $M^\dag$ uncomputing the garbage.

Each qubit labelled by $j$ on the main register acts as control for $2^{j-1}$ controlled multi-qubit SWAP gates, each of which requires $b$ Toffoli gates. The total Toffoli cost is therefore
\begin{align}
\mathtt{Toffoli}(\bar{S}) = \sum_{j=1}^n 2^{j-1}\cdot\mathtt{Toffoli}(C(S)) = (2^n-1)b,
\end{align}
and the number of clean qubits is 
\begin{align}
\mathtt{Qubit}(\bar{S}) = 2n+2^nb-2.
\end{align}

\subsubsection{SELSWAP (hybrid) variant}
\begin{figure}
    \centering
    \includegraphics[width=\textwidth]{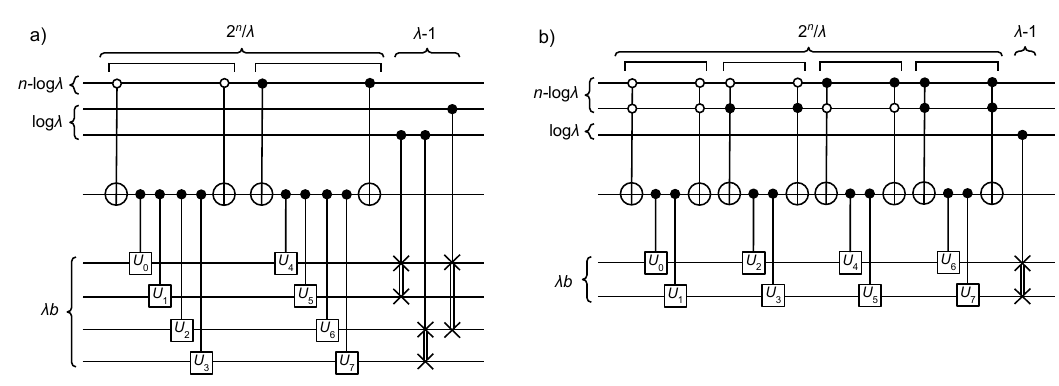}
    \caption{Hybrid SELSWAP implementation of $M$ for an example with $n=3$ and two different choices of $\lambda$: (a) $\lambda=4$,  and (b) $\lambda=2$.}
    \label{fig:qmux3}
\end{figure}
We now discuss how one may `hybridize' the SELECT and SWAP variants using $0\leq \lambda \leq 2^n$ additional registers for some \(\lambda\) that is a power of \(2\),  i.e. $\lambda b$ auxiliary qubits. 
We will split the main register into two components labelled upper ($u$) and lower ($l$), where the upper register consists of $n_u=n-\log\lambda$ qubits and lower register has $n_l=\log\lambda$ qubits. We can then write a bitstring state on the main register as $\ket{k}=\ket{k_u}\otimes\ket{k_l}$. We begin by applying a multi-controlled-NOT gate $C^{n_u}_{k_u}(X_a)$ with control from the upper component of the main register, followed by $\lambda$ controlled unitaries $C_a(U_{k_uk_l})$  on each of the auxiliary registers indexed by $k_l$, and subsequent uncomputation of $C^{n_u}_{k_u}(X_a)$. We would like the state $U_{k_uk_l}\ket{\phi}$ to be moved to the first auxiliary register conditional on the state $\ket{k_u}\otimes\ket{k_l}$ on the main register. This is achieved by the SWAP network $\bar{S}$ consisting of $\lambda-1$ controlled multi-qubit SWAPs with control from the lower component of the main register. Putting these operations together, we write
\begin{align}
    M = \bar{S}\prod_{k_u=0}^{2^n/\lambda-1}C^{n_u}_{k_u}(X_a)\left[\prod_{k_l=0}^{\lambda-1}C_a(U_{k_uk_l})\right]C^{n_u}_{k_u}(X_a),
\end{align}
which implements the action of the multiplexor,
\begin{align}
M\ket{k}\otimes\ket{0}\otimes\ket{\phi}^{\otimes \lambda} = \ket{k}\otimes\ket{0}\otimes U_k\ket{\phi}\otimes\ket{\text{garbage}_k},
\end{align}
as desired. Once again the dense notation becomes much clearer with the help of a circuit diagram. This operation is depicted in Figure \ref{fig:qmux3} for an example with $n=3$ and two different choices of $\lambda$. The SELECT step requires $2^{n+1}/\lambda$ applications of the multi-controlled-NOT circuits $C_{k_u}^{n_u}(X_a)$, which can be efficiently implemented via unary iteration. The second step requires $\lambda-1$ controlled multi-qubit SWAPs. The combined Toffoli cost is therefore  
\begin{align}
\mathtt{Toffoli}(M) = 2^n/\lambda - 1 + (\lambda-1)b,
\end{align}
plus the cost of the controlled-$U$ gates, with a number of clean qubits
\begin{align}
\mathtt{Qubit}(M) = 2n + \lambda b -\log\lambda-1.
\end{align}
Choosing $\lambda\sim 2^\frac{n}{2}$ minimizes the combined Toffoli cost. This procedure can be modified to use `dirty' qubits left over from previous algorithmic steps, meaning that they do not have to be in a pure state, and are returned to their prior state after the operation. This modification requires two applications of the SELECT circuit and four applications of the SWAP network (for further details we refer to Ref.~\citenum{Low2018-uu}, or the exposition in Ref.~\citenum{Berry2019-qo}). With this modification the Toffoli count becomes 
\begin{align}
\mathtt{Toffoli}(M) = 2^{n+1}/\lambda - 2 + 4(\lambda-1)b,
\end{align}
and requires a number of qubits
\begin{align}
\mathtt{Qubit}(M) = 2n + (\lambda+1) b -\log\lambda-1,
\end{align}
of which up to $\lambda b$ may be dirty qubits.

\subsection{Phase rotations via reversible addition}

The reversible addition circuit described in Ref.~\citenum{Gidney2018-xg}, acting on two \(b\)-qubit registers in the computational basis states \(\ket{x}\) and \(\ket{y}\), implements
\begin{align}
    A\ket{x}\otimes\ket{y} = \ket{x}\otimes\ket{y+x \mod 2^b}.
\end{align}
This circuit requires $b$ additional clean qubits to store the temporary values of a sequence of logical-AND operations between pairs of qubits on each register, each of which is computed using a single Toffoli gate. These values are uncomputed using the same `measure-and-fix-up' technique that was used for the multi-controlled-NOT gates. The adder circuit therefore requires $b$ Toffoli gates. The adder circuit can also be given control from an ancilla qubit, which increases the Toffoli cost to $2b$, as discussed in Ref.~\citenum{Gidney2018-xg}.

This circuit can be used to approximate a phase rotation on a target qubit in state $\ket{z}=\alpha\ket{0}+\beta\ket{1}$,
\begin{align}
e^{i\varphi Z}\ket{z} = \alpha\ket{0}+\beta e^{i\varphi}\ket{1},
\end{align}
using two auxiliary $b$-qubit registers. We will first define an approximate phase angle in binary. Given a computational basis state $\ket{y}$ on $b$ qubits, we define the `phase angle'
\begin{align}
    \tilde{\varphi}(y) \equiv \frac{2\pi y}{2^b},
\end{align}
which can approximate any continuous variable angle $\varphi\in [0,2\pi)$ with the appropriate choice of $y$ up to precision determined by the number of qubits,
\begin{align}
    \left|\varphi-\tilde{\varphi}(y)\right| \leq \frac{\pi}{2^b}.
\end{align}
The second key ingredient is the preparation of a Fourier, or `phase gradient', state on $b$ qubits as
\begin{align}
    \ket{f} = \frac{1}{\sqrt{2^b}}\sum_{y=0}^{2^b-1}e^{-i\tilde{\varphi}(y)}\ket{y}.
\end{align}
Now consider that the substitution $y'=y+x\mod 2^b$ implies
\begin{align}
    A\ket{x}\otimes\ket{f} & = \frac{1}{\sqrt{2^b}}\sum_{y=0}^{2^b-1}e^{-i\tilde{\varphi}(y)}\ket{x}\otimes\ket{y+x \mod 2^b}\\
    & = \frac{1}{\sqrt{2^b}}\sum_{y'=0}^{2^b-1}e^{-i\tilde{\varphi}(y'-x)}\ket{x}\otimes\ket{y'}\\
    & = e^{i\tilde{\varphi}(x)}\ket{x}\otimes\ket{f}.
\end{align}

Then to implement the phase rotation, starting in the state $\ket{z}\otimes\ket{0}^{\otimes b}\otimes\ket{f}$, we write the bitstring $\ket{x}$ to the first auxiliary register using the unitary $B$:
\begin{align}
B\ket{z}\otimes\ket{0}^{\otimes b}\otimes\ket{f} = \ket{z}\otimes\ket{x}\otimes\ket{f}.
\end{align}
Then applying the controlled adder circuit $C(A)$ to the auxiliary registers with control from the target qubit yields
\begin{align}
    C(A)B\ket{z}\otimes\ket{0}^{\otimes b}\otimes\ket{f} = (\alpha\ket{0}+\beta e^{i\tilde{\varphi}(x)}\ket{1})\otimes\ket{x}\otimes\ket{f}.
\end{align}
Finally uncomputation of the bitstring yields
\begin{align}
    B^\dag C(A) B\ket{z}\otimes\ket{0}^{\otimes b}\otimes\ket{f} = e^{i\tilde{\varphi}(x)Z}\ket{z}\otimes\ket{0}^{\otimes b}\otimes\ket{f},
\end{align}
as desired.

\subsection{$Z$-rotation multiplexor}

Suppose that we have a quantum register of $n$ qubits, for which we want to apply a $Z$-rotation gate to the $n$th qubit with a different angle for each possible computational basis state $\ket{k}$ of the first $n-1$ qubits,
\begin{align}
    Q = \sum_{k=0}^{2^{n-1}-1}\ket{k}\bra{k}\otimes e^{i\varphi_kZ}.
\end{align}
An approximate implementation of this multiplexor operation is described in Ref.~\citenum{Low2018-uu} as follows.
We require at least $3b$ clean auxiliary qubits to execute the multi-controlled rotation: $b$ qubits to store a Fourier state, $b$ qubits on which a bitstring $\ket{x_k}$ approximating the rotation angle $\varphi_k$ is temporarily stored, and $b$ qubits to store the temporary values computed by the addition circuit.
To begin, the first additional register of $b$ qubits is placed into a superposition of the bitstring states $\ket{x_k}$ conditional on the state $\ket{k}$ of the first $n-1$ qubits on the main register. This is achieved using a `data lookup oracle' multiplexor, written as
\begin{align}
    O = \sum_{k=0}^{2^{n-1}-1}\ket{k}\bra{k}\otimes B_k,
\end{align}
where the `bitstring operators' $B_k$ are strings of single-qubit Pauli rotations $\in\{X,\mathbb{I}\}$, such that
\begin{align}
B_k \ket{0}^{\otimes b} = \ket{x_k}.
\end{align}

We assume prior preparation of a Fourier state on the second register of $b$ qubits, so given some starting state $\ket{\psi}$ of the main register, we have the starting state
\begin{align}
    \ket{\psi}\otimes\ket{0}^{\otimes b}\otimes\ket{f}.
\end{align}
We can also write the state $\ket{\psi}$ on the main register in terms of the unique bitstring states $\ket{k}$ of the first $n-1$ qubits and the corresponding states $\ket{z_k}$ of the last qubit as
\begin{align}
    \ket{\psi} = \sum_{k=0}^{2^{n-1}-1}\psi_k\ket{k}\otimes\ket{z_k},
\end{align}
so that our combined registers are in the state
\begin{align}
    \sum_{k=0}^{2^{n-1}-1}\psi_k\ket{k}\otimes\ket{z_k}\otimes\ket{0}^{\otimes b}\otimes\ket{f}.
\end{align}
We then use the data lookup oracle $O$ to write each bitstring $\ket{x_k}$ encoding $\varphi_k$ to the first auxiliary register dependent on the value of $\ket{k}$, to place our registers into the state
\begin{align}
    \sum_{k=0}^{2^{n-1}-1}\psi_k\ket{k}\otimes\ket{z_k}\otimes\ket{x_k}\otimes\ket{f}.
\end{align}
We now apply the adder circuit to the auxiliary registers controlled on the value of $\ket{z_k}$, which implements the desired phase rotation on the target qubit via phase kickback, before uncomputing $\ket{x_k}$ by applying $O^\dag$, yielding the transformation
\begin{align}
    \sum_{k=0}^{2^{n-1}-1}\psi_k\ket{k}\otimes e^{i\tilde{\varphi}(x_k)Z}\ket{z_k}\otimes\ket{0}^{\otimes b}\otimes\ket{f}.
\end{align}
The `unlookup' $O^\dag$ can be efficiently implemented via a measurement-based approach as described in Appendix C of Ref.~\citenum{Berry2019-qo}, requiring $2^{n}/\lambda + \lambda-2$ Toffolis. 

To summarize, we can approximate the action of the multiplexor $Q$ on the main register with a minimum of roughly $2n+3b$ clean qubits, enabling the computation
\begin{align}
    \ket{\psi}\otimes\ket{0}^{\otimes b}\otimes\ket{f} \rightarrow \tilde{Q}\ket{\psi}\otimes\ket{0}^{\otimes b}\otimes\ket{f}.
\end{align}
We may bound the spectral norm error $\|Q-\tilde{Q}\|$ by observing that
\begin{align}
Q-\tilde{Q} = \sum_{k=0}^{2^{n-1}-1}\ket{k}\bra{k}\otimes (e^{i\varphi_kZ}-e^{i\tilde{\varphi}_kZ}),
\end{align}
which is a diagonal operator with largest singular value bounded by
\begin{align}
\|Q-\tilde{Q}\| &\leq \max_k |e^{i\varphi_k}-e^{i\tilde{\varphi}_k}| \\
&\leq \max_k|\varphi_k-\tilde{\varphi}_k| \\ 
\implies \|Q-\tilde{Q}\| &\leq \pi 2^{-b}.
\end{align}
Neglecting the cost of the Fourier state preparation, the total Toffoli cost is the sum of the contributions from the data oracle multiplexors and the reversible adder circuit. Using the SELSWAP variant of the multiplexor, with additional clean qubits, we have
\begin{align}
\mathtt{Toffoli}(\tilde{Q}) &= \mathtt{Toffoli}(O) + \mathtt{Toffoli}(C(A)) +\mathtt{Toffoli}(O^\dag) \\ &= 
2^{n+1}/\lambda + (b+1)(\lambda-1)+2b-3.
\end{align}
The total number of qubits required is
\begin{align}
\mathtt{Qubit}(\tilde{Q}) = 2n + (\lambda+2) b -\log\lambda-3.
\end{align}

\subsection{Arbitrary state preparation}
\label{app:arbitrary_state_prep}
Given a state $\ket{\psi}$ defined on $n$ qubits, a simple way to think about constructing a circuit to prepare \(\ket{\psi}\) from the \(\ket{0}\) state is to work backwards.
I.e., we can imagine starting with $\ket{\psi}$ and ``disentangling'' the qubits one-by-one via multi-controlled rotations.
Starting with the least significant qubit, we may write
\begin{align}
    \ket{\psi} = \alpha\ket{\psi_0}\otimes\ket{0} + \beta\ket{\psi_1}\otimes\ket{1},
\end{align}
where $\ket{\psi_0}$ and $\ket{\psi_1}$ are states defined on the remaining $n-1$ qubits.
Our goal is to find a quantum gate $Q$ which maps the above state to a new state
\begin{align}
    Q\ket{\psi} = \ket{\psi'}\otimes\ket{0}.
\end{align}
We can also express \(\ket{\psi}\) in the computational basis as
\begin{align}
    \ket{\psi} = \sum_{k=0}^{2^{n-1}-1}\psi_k\ket{k}\otimes(\alpha_k\ket{0}+\beta_k\ket{1}),
\end{align}
where we have chosen $|\alpha_k|^2+|\beta_k|^2=1$.

The single-qubit state on the least significant qubit can be rewritten without loss of generality as
\begin{align}
    \alpha_k\ket{0}+\beta_k\ket{1} = e^{i\varphi_kZ}e^{i\theta_kY}\ket{0},
\end{align}
so it suffices to apply a controlled-Z rotation followed by a controlled-Y rotation to the least significant qubit to rotate it into the zero state.
A different pair of $Z$ and $Y$ rotations must be applied to this target qubit for each computational basis state of the remaining $n-1$ qubits, so a multi-controlled rotation is required:
\begin{align}
    Q = \sum_{k}\ket{k}\bra{k}\otimes e^{-i\theta_kY}e^{-i\varphi_kZ}.
\end{align}
It is then easily verified that
\begin{align}
    Q\ket{\psi} = \sum_k\psi_k\ket{k}\otimes\ket{0}.
\end{align}
We then define $Q_p$ to be the multi-controlled rotation that disentangles the \(p\)th qubit from the first \(p-1\) qubits in this fashion.
It follows from the above that we can write
\begin{align}
    \ket{0}^{\otimes n} & = \left(\prod_{p=1}^{n}Q_p\right)\ket{\psi}
    \\
    \implies \ket{\psi} & = \left(\prod_{p=1}^{n}Q_p\right)^\dag\ket{0}^{\otimes n}.
\end{align}

It was shown in Ref.~\citenum{Shende2005-vm} that this construction is asymptotically optimal for arbitrary state preparation, and thus the problem is reduced to the efficient implementation of the multiplexor gates $Q_p$. In reality we use the $Z$-rotation multiplexor implementation outlined in the previous subsection, with the $Y$-rotation multiplexor obtained via a single-qubit change of basis, to prepare the state
\begin{align}
    |\tilde{\psi}\rangle & = \left(\prod_{p=1}^{n}\tilde{Q}_p\right)^\dag\ket{0}^{\otimes n}.
\end{align}
We may bound the error in the state preparation by considering that
\begin{align}
\tilde{Q}_p = Q_p + (\tilde{Q}_p-Q_p),
\end{align}
and that $Q_p$ and $\tilde{Q}_p$ are both unitary, from which we may use standard spectral norm identities to arrive at
\begin{align}
\|\prod_p Q_p - \prod_p \tilde{Q}_p\|\leq \sum_p\|Q_p-\tilde{Q}_p\|.
\end{align}
Using our previous result for the errors in the rotation multiplexors we may then bound
\begin{align}
\|\ket{\psi}-|\tilde{\psi}\rangle\| &\leq 2\pi n\cdot 2^{-b},
\end{align}
where \(b\) is the number of bits used to specify the rotation angles.

The Toffoli cost of state preparation is given by
\begin{align}
\mathtt{Toffoli}(|\tilde{\psi}\rangle) = 2\sum_{p=1}^n 2^{p+1}/\lambda_p + (b+1)(\lambda_p-1)+2b-3
\end{align}
and a number of qubits
\begin{align}
\mathtt{Qubit}(|\tilde{\psi}\rangle) = 2n + (\lambda_n+3) b -\log\lambda_n-3.
\end{align}
We can minimize the above expression by setting $\lambda_p\propto2^\frac{p}{2}$. We shall assume the same proportionality constant $\mu$ for all of the $\lambda_p$. In reality however we must choose $\lambda_p$ to be an exact power of two, i.e.,
\begin{align}
\lambda_p = 2^{\lceil \log (\mu2^{p/2}) \rceil},
\end{align}
which implies that
\begin{align}
\mu2^\frac{p}{2}\leq\lambda_p\leq\mu2^{\frac{p}{2}+1}.
\end{align}
We can in this case strictly bound the Toffoli count by
\begin{align}
\mathtt{Toffoli}(|\tilde{\psi}\rangle) < 4(1+\sqrt{2})\left(\frac{1}{\mu}+\mu (b+1)\right)\cdot 2^{\frac{n}{2}+\frac{3}{2}}+2n(b-4),
\end{align}
which is minimized by $\mu=(b+1)^{-1/2}$, in which case
\begin{align}
\mathtt{Toffoli}(|\tilde{\psi}\rangle) < (1+\sqrt{2})\left(2^{n+7}(b+1)\right)^{\frac{1}{2}}+2n(b-4).
\end{align}
The number of (clean) qubits is then bounded by
\begin{align}
\mathtt{Qubit}(|\tilde{\psi}\rangle) < \frac{3n}{2} + \frac{2^{\frac{n}{2}+1}b}{\sqrt{b+1}} + \frac{1}{2}\log(b+1) + 3b -4.
\label{eq:state_prep_app_qubit_count}
\end{align}
Note that these bounds are not expressed as whole numbers; the number of Toffolis and qubits will be bounded by the floor of these quantities.

\subsection{Unitary synthesis}
\label{app:unitary_synthesis_description}

Following Refs.~\citenum{Low2018-uu,Kliuchnikov2013-il}, we can implement an arbitrary $n-$qubit unitary $U$ using an additional ancilla qubit via a product of reflection operators as
\begin{align}
    W = \prod_{k=0}^{2^n-1} (\mathbb{I}-2\ket{w_k}\bra{w_k}),
\end{align}
where $\ket{w_k}$ is defined in terms of the $k$-th column $\ket{u_k}$ of $U$, and the $k$-th bistring state $\ket{k}$, as
\begin{align}
    \ket{w_k} = \frac{1}{\sqrt{2}}(\ket{1}\otimes\ket{k}-\ket{0}\otimes\ket{u_k}).
\end{align}
To see why this works, consider that
\begin{align}
(\mathbb{I}-2\ket{w_k}\bra{w_k})\ket{1}\otimes\ket{k} &= \ket{0}\otimes\ket{u_k} \\
\implies W\ket{1}\otimes\ket{k} &= \ket{0}\otimes\ket{u_k},
\end{align}
where the last line follows since the terms in $\ket{k'}\neq\ket{k}$ simply evaluate to zero. The same arguments also show that $W\ket{0}\otimes\ket{u_k}=\ket{1}\otimes\ket{k}$, and it is shown in Ref.~\citenum{Kliuchnikov2013-il} that
\begin{align}
W = \ket{0}\bra{1}\otimes U + \ket{1}\bra{0}\otimes U^\dag.
\end{align}
This means that the action of $U$ on an input state $\ket{\psi}$ is implemented by
\begin{align}
W\ket{1}\otimes\ket{\psi} = \ket{0}\otimes U\ket{\psi}.
\end{align}

Given a unitary \(W_k\) such that \(\ket{w_k} = W_k \ket{0}\), we can write the reflection about \(\ket{w_k}\) as
\begin{align}
    \mathbb{I} - 2\ket{w_k}\bra{w_k} = W_k(\mathbb{I}-2\ket{0}\bra{0}^{\otimes n+1})W_k^\dag = W_kR_0W_k^\dag,
\end{align}
where $R_0 = \mathbb{I} - 2\ket{0}\bra{0}$ is the reflection about the zero state on \(n+1\) qubits.
It is easily seen that $W_k$ can be constructed from the unitaries
\begin{align}
U_k\ket{0}^{\otimes n} = \ket{u_k}, \quad B_k\ket{0}^{\otimes n} = \ket{k},
\end{align}
by adding a control from the ancilla qubit. 
We can then write
\begin{align}
W_k = C_0(U_{k})C_1(B_{k})(XHX\otimes\mathbb{I}_n),
\end{align}
where $C_0(U_k)$ is controlled on the $\ket{0}$ state of the ancilla qubit, while $C_1(B_k)$ is controlled on the $\ket{1}$ state.

The reflection $R_0$ is equivalent to a multi-controlled-$Z$ gate on the $n+1$-th qubit with control from the first $n$ qubits, which is implementable with $n$ Toffolis and $n-1$ work qubits using the multi-controlled-NOT circuit primitive in Figure \ref{fig:cprim1}. The $C(B_{k})$ operators require only Clifford gates, and we can compile the $n-$qubit unitaries $U_k$ to prepare each $\ket{u_k}$ from the all-zero state via the state preparation routine described above. 

We may now derive a bound on the 2-norm error incurred due to approximations in state preparation circuits. Let
\begin{equation}
    \tilde{W} = \prod_{k}(\mathbb{I}-2\ket{\tilde{w}_k}\bra{\tilde{w}_k}),
\end{equation}
where each \(\ket{w_k}\) is defined by replacing the desired \(\ket{u_k}\) with some state \(\ket{\tilde{u_k}}\).
Then we can bound the error in \(\tilde{W}\) by 
\begin{equation}
    \norm{W - \tilde{W}} \leq \sum_k \norm{\left( \mathbb{I} - 2 \ketbra{w_k} \right) - \left( \mathbb{I} - 2 \ketbra{\tilde{w_k}} \right)} = 2 \sum_k  \norm{\ketbra{w_k}- \ketbra{\tilde{w_k}}}.
\end{equation}
We can bound the spectral norm of each term \(\left( \ketbra{w_k} - \ketbra{\tilde{w_k}} \right)\) by 
\begin{align}
\|\ketbra{w_k} - \ketbra{\tilde{w}_k}\| &= \|\ket{w_k}\left(\bra{w_k}-\bra{\tilde{w}_k}\right) + \left(\ket{w_k}-\ket{\tilde{w}_k}\right)\bra{\tilde{w}_k}\| \\
&\leq 2\|\ket{w_k}-\ket{\tilde{w}_k}\|
\end{align}

to show that 
\begin{equation}
    \norm{W - \tilde{W}} \leq 4 \sum_k \|\ket{w_k}-\ket{\tilde{w}_k}\|.
\end{equation}

We have previously bounded the errors in arbitrary state preparation via the rotation multiplexors by
\begin{align}
\|\ket{u_k}-\ket{\tilde{u}_k}\| < 2\pi n\cdot2^{-b},
\end{align}
which allows us to bound
\begin{equation}
    \norm{\ket{w_k} - \ket{\tilde{w}_k}} < 2 \pi \sqrt{2}n\cdot2^{-b},
\end{equation}
which in turn implies that 
\begin{equation}
    \norm{W - \tilde{W}} \leq 8 \pi \sqrt{2}n\cdot2^{n-b}.
\end{equation}

Adding control to each $\tilde{U}_k$ is achieved by adding an extra control qubit to the controlled adder circuit. This also requires a work qubit to store the logical-AND value of both the original and the new control qubit. The controlled adder circuit will then only perform addition if the work qubit is in the $\ket{1}$ state, i.e., if both the original control qubit and the newly added control qubit are simultaneously in the $\ket{1}$ state. This modification increases the Toffoli cost of each controlled adder circuit from $2b$ to $3b$. Then choosing $\mu=(b+1)^{-1/2}$ as before to minimize the Toffoli count, we have that
\begin{align}
\mathtt{Toffoli}(\tilde{W}) < 2^{\frac{3n}{2}+\frac{9}{2}}(1+\sqrt{2})(b+1)^\frac{1}{2} + 2^{n}n(8b-15).
\end{align}

The work qubits for the data lookup oracles and controlled additions can be re-used for the zero-state reflections. Therefore only three additional qubits are required in addition to those required for state preparation (see \cref{eq:state_prep_app_qubit_count}), namely the ancilla qubit and the two new control qubits. 

One may also add control to the unitary synthesis by controlling the zero-state reflections. This is achieved by extending the `Toffoli staircase', at the cost of a single additional work qubit and an additional Toffoli gate per reflection operator, i.e., a sub-dominant addition of $2^n$ Toffoli gates.

\subsection{MPS state preparation}
\label{app:mps_state_prep_details}

Using the approach introduced in Ref.~\citenum{Schon2005-hc} and costed out in terms of Toffoli gates in Ref.~\citenum{Fomichev2023-vs}, a matrix product state 
\begin{align}
\ket{\tau} = \sum_{\left\{ s \right\}, \left\{ \alpha \right\}} A_{\alpha_1}^{s_1} A_{\alpha_1 \alpha_2}^{s_2} \cdots A_{\alpha_{n-1}}^{s_n} \ket{s_1 s_2 \cdots s_n}
\end{align}
can be encoded into a quantum circuit by encoding its tensors into unitary operators. To see how this encoding scheme works, consider the tensor $A^{s_j}_{\alpha_{j-1}\alpha_j}$ of the MPS at site $j$, with left and right bond dimensions equal to $m_{j-1}$ and $m_j$, which differ by up to a factor of two. We define 
\begin{align}
\bar{m}_j \equiv \max(2^{\lceil\log m_{j-1}\rceil},2^{\lceil\log m_j\rceil}),
\end{align}
and pad the tensor with zeros so that it has $2\times \bar{m}_j\times \bar{m}_j$ entries, with $\bar{m}_j$ now being an exact power of two. Grouping the site index $s_j$ with the bond index $\alpha_{j-1}$ allows us to write the tensor at site $j$ as a rectangular $2\bar{m}_j\times \bar{m}_j$ matrix,
\begin{align}
A^{s_j}_{\alpha_{j-1}\alpha_j} = \sum_{\alpha_{j-1}s_j}^{2\bar{m}_j}\sum_{\alpha_j}^{\bar{m}_j} a^{s_j}_{\alpha_{j-1}\alpha_j} \ket{\alpha_{j-1}}\otimes\ket{s_j}\bra{\alpha_{j}}.
\end{align}
Note that only the first $m_j$ columns are non-zero, the rest being padded with zero entries. We will assume that the MPS is in left-orthogonal form. This means that, denoting the above matricized tensor with $A\equiv A^{s_j}_{\alpha_{j-1}\alpha_j}$, we have
\begin{align}
A^\dag A = \sum_{\alpha_j=1}^m\ket{\alpha_j}\bra{\alpha_j} = \mathbb{I}_{m_j}\oplus\bm{0}_{\bar{m}_j - m_j}.
\end{align}
This implies that the first $m_j$ columns of $A$ are a set of orthonormal vectors. Hence we may construct a $2m\times2m$ unitary $U$ that implements 
\begin{align}
U\ket{0}\otimes\ket{\psi} = A\ket{\psi}
\end{align}
by choosing the first $m_j$ columns identical to those of $A$ and the remaining $2\bar{m}_j-m_j$ columns to be any set of orthonormal vectors in the kernel space.

Then defining
\begin{align}
\bar{U}^{[j]}\equiv \mathbb{I}_{2}^{\otimes j-\log \bar{m}_j-1}\otimes U^{[j]} \otimes \mathbb{I}_{2}^{\otimes n-j},
\end{align}
we obtain the desired unitary product
\begin{align}
\bar{U}^{[1]}\cdots\bar{U}^{[n]}\ket{0}^{\otimes n} = \sum_{\left\{ s \right\}, \left\{ \alpha \right\}} A_{\alpha_1}^{s_1} A_{\alpha_1 \alpha_2}^{s_2} \cdots A_{\alpha_{n-1}}^{s_n} \ket{s_1 s_2 \cdots s_n} = \ket{\tau},
\label{eq:mps_synth}
\end{align}
which is illustrated in Figure \ref{fig:mps_diagram} of the main text. Note that the tensor product construction in Equation \ref{eq:mps_synth} requires no additional ancilla qubits.

The unitary $U^{[j]}$ can be implemented via a sequence of $m_j$ reflection operators about each of the column vectors, following the unitary synthesis technique of ~\cite{Low2018-uu} described above in Appendix \ref{app:unitary_synthesis_description}. The error incurred in the MPS preparation is then bounded by
\begin{align}
\|\ket{\tau}-|\tilde{\tau}\rangle\| &\leq \sum_{j=1}^n \|W^{[j]}-\tilde{W}^{[j]}\|\\
\implies \|\ket{\tau}-|\tilde{\tau}\rangle\| &< 2^{\frac{7}{2}-b}\sum_{j=1}^nm_j\log (2\bar{m}_j).
\end{align}
Using the tradeoff ratio $\mu=(b+1)^{-1/2}$ in the state preparation (see Appendix \ref{app:arbitrary_state_prep}), the Toffoli cost for the MPS preparation is bounded by
\begin{align}
\mathtt{Toffoli}(\ket{\tilde{\tau}}) < \sum_{j=1}^n 32(1+\sqrt{2})(b+1)^\frac{1}{2}m_j\bar{m}_j^\frac{1}{2} + (8b-15)m_j\log (2\bar{m}_j).
\label{eq:mps_state_prep_app_toff_count}
\end{align}

Note that if $m_j = \text{poly}(n)$ then the Toffoli count grows only polynomially in $n$. Contrast this scaling with that of arbitrary state preparation, which has a Toffoli cost that is exponential in $n$. Also note that since each column vector has $2m_{j-1}$ nonzero entries, then one may further optimize the Toffoli count of the state preparation. For brevity we omit this further optimization here.
\subsection{Antisymmetrization}

First quantized simulations require fermionic antisymmetry to be enforced at the start of the computation, after which all time-evolved states correctly preserve this property. Here we detail the method introduced in Ref.~\citenum{Berry2018-ey} for preparing an antisymmetrized register of bitstring states before explaining the application of this technique to Slater determinant state preparation in the last section.

Concretely, let 
\begin{align}
|\vec{x}\rangle = \ket{x_1}\otimes\cdots\otimes\ket{x_\eta}
\end{align}
be a tensor product of the computational basis states $\ket{x_p}$ on particle registers $p=1,...,\eta$. We assume that the integers represented by the basis states are sorted, i.e., $x_1<x_2<...<x_\eta$.
We aim to prepare the state
\begin{align}
|\vec{x}_\text{AS}\rangle = \frac{1}{\sqrt{\eta!}}\sum_{\sigma}(-1)^{\pi(\sigma)}\ket{\sigma(\vec{x})},
\end{align}
where $\sigma$ runs over all $\eta!$ possible permutations of the particle registers, \(\ket{\sigma(\vec{x})}\) denotes the state \(\ket{\vec{x}}\) with the registers permuted according to \(\sigma\), and $\pi(\sigma)$ is the parity of the permutation $\sigma$. The approach we shall introduce makes use of a sorting network, which is built up from comparator circuits. A comparator circuit $C$ simply compares the values of two bitstrings $\ket{x_p}$ and $\ket{x_q}$, and swaps their positions if they are in the wrong order. The comparator also stores the outcome of the comparison in a separate register in order to make the operation reversible. The action of $C$ on a pair of computational basis states is given by
\begin{align}
C\ket{x_p}\otimes\ket{x_q}\otimes\ket{0} = \begin{cases}
\ket{x_p}\otimes\ket{x_q}\otimes\ket{0} & x_p \leq x_q, \\
\ket{x_q}\otimes\ket{x_p}\otimes\ket{1} & x_p > x_q.
\end{cases}
\end{align}

A sorting network $\bar{C}$ consists of a sequence of these pairwise comparator circuits arranged to sort an arbitrary product of computational basis on \(\eta\) registers.
For example, Batcher's odd-even mergesort requires $O(\eta\log^2\eta)$ comparators, many of which can be executed in parallel, reducing the circuit depth to $O(\log^2\eta)$. Note that each comparator circuit must write the outcome of its comparison to a separate qubit, so in this case a register of $O(\eta\log^2\eta)$ additional qubits are required to store the comparator values.

Let \(\ket{\vec{y}}\) be a product of computational basis states that arises from applying some permutation \(\ket{\sigma_y}\) to the state \(\ket{\vec{x}}\) (which is sorted in increasing order).
Then let \(\ket{y}\) be the computational basis state defined implicitly by the expression
\begin{equation}
    \bar{C} \ket{\vec{y}} \ket{0} = \ket{\vec{x}} \ket{y}.
\end{equation}
In other words, each bit of \(\ket{y}\) specifies whether or not one of the comparators had to apply a transposition when sorting \(\ket{\vec{y}}\) into \(\ket{\vec{x}}\).
\(\bar{C}\) is reversible, so we could also apply the transposition \(\sigma_y\) to \(\ket{x}\) using \(\bar{C}^\dagger\),
\begin{equation}
    \bar{C}^{\dagger} \ket{\vec{x}} \ket{y} = \ket{\vec{y}} \ket{0} = \sigma_y \ket{\vec{x}} \otimes \ket{0}.
\end{equation}

Since the desired antisymmetric state is a superposition over different permutations of \(\ket{\vec{x}}\), there exists some state \(\ket{\upsilon} = \frac{1}{\sqrt{\eta !}} \sum_{y} \ket{y}\) such that
\begin{align}
\bar{C}^\dag\ket{\vec{x}}\otimes(\frac{1}{\sqrt{\eta!}}\sum_y\ket{y}) &= \frac{1}{\sqrt{\eta!}}\sum_y \ket{\sigma_y(\vec{x})}\otimes\ket{0}^{\otimes \eta\log^2\eta} \\ &= \ket{\vec{x}_\text{AS}}\otimes\ket{0}^{\otimes \eta\log^2\eta}.
\end{align}
By the sum over \(y\), we mean that we sum over all \(y\) corresponding to unique permutations.
This is almost the desired state, except that it is missing the correct phase factors. We can add a controlled phase gate after each swap to obtain a modified reverse sorting network \(\bar{\bar{C}}\) such that 
\begin{align}
    \bar{\bar{C}}^\dag\ket{\vec{x}}\otimes(\frac{1}{\sqrt{\eta!}}\sum_y\ket{y}) &= \frac{1}{\sqrt{\eta!}}\sum_y(-1)^{\pi(\sigma_y)}\ket{\sigma_y(\vec{x})}\otimes\ket{0}^{\otimes \eta\log^2\eta} \\ &= \ket{\vec{x}_\text{AS}}\otimes\ket{0}^{\otimes \eta\log^2\eta}.
\end{align}
Then the problem of preparing $\ket{\vec{x}_\text{AS}}$ reduces to the problem of preparing the state $\ket{\upsilon}$.

We now explain how to prepare the state $\ket{\upsilon}$ on the control register using a forward sorting network and a projective measurement. 
The key idea is that, given an unsorted input state, which corresponds to a particular permutation of the corresponding sorted state, the forward comparator network will write a bitstring to the control register that enacts the permutation when applied in reverse. 
The precise values of the permuted elements are not important, as only their ordering matters for determining the action of the sorting network, so the same bitstring state $\ket{y}$ corresponds to the same permutation regardless of the exact values. 
For example, swapping the elements $(3,1)\rightarrow(1,3)$ is exactly the same permutation as $(2,1)\rightarrow(1,2)$.

Given that the values are not important, it is sufficient to prepare an input state that is an equal superposition of states with all possible orderings over $\eta$ arbitrary elements, such that each ordering appears with equal likelihood. Let $f$ be some number such that $f\geq\eta$. 
We prepare a `seed' register in an equal superposition of all $f^\eta$ product states over $\eta$ registers such that the state on each register is one of $f$ possible bitstrings. 
The seed register then requires $\eta\log f$ qubits. We may group these terms into two categories: those that contain duplicate elements on different registers (`collisions'), and those that do not contain duplicate elements (``collision-free''), so we can write the state on the seed register as the sum of these two grouped terms:
\begin{align}
\ket{s} = \ket{s_\text{free}} + \ket{s_\text{coll}}
\end{align}

The relative number of strings in the second grouping decreases with increasing $f$ according to
\begin{align}
|\braket{s_\text{coll}}{s_\text{coll}}|^2 = \frac{\eta(\eta-1)}{2f}.
\end{align}
Hence, applying a projective measurement onto the subspace of $\ket{s_\text{free}}$ has a success probability of greater than one half when $f>\eta^2$. The state $\ket{s_\text{free}}$ then has the stated property that each ordering over non-duplicate elements appears with equal likelihood. Given that the state recorded on the output register of the comparator network is independent of the particular values of the bitstrings, we must have that
\begin{align}
\bar{C}\ket{s_\text{free}}\otimes\ket{0}^{\otimes\eta\log^2\eta} = \ket{s_\text{sort}}\otimes\ket{\upsilon}.
\end{align}

We may now discard $\ket{s_\text{sort}}$ since it is disentangled from $\ket{\upsilon}$. Note that we may equivalently perform the projective measurement onto the collision-free subspace after applying $\bar{C}$ to $\ket{s}$, which does not cause any interaction between the subspaces. In this case the projective measurement is easy to implement with $\eta-1$ comparison operations between neighboring registers, each of which writes the value 1 to a separate qubit on a register of $\eta-1$ qubits if the bitstrings are identical (hence recording a collision). Then measuring the all-zero state on this register implies projection onto the collision-free subspace, and hence projection onto $\ket{s_\text{sort}}\otimes\ket{\nu}\otimes\ket{0}^{\otimes \eta-1}$.

The Toffoli cost for preparing $\ket{\vec{x}_\text{AS}}$ is dominated by the sorting steps. As discussed in Ref.~\citenum{Berry2018-ey}, the optimal sorting networks use a number of comparators that scales as \(\mathcal{O}(\eta \log \eta)\) and each comparator requires \(\mathcal{O}(\log N)\) Toffoli gates, so the overall cost scales as \(\mathcal{O}(\eta \log \eta \log N)\).
In general this scaling is subdominant to the other steps that we consider in our state preparation scheme and the constant factors are relatively small, so we usually neglect this cost entirely.

 \subsection{Slater determinant state preparation}
 \label{app:sd_preparation}

\begin{figure}
    \centering
    \includegraphics[width=0.375\textwidth]{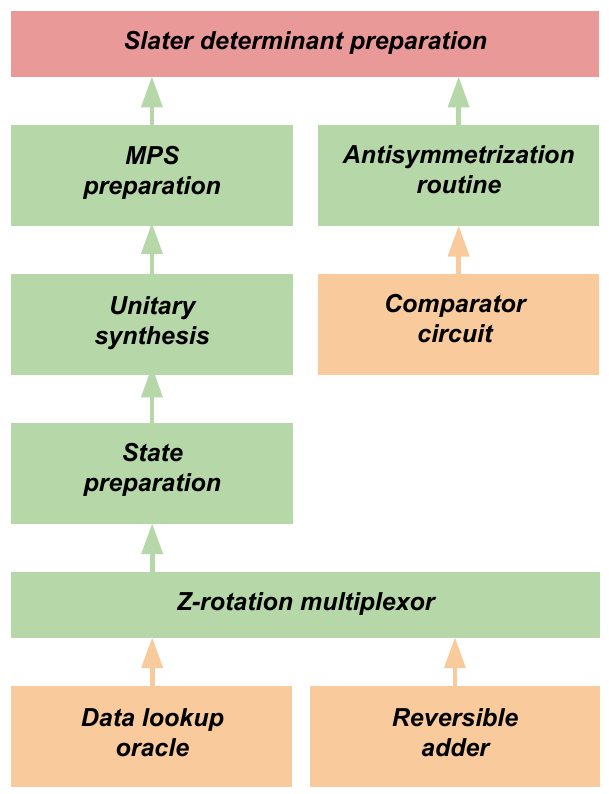}
    \caption{A schematic of the structure of the Slater determinant state preparation routine. The subroutines in the orange boxes are the primary contributors to the Toffoli complexity. The unitary synthesis and Slater determinant preparation routines also make use of zero-state reflections which introduce additional sub-dominant terms in the Toffoli count. We have omitted the Fourier state preparation which is a one-time cost and is logarithmic in the system size parameters $N$ and $\eta$ and in the inverse target error $\epsilon^{-1}$. By far the dominant source of Toffoli gates is the data lookup oracle subroutine, which can be minimized at the cost of using additional ancilla qubits.}
    \label{fig:sd_prep}
\end{figure}

Here we explain the Slater determinant state preparation routine, which relies on the subroutines introduced previously, as summarized in Figure \ref{fig:sd_prep}. This routine relies on the same unitary synthesis techniques previously introduced in \Cref{app:unitary_synthesis_description}, except that we use the MPS preparation routines introduced \Cref{app:mps_state_prep_details} in place of techniques for preparing arbitrary states.
More concretely, we implement the unitary 
\begin{align}
    X = \prod_{k=1}^{\eta} (\mathbb{I}-2\ket{x_k}\bra{x_k}),
\end{align}
where
\begin{align}
    \ket{x_k} = \frac{1}{\sqrt{2}}(\ket{1}\otimes\ket{k}-\ket{0}\otimes\ket{\tilde{\tau}_k}),
\end{align}
and $\ket{\tilde{\tau}_k}$ is the actual state prepared by a quantum circuit that approximates a matrix product state \(\ket{\tau_k}\) encoding the $k$-th molecular orbital. 

We begin by preparing the system register in the antisymmetrized reference state $|\vec{k}_\text{AS}\rangle$, 
\begin{align}
|\vec{k}_\text{AS}\rangle = \frac{1}{\sqrt{\eta!}}\sum_{\sigma}(-1)^{\pi(\sigma)}|\sigma(\vec{k})\rangle,
\end{align}
where
\begin{align}
|\vec{k}\rangle = \ket{1}\otimes\cdots\otimes\ket{\eta},
\end{align}
and $|\sigma(\vec{k})\rangle$ denotes a permutation $\sigma$ of the computational basis states on the registers $1,...,\eta$ with parity $\pi(\sigma)$. 
We can then apply $X^{\otimes\eta}$, which approximately yields the state
\begin{align}
|\vec{\tau}_\text{AS}\rangle \approx \frac{1}{\sqrt{\eta!}}\sum_{\sigma}(-1)^{\pi(\sigma)}|\sigma(\vec{\tau})\rangle,
\label{eq:vec_tau}
\end{align}
where
\begin{align}
|\vec{\tau}\rangle = \ket{\tilde{\tau}_1}\otimes\cdots\otimes\ket{\tilde{\tau}_\eta}.
\end{align}
The \(\ket{\tilde{\tau}_m}\) states may be slightly non-orthogonal because of errors in the approximations of each molecular orbital.
This potential non-orthogonality is the source of the approximation in \Cref{eq:vec_tau}, since \(X\) does not implement an isometry if the \(\ket{\tilde{\tau}_m}\) are not orthogonal.

Using the MPS state preparation routine of \Cref{app:mps_state_prep_details} for all of the states, and ignoring the effect of the finite computational unit cell which we assume to be negligible, we have two sources of error in the MOs. The first of these is the combined error due to plane wave projection and singular value truncation, which is quantified by the trace distance,
\begin{align}
D(\ket{\chi_k}, \ket{\tau_k}) \leq \epsilon_1.
\end{align}
The second source of error is incurred during state preparation via quantum circuits, and is due to the discretized angles of the rotation multiplexors, quantified by
\begin{align}
D(\ket{\tau_k}, \ket{\tilde{\tau}_k}) \leq \|\ket{\tau_k}-\ket{\tilde{\tau}_k}\| \leq \epsilon_2.
\end{align}
By the triangle inequality, the overall error is bounded by
\begin{align}
D(\ket{\chi_k}, \ket{\tilde{\tau}_k}) \leq \epsilon_1 + \epsilon_2.
\end{align}

Let
\begin{equation}
    W = \prod_{m=1}^{N_{mo}} \mathbb{I} - 2 \ketbra{w_m},
    \;\;\;\;\;\;\;\;
    \ket{w_m} = \frac{1}{\sqrt{2}} \left( \ket{1}\ket{m} - \ket{0}\ket{\chi_m} \right).
\end{equation}
In the absence of error, \(W^{\otimes \eta}\ket{1}\ket{\vec{k}_{AS}}\) would prepare the desired Slater determinant, which we denote by \(\ket{\vec{\chi}_{AS}}\).
We would like to quantify the error between the state prepared by using \(X^{\otimes \eta}\) and the state \(\ket{\vec{\chi}_{AS}}\).
Rather than produce a worst-case bound, we make the assumption that the \(\ket{\tilde{\tau}_m}\) are orthogonal, or can be made orthogonal with negligible additional error.
Under this assumption, \Cref{eq:vec_tau} is exact and we can focus on comparing the distance between the two Slater determinants \(\ket{\vec{\tau}_{AS}}\) and \(\ket{\vec{\chi}_{AS}}\).

The overlap of two Slater determinants is given by the determinant of the matrix of overlaps of the individual orbitals~\citenum{Plasser2016-rn}.
Making the additional assumption that the overlaps between \(\ket{\chi_j}\) and \(\ket{\tilde{\tau}_k}\) are small for any \(j \neq k\) (which is true whenever the approximate molecular orbitals are close to the true ones), then
\begin{equation}
    \braket{\vec{\chi}_{AS}}{\vec{\tau}_{AS}} \approx \prod_{m=1}^{\eta} \braket{\chi_m}{\tilde{\tau}_m}.
\end{equation}
This tells us that we can approximate the fidelity between the two states in terms of the fidelity between \(\bigotimes_{m=1}^{\eta} \ket{\chi_m}\) and \(\bigotimes_{m=1}^{\eta} \ket{\tilde{\tau}_m}\).
Therefore, we can also approximate the trace distance between the two determinants in terms of the trace distance between these two product states,
yielding
\begin{equation}
    D\left(\ket{\vec{\chi}_{AS}}, \ket{\vec{\tau}_{AS}}\right) \approx D\left(\bigotimes_{m=1}^{\eta} \ket{\chi_m}, \bigotimes_{m=1}^{\eta} \ket{\tilde{\tau}_m}\right) \leq \sum_{m=1}^{\eta} D(\ket{\chi_m}, \ket{\tilde{\tau}_m}).
    \label{eq:trace_distance_error_rough_approx}
\end{equation}

Altogether then, if \(D(\ket{\chi_k}, \ket{\tilde{\tau}_k}) \leq \epsilon_1+\epsilon_2\) for all \(k\), we have that 
\begin{equation}
  D\left(\ket{\vec{\chi}_{AS}}, \ket{\vec{\tau}_{AS}}\right) \approx \eta \left(\epsilon_1 + \epsilon_2\right).
\end{equation}
Using \Cref{lemma:mo_as_mps}, we can get conditions on the number of qubits and the bond dimension in order to bound the error term \(\epsilon_1\) to within any desired accuracy.
The $\epsilon_2$ error term is bounded by 
\begin{align}
\epsilon_2 < 2^{\frac{7}{2}-b}\sum_{j=1}^{n}m_j\log (2\bar{m}_j),
\end{align}
where \(b\) is the number of bits used for the rotation synthesis and the \(m_j\) are the bond dimensions of the tensors.
Alternatively, we can numerically characterize one or both sources of error for a particular application.

The overall cost is a factor of \(\eta\) times the cost of implementing a single \(W\).
We can account for the cost of a single \(W\) by adding up the costs of controlled preparation and unpreparation for each of the $\ket{\tilde{\tau}_m},$ plus the cost of the reflections about the \(\ket{0}\) states.
Overall then, the number of Toffoli gates required is 
\begin{align}
\mathtt{Toffoli}(X^{\otimes \eta}) &= \eta^2n + 2\eta\sum_{k=1}^\eta\mathtt{Toffoli}(\ket{\tilde{\tau}_k}) \\
&<\eta^2n + 2\eta\sum_{k=1}^\eta\sum_{j=1}^n 32(1+\sqrt{2})(b+1)^\frac{1}{2}m_j\bar{m}_{j}^\frac{1}{2} + (8b-15)m_j\log (2\bar{m}_{j}).
\label{eq:toff_count_state_prep_mps_sd}
\end{align}

The error analysis we performed above required some additional assumptions.
We could instead perform a more precise but more pessimistic error analysis that seeks to bound \(\norm{W^{\otimes \eta} - X^{\otimes \eta}}\) in the spectral norm.
To begin, we note that
\begin{equation}
  \norm{W^{\otimes \eta} - X^{\otimes \eta}} \leq \eta \norm{W - X}, 
\end{equation}
and
\begin{equation}
  \norm{W - X} \leq \sum_k \norm{\left( \mathbb{I} - 2 \ketbra{w_k} \right) - \left( \mathbb{I} - 2 \ketbra{x_k} \right)} = 2 \sum_k  \norm{\ketbra{w_k}- \ketbra{x_k}}.
\end{equation}
We can bound the spectral norm of each term \(\left( \ketbra{w_k} - \ketbra{x_k} \right)\) by the trace norm to show that 
\begin{equation}
  \norm{W - X} \leq 4 \sum_k D(\ket{w_k}, \ket{x_k}).
\end{equation}
We can further simplify to find that 
\begin{equation}
  D(\ket{w_k}, \ket{x_k}) = \sqrt{1 - \abs{\braket{w_k}{x_k}}^2} = \frac{1}{\sqrt{2}}\sqrt{1 - \abs{\braket{\chi_k}{\tilde{\tau}_k}}} = \frac{1}{\sqrt{2}} D(\ket{\chi_k}, \ket{\tilde{\tau}_k}).
\end{equation}
Altogether then, if \(D(\ket{\chi_k}, \ket{\tilde{\tau}_k}) \leq \epsilon_1+\epsilon_2\) for all \(k\), we have that 
\begin{equation}
  \norm{W^{\otimes \eta} - X^{\otimes \eta}} \leq 2^{3/2} \eta N_{mo} (\epsilon_1+\epsilon_2).
\end{equation}

\subsection{`Naive' Slater determinant state preparation}
\label{app:naive_sd_prep}

In order to compare the cost of our method against prior work on first-quantized Slater determinant state preparation from Ref.~\citenum{Babbush2023-ud}, we briefly describe that prior work and how we count the number of Toffoli gates required to implement it. 
In essence, this prior approach works by preparing a Slater determinant in second quantization and converting the state qubit-by-qubit into a first-quantized representation.
Notably, this can be done in a sequential fashion so that the space complexity (on the quantum computer) is linear in the number of particles rather than the number of plane waves.
This method is capable of preparing an arbitrary Slater determinant, whether it is derived from calculations in a classical Gaussian basis set or from some other source.

Given the classical description for the \(\eta\) orthogonal single-particle basis functions expressed in a plane wave basis, it is straightforward to construct a series of Givens rotations that prepare this determinant in a second-quantized plane wave basis using the techniques described in ~\cite{Kivlichan2018-yz}. 
This requires $N$ layers of Givens rotations acting on at most $\eta$ qubits, after which the first of these qubits is not acted on by any of the subsequent layers.
Hence each layer of Givens rotations can be implemented using only $\eta$ qubits. 
The position of the first qubit acted on by the layer is then written into the first quantized registers controlled on the qubit's value. 
The details of this procedure are described in Appendix G of ~\cite{Babbush2023-ud}, and incur a cost of $3\eta+\lceil\log(\eta+1)\rceil-2$ Toffoli gates per write operation. 

Each second quantized Givens rotation can be implemented using the circuit construction given in Figure 12 of ~\cite{Arrazola2022-tb}, which requires two single-qubit $R_y$ rotations that can be implemented using the reversible adder circuit with $2b$ additional qubits and $2b$ Toffoli gates, where \(b\) is a parameter governing the number of bits of precision in the specification of the rotation angles. 
This implies a total Toffoli cost of 
\begin{align}
\mathtt{Toffoli}(\ket{\vec{\tau}_\text{AS}}) < N\left((3+4b)\eta+\lceil\log(\eta+1)\rceil-2\right)
\label{eq:toff_count_standard_old_way}
\end{align}
Toffoli gates to prepare a Slater determinant, neglecting the cost of antisymmetrization which is subdominant to the other terms.

\section{Computational methods}
\label{app:computational_methods}

For practical calculations we compute restricted Hartree-Fock solutions for the molecular orbitals using Gaussian basis sets with the PySCF python package \cite{Sun2017-sx}.
For all of the molecules we have considered in this work, PySCF's default parameters reliably lead to convergence to a unique solution.
We then construct plane wave matrix product state representations for each molecular orbital using the ITensor Julia library \cite{Fishman2022-bl}. 
The code that implements this construction is available at \url{https://www.github.com/oskar-leimkuhler/orb2mps-fq}.
The XYZ files that specify the molecular geometries are also available as part of this repository.

Starting from optimized molecular orbitals provided by the Hartree-Fock solution, we decompose each orbital over the full set of $N_g$ Cartesian primitive Gaussian functions present in the underlying Gaussian basis set.
This gives a vector of coefficients normalized relative to the overlap matrix of the primitive Gaussians.
We approximately evaluate the projection of each primitive Gaussian basis function onto a plane wave basis by taking the integrals over the whole real line, as in \Cref{eq:tilde_g_x_def}.

In our practical calculations, we justify this approximation by using values of \(L\) that are sufficiently large that the correction is negligible.
For example, in the calculations we have performed on acenes in the cc-pVDZ basis set, the most diffuse Gaussian function is attributed to the 2s subshell of hydrogen, with a diffuseness coefficient of $\gamma=0.122$. 
For the tetracene molecule aligned in the $yz$-plane and lengthwise along the $z$-axis, the terminal hydrogens have $z$ coordinate values of $\pm10.96$ Bohr. 
This means that the most diffuse  Gaussian primitives centred on the terminal hydrogens have a maximal value of $\approx 4.74 \times 10^{-22}$ at the boundary of the computational unit cell, which is significantly smaller than machine epsilon assuming double floating point precision.

We analytically compute the integrals using the expansion in terms of Hermite Gaussian functions of Equation \ref{eq:tilde_g_hermite_gaussian_def}.
We have that
\begin{align}
  \int_{-\infty}^{+\infty}H_m(x)H_n(x)e^{-x^2}dx = \frac{\delta_{mn}}{c_n^2},
\end{align}
and we have an explicit expression for the power $x^l$ in terms of Hermite polynomials:
\begin{align}
  x^l = \frac{l!
  }{2^l}\sum_{m=0}^{\lfloor\tfrac{l}{2}\rfloor}\frac{H_{l-2m}(x)}{m!(l-2m)!}.
\end{align}
We evaluate the $h_n$ normalization factor coefficients by
\begin{align}
  h_n          & \propto c_n\int_{-\infty}^{+\infty}g(\frac{u}{\sqrt{2\gamma}})H_n(u)e^{-\frac{u^2}{2}}du
  \\
               & \propto c_n\int_{-\infty}^{+\infty}u^lH_n(u)e^{-u^2}du
  \\
               & \propto c_n\sum_{m=0}^{\lfloor\tfrac{l}{2}\rfloor}\left[\frac{1}{m!
      (l-2m)!}\int_{-\infty}^{+\infty}
    H_{l-2m}(u)H_n(u)e^{-u^2}du\right]
  \\
               & = \frac{1}{c_n ((l-n)/2)!
    n!}
  \\
  \implies h_n & \propto \frac{2^{n/2}}{((l-n)/2)!\sqrt{n!}}.
\end{align}
Hence computing the right-hand side for each coefficient and normalizing such that $\sum_n|h_n|^2=1$ gives the $h_n$ coefficients for a given $l$.
We then directly compute
\begin{align}
  \int_{-\infty}^{+\infty}g(x-a)\phi_k(x)dx = e^{ika} \frac{2^\frac{1}{4}}{\gamma^\frac{1}{4}L^\frac{1}{2}}\sum_{n=0}^l\frac{i^nh_n}{\sqrt{2^n n!
    }}
  H_n(\frac{k}{\sqrt{2\gamma}})e^{-\frac{k^2}{4\gamma}}
\end{align}
for a Gaussian primitive of angular momentum $l$ and exponent $\gamma$ centred at $a$.

To avoid having to construct the entire plane wave coefficient vector of dimension $N$, we compute separate coefficient vectors of dimension $N^{1/3}$ for the $x,y,$ and $z$ components.
We then symmetrically pad the start and end of each vector with zeros so that the total length is $2^{\lceil\frac{1}{3}\log_2(N)\rceil}$ and reshape to a $2\times 2\times \cdots \times 2$ array for tensor-train factorization to an MPS of $\lceil\frac{1}{3}\log(N)\rceil$ qubit sites.
We exploit the factorization of the Gaussian plane wave overlap integrals to construct the pure product MPS $\ket{g}=\ket{g_x}\otimes\ket{g_y}\otimes\ket{g_z}$ consisting of $3\lceil\frac{1}{3}\log_2(N)\rceil$ qubit sites and bond dimension $\chi=1$ at the $x-y$ and $y-z$ connective bonds.

We construct the projected (un-normalized) MPS representation $\ket{\tau}$ of the molecular orbital by direct summation of the primitive Gaussian MPS representations weighted by the primitive Gaussian coefficients,
\begin{align}
  \ket{\tau} = \sum_{g=1}^{N_g} c_g\ket{g}.
\end{align}
After each direct summation we apply a truncation of singular values smaller than $\varepsilon_\text{sum}$ to control the growth of the bond dimension, and hence a total error of $\mathcal{O}(\varepsilon_\text{sum})$ is incurred at this step.
In practical calculations we have used a direct sum truncation of $\varepsilon_\text{sum}=10^{-9}$ and so we typically see cumulative error in the molecular orbital representations on the order of $10^{-7}$ from the direct sum.
We note that these successive truncation operations are not strictly projections, since the Gaussian primitives are not orthogonal, and so can result in the norm of $\ket{\tau}$ going above unity.

We apply singular value truncation of order $\varepsilon$ to study the effect on the bond dimension as well as the fidelity of the molecular orbital representation.
Provided that \(L\) is sufficiently large so that the integrals over $[-L/2,L/2]$ are identical to the integrals over $(-\infty,\infty)$ up to numerical precision, and that $\varepsilon\gg N_g\varepsilon_\text{sum}$, then both non-negligible sources of error (from the projection onto the plane wave basis and the truncation of singular values) can be thought of as projections, in which case the fidelity of the re-normalized truncated plane wave MPS with the exact molecular orbital $\ket{\chi}$ is given by
\begin{align}
  \left|\frac{\braket{\tau}{\chi}}{\sqrt{\braket{\tau}{\tau}}}\right|^2 = \frac{|\braket{\tau}{\tau}|^2}{|\braket{\tau}{\tau}|} = |\braket{\tau}{\tau}|,
\end{align}
i.e. we simply compute the squared norm of the projected state $\ket{\tau}$.
For practical calculations at small singular value truncations where the direct sum truncation is non-negligible, we approximate the absolute infidelity $\varepsilon_\text{inf}$ as
\begin{align}
  \varepsilon_\text{inf} = \left|1-\|\ket{\tau}\|^2\right|.
\end{align}
The trace distance is then given by 
\begin{equation}
  D(\ket{\tau}, \ket{\chi}) = \sqrt{1 - \abs{\frac{\braket{\tau}{\chi}}{\sqrt{\braket{\tau}{\tau}}}}^2} = \sqrt{1 - \abs{\braket{\tau}{\tau}}}.
\end{equation}

\end{document}